\documentclass[lettersize,journal]{IEEEtran}
\usepackage{algorithmic}
\usepackage{algorithm}
\usepackage{array}
\usepackage[caption=false,font=normalsize,labelfont=sf,textfont=sf]{subfig}
\usepackage{textcomp}
\usepackage{stfloats}
\usepackage{url}
\usepackage{verbatim}
\usepackage{graphicx}
\usepackage{cite}
\usepackage{todonotes}
\usepackage{amsmath,amsfonts,amsthm}
\usepackage{amssymb}
\usepackage{fourier}
\usepackage{mathrsfs}
\usepackage{cases}

\newcommand{\babs}[1]{\left|{#1}\right|}

\newcommand{\leqs}{\leqslant}
\newcommand{\geqs}{\geqslant}
\newcommand{\pare}[1]{\left( {#1} \right)}
\newcommand{\bracket}[1]{\left[ {#1} \right]}
\newtheorem{defi}{Definition}[section]
\newtheorem{lem}{Lemma}[section]
\newtheorem{thm}{Theorem}[section]

\newtheorem{remark}{Remark}[section]

\begin{document}

\title{Two-Point Resolution in Spectral Super-Resolution }

\author{Xiaole He, Ping Liu, Junling Wang~\IEEEmembership{Member,~IEEE}
        % <-this % stops a space
\thanks{ XXX}% <-this % stops a space
\thanks{ XXX} }

\markboth{Journal of \LaTeX\ Class Files,~Vol.~14, No.~8, August~2021}%
{Shell \MakeLowercase{\textit{et al.}}: A Sample Article Using IEEEtran.cls for IEEE Journals}

\IEEEpubid{0000--0000/00\$00.00~\copyright~2021 IEEE}

\maketitle

\begin{abstract}
Two-point super-resolution is an important problem in many signal processing applications. In this paper, we aim to establish a resolution theory for two-point super-resolution from a single snapshot. We consider a complex two-point model with unequal amplitudes and a nontrivial relative phase, and derive super-resolution upper bounds (SRUs) guaranteeing resolvability as well as super-resolution lower bounds (SRLs) below which stable reconstruction is impossible. 
The resulting bounds provide an explicit characterization of how the amplitude ratio and, more importantly, the relative phase affect the resolution limit for both source-number detection and location estimation.
In the in-phase regime, the classical resolution exponents are retained: \((\sigma/m)^{1/2}\) for source-number detection and \((\sigma/m)^{1/3}\) for location estimation. 
In the out-of-phase regimes, the phase term significantly changes the resolution limit: it acts as a direct subtractive term in the near-endpoint regime, and improves the scaling orders in the large-phase regime to \(\sigma/m\) for source-number detection and \((\sigma/m)^{1/2}\) for location estimation.
Extensive numerical experiments across different phase regimes and reconstruction algorithms validate the predicted scaling laws and theoretical resolution boundaries. 
Moreover, comparison with our resolution limit in all phase regimes reveals the optimality of $\ell_0$, ML, and ESPRIT algorithms, and the non-optimality of SVT, MUSIC, and the convex method—a finding that, to the best of our knowledge, has not been reported before. 
Collectively, our results show that the phase of amplitudes is not merely a nuisance in super-resolution, but a key factor that can be exploited to improve stable resolvability.

%demonstrate that phase information is not merely a nuisance parameter in complex-valued super-resolution, but a key factor that can be exploited to improve stable resolvability.

%In particular, in the large-phase regime, the \(\ell_0\), ML, and ESPRIT methods are observed to not only attain optimal scaling orders but also nearly achieve the corresponding SRLs (namely, the optimal resolution), demonstrating their superiority.

\end{abstract}


\section{Introduction}
 
\IEEEPARstart{S}{pectral} super-resolution aims to recover point sources from Fourier samples at a resolution finer than the classical limit imposed by the observation bandwidth \cite{candes2014towards}. It is an inverse problem of great theoretical
and practical interest. A particularly important instance of this general problem is the two-point super-resolution model, in which the number of sources is $n=2$. Despite its apparent simplicity, this model arises pervasively in practical signal processing applications. In electromagnetic propagation, for example, the received signal is often well approximated by two dominant echoes, such as a direct-path echo and a ground- or sea-reflected echo, echoes returned from different interfaces of a layered medium, or two dominant skywave returns associated with different ionospheric propagation paths. Since these components may differ in delay, phase, amplitude, or angle of arrival, their relative parameters can be exploited to infer quantities of interest. Representative examples include low-elevation target height estimation \cite{liu2017height,liu2025signal}, layer-thickness measurement \cite{zhao2018super}, ice-thickness estimation \cite{mangilli2022new}, multilayer coating thickness inversion \cite{yakovlev2015non}, and over-the-horizon radar target tracking \cite{lan2017distributed}.

The two-point model also appears naturally in SAR/ISAR imaging through the layover phenomenon \cite{gini2003layover,martorella20143d}. A SAR/ISAR image can be viewed as a two-dimensional projection of three-dimensional scatterers. Consequently, for complex targets, two distinct scatterers may be projected onto nearly the same location in the image plane. This gives rise to two representative resolution scenarios. In the first scenario, the two scatterers are close but still separable in the image plane \cite{wang2021layover}. In the second scenario, which is typical in InSAR/InISAR, the scatterers remain indistinguishable in a conventional two-dimensional image, but become separable through phase differences induced, for instance, by multi-baseline configurations \cite{lombardini2003reflectivity} or multi-view observations \cite{liu2024multi}. From this perspective, layover resolution can be interpreted as a two-point resolution problem, often along the elevation dimension.

More generally, the two-point model may be regarded as the simplest unresolved-target configuration, where multiple physical targets or scatterers are separated by less than the nominal system resolution in range, angle, or Doppler and therefore occupy the same resolution cell \cite{blair2002unresolved,rogel2021time,kase2018doa}. The received signal is then a superposition of multiple target or scatterer echoes \cite{huang2023joint}. In the simplest case with two dominant contributors, the problem reduces to two-target resolution \cite{lee2015unambiguous,wang2020efficient}. More broadly, two-target resolution provides a fundamental building block for multitarget resolution in multitarget tracking \cite{nandakumaran2008joint,angle2021multiple} and high-resolution imaging \cite{he2017high}.

These examples demonstrate that the two-point model captures a common structural feature shared by many reconstruction problems.  Motivated by these examples, in this paper we study the two-point super-resolution problem with complex amplitudes from a single snapshot.
\IEEEpubidadjcol  %防止首页内容与下方版权IEEEpubid重叠

% In many practical scenarios, however, the number of available snapshots can be extremely limited. In particular, under a short coherent processing interval (CPI), rapid target maneuvering, or nonstationary conditions, the recovery task may even be reduced to a single observation record \cite{hacker2010single}. Meanwhile, several important applications can be locally modeled by two closely spaced point sources. For example, the same transmitted signal may generate multiple propagation paths under atmospheric refraction \cite{feng2019overcoming} or through-wall propagation \cite{ma2018interaction}. In super-resolution imaging and group target detection, echoes are often modeled as the superposition of a few closely spaced scattering centers, whose responses generally differ in both amplitude and phase \cite{tan2024radar,rezaei2022unambiguous}.  This makes the complex-amplitude point-source model particularly relevant in practice. Motivated by these considerations, in this paper we study the super-resolution problem of two one-dimensional point sources with complex amplitudes from a single snapshot.

\subsection{Related Work}

Super-resolution limits have been studied from several perspectives, including analyses tailored to specific reconstruction algorithms. 
Convex optimization (CVX) \cite{bertsekas2009convex,bubeck2015convex,li2023convex} provides one major approach to super-resolution, in which a nonconvex recovery problem can be relaxed into a tractable convex program. In particular, \cite{fernandez2016super} showed that super-resolution is achievable via convex programming when the minimum separation exceeds $1.26/f_c$, while \cite{da2020stable} studied the stable resolution limit of the Beurling--LASSO estimator for spike deconvolution with total variation regularization. Furthermore, \cite{duval2020characterization} established necessary and sufficient conditions under which Beurling–LASSO achieves stable super-resolution recovery of positive sources under the Laplacian kernel and several Gaussian sampling schemes.

Maximum likelihood (ML) methods \cite{pan2002maximum} form a class of high-resolution parametric estimators that recover the source parameters by optimizing the likelihood function. 
The resolution capability of ML-type estimators has been explicitly studied through the probability of resolution. 
In particular, \cite{mestre2020resolution} characterized the resolution probability of conditional and unconditional ML DoA estimators in the threshold region, where the SNR and/or the number of snapshots are limited. Along this line, \cite{schenck2020probability} derived an asymptotic characterization of the probability of resolution for partially relaxed deterministic maximum likelihood. 

Subspace methods \cite{katayama2005subspace,lee2012subspace,liao2016iterative} provide yet another algorithmic perspective, exploiting the orthogonality between the signal and noise subspaces to estimate spectral locations or source parameters. Representative examples include MUSIC (Multiple Signal Classification) \cite{agarwal2016multiple}, and ESPRIT (Estimation of Signal Parameters via Rotational Invariance Techniques) \cite{roy2002esprit}, both of which are particularly effective in multi-snapshot spectral estimation. For single-snapshot data, spatial smoothing may be used to induce a multi-snapshot structure, thereby allowing subspace methods to be applied to the model in \eqref{equ:modelsetting0}. In this direction, \cite{li2022stability} analyzed the stability and super-resolution performance of MUSIC and ESPRIT in a multi-snapshot setting, with explicit dependence on the noise level, the number of snapshots, and the super-resolution factor (SRF). More recently, \cite{ding2024esprit} proved that, under suitable assumptions on the bias and high-noise regime, the localization error of ESPRIT can attain the optimal scaling $\widetilde{O}\!\left(n^{-3/2}\right)$ with respect to the cutoff frequency $n$. 

In recent years, motivated by major advances in super-resolution microscopy \cite{hell1994breaking,STED,hess2006ultra,PALM,STORM} and by the rapid development of super-resolution algorithms in applied mathematics \cite{duval2015exact,poon2019,tang2014near,morgenshtern2020super,denoyelle2017support}, the inherent super-resolving capacity of the imaging problem has become increasingly popular, and the one-dimensional case was well studied \cite{liu2024mathematical}. For sparse recovery from band-limited Fourier measurements, the minimax reconstruction error was shown to scale polynomially with the SRF, indicating that the deterioration in resolution is fundamentally caused by the measurement model rather than by any specific reconstruction algorithm \cite{demanet2015recoverability}. This viewpoint was further sharpened by relating the resolution limit to the smallest singular value of partial Fourier or Vandermonde matrices with closely spaced nodes. In particular, when the support contains clustered sources, the smallest singular value decays according to an SRF-dependent power whose exponent is determined by the local cluster cardinality \cite{batenkov2020conditioning, li2021stable}. Extending this perspective to the off-the-grid setting, sharp minimax recovery rates were established for near-colliding point sources, showing that clustered nodes and their amplitudes obey fundamentally worse scaling laws than isolated ones under the same band-limited observation model \cite{batenkov2021super}. More recently, the Cramér–Rao lower bound has been employed to characterize the transition between well- and ill-conditioned regimes through the Fisher information matrix, thereby providing a statistical interpretation of a Rayleigh-type resolution threshold \cite{hockmann2024analysis}. These results suggest that the ultimate resolution limit is largely governed by the bandwidth-limited observation operator and the local geometry of the support, whereas specific algorithms differ mainly in how closely they approach this intrinsic limit. 

Overall, the above results mainly characterize the resolution limit through the scaling of reconstruction error with the SRF, rather than through an explicit separability statement. 
To obtain a more direct characterization, Liu introduced the concept of the computational resolution limit (CRL) in \cite{liu2024mathematical,liu2021theory,liu2022rslpositive}, which provides explicit bounds on the minimum separation required for stable super-resolution in terms of the noise level and the minimum source amplitude. 
However, existing CRL analyses based on equal-amplitude or in-phase models do not fully exploit the additional resolving information carried by amplitude and phase variations, and may therefore lead to overly conservative resolution bounds. 
This limitation is particularly relevant in radar sensing and imaging scenarios, where phase differences can enhance the distinguishability of scatterers or targets.
Motivated by this observation, we extend the CRL framework to the complex-amplitude setting and derive bounds that explicitly incorporate the effects of amplitude ratio and relative phase.

% However, these bounds do not account for the effect of phase difference. Moreover, despite its direct relevance to the fundamental resolution bottleneck in single-snapshot or short-CPI imaging, a systematic characterization of the minimum resolvable two-point separation remains limited for this setting. Although two-point resolvability has been extensively studied in optical imaging and statistical detection theory \cite{helstrom2003detection,shahram2004imaging,shahram2005resolvability,reddy2013two,reddy2018apodization,massaro2024two}, those results are typically derived under substantially different assumptions, such as incoherent intensity-only measurements or long-exposure observations, and are therefore not directly applicable to the single-snapshot model in \eqref{equ:modelsetting0}.

\subsection{Our Contributions}

This paper studies the two-point super-resolution limit under the single-snapshot model, where the two sources may have unequal amplitudes and a nontrivial relative phase. We focus on two fundamental objectives: detecting the number of sources and stably estimating their locations once the source number has been resolved. The main contributions are summarized as follows:
\begin{itemize}
    \item \textbf{Explicit super-resolution bounds for complex two-point sources.}
    We establish a  CRL characterization for the complex two-point model $\mu = m e^{i\theta_1}\delta_{y_1} + \beta m e^{i\theta_2}\delta_{y_2}$ under deterministic bounded noise and single-snapshot band-limited Fourier measurements. 
    For the two fundamental tasks of source-number detection and location estimation, we derive super-resolution upper bounds (SRUs) that guarantee stable resolvability, together with super-resolution lower bounds (SRLs) below which stable resolution is impossible in general. 
    The resulting bounds explicitly quantify how the amplitude ratio $\beta$ and, more importantly, the relative phase $\theta=\theta_1-\theta_2$ affect the minimum resolvable separation.

    \item \textbf{A phase-dependent regime decomposition with improved scaling laws.}
    We show that the SRUs and SRLs admit a phase-dependent decomposition into three qualitatively distinct regimes: the in-phase regime, the near-endpoint phase regime, and the large-phase regime. 
    In the in-phase regime, the SRLs recover the classical exponent laws, namely \((\sigma/m)^{1/2}\) for source-number detection and \((\sigma/m)^{1/3}\) for stable location estimation. 
    In the near-endpoint phase regime, the phase term reduces the corresponding resolution bounds while preserving the same scaling exponents. 
    By contrast, in the large-phase regime, the SRLs become strictly sharper, improving to the orders \(\sigma/m\) and \((\sigma/m)^{1/2}\) for source-number detection and location estimation, respectively. 
    Moreover, for source-number detection, we identify a critical near-\(\pi\) phenomenon in the equal-amplitude case: the scaling law improves from the generic order \((\sigma/m)^{1/2}\) to the sharper order \(\sigma/m\). 
    These results demonstrate that the relative phase can fundamentally alter two-point super-resolution, either by lowering the resolution threshold within the same exponent regime or by changing the resolution scaling law itself.

    \item \textbf{Systematic numerical validation and algorithm-selection guidance.}
    We perform extensive Monte Carlo simulations across the three phase regimes and across several representative algorithms. 
    For source-number detection, we compare the $\ell_0$ and singular-value-thresholding (SVT) methods; for location estimation, we compare MUSIC, ESPRIT, ML, and CVX. 
    The measured phase-transition slopes and empirical resolution boundaries agree with the proposed SRLs, validating the predicted phase-dependent scaling laws. In particular, the $\ell_0$, ML, and ESPRIT algorithms achieve the optimal scaling law in all phase regimes, while the SVT, MUSIC, and the convex methods deviate from optimal resolution order in certain phase regimes. Moreover, the SRLs also serve as quantitative benchmarks for algorithm selection: among the tested methods, $\ell_0$ is closest to the SRLs for source-number detection, whereas ESPRIT and ML exhibit the best performance for location estimation. 
    In the large-phase regime, these best-performing methods attain the improved optimal scaling orders and nearly approach the corresponding SRLs, namely, the optimal resolution.
\end{itemize}

The remainder of this paper is organized as follows. Section~\ref{section:model_setting_with_CRL} introduces the problem formulation and basic definitions. Sections~\ref{section:zero_theta}--\ref{section:large_theta} present super-resolution upper and lower bounds for the source-number detection and location estimation problem under three different relative phase regimes, and provide numerical results to validate the proposed theory. Section~\ref{section:conclusion} concludes the paper. The proofs of the derived CRL bounds are deferred to the Appendix.

\section{Model Setting and Computational Resolution limits}\label{section:model_setting_with_CRL}

\subsection{Model Setting}\label{subsection:model_setting}
We first introduce the two-point super-resolution model considered throughout this paper. Let
\(
\mu=\sum_{j=1}^2 a_j\delta_{y_j}
\)
be a discrete measure, where $y_j\in\mathbb{R}$ denotes the location of the $j$th point source and $a_j\in\mathbb{C}$ denotes its complex amplitude. We observe the noisy Fourier data
\begin{align}
    \mathbf{Y}(\omega)
    =\mathscr{F}[\mu](\omega)+\mathbf{W}(\omega)
    =\sum_{j=1}^2 a_j e^{i y_j\omega}+\mathbf{W}(\omega),\quad |\omega|\leq \Omega,
    \label{equ:modelsetting0}
\end{align}
where $\mathscr{F}[\cdot]$ denotes the Fourier transform, $\Omega$ is the cut-off frequency, and $\mathbf{W}(\cdot)$ denotes the noise. Throughout this paper, we adopt a deterministic bounded-noise model:
\begin{align}
    |\mathbf{W}(\omega)|<\sigma,\quad |\omega|\leq \Omega,
    \label{equ:modelsetting_of_sigma}
\end{align}
where $\sigma$ is the noise level. The goal of spectral super-resolution is to recover the discrete measure $\mu$ from the band-limited noisy data $\mathbf{Y}$. 

% The basic measurement model above will be supplemented by two aspects that are important to our analysis: the spatial configuration of the two point sources and the relative phase between their complex amplitudes.

% \subsubsection{Point Source Configuration}\label{subsubsec:cluster}

Since we are interested in the super-resolution regime, we follow \cite{donoho1992superresolution} and define the Rayleigh length as
\begin{align}\label{equ:rayleighlength1}
    d_{\mathrm{RL}}=\frac{\pi}{\Omega}.
\end{align}
Super-resolution then refers to resolving point sources at a scale finer than $d_{\mathrm{RL}}$. Therefore, we assume that the two point sources lie within one Rayleigh-length window centered at the origin, that is,
\begin{equation}\label{eq:cluster_assump}
    y_j \in B_{\frac{\pi}{2\Omega}}(0),\quad j=1,2,
\end{equation}
where
\begin{equation}\label{eq:Bdelta}
    B_{\delta}(x) := \left\{ y\in \mathbb{R}: |y - x| < \delta \right\},
\end{equation}
represents a one-dimensional open neighborhood. 
%This assumption is standard in off-the-grid super-resolution with band-limited Fourier data; see, e.g., \cite{batenkov2021super,liu2021theory}. 

%We assume that the two point sources lie within one Rayleigh-length window centered at the origin, namely,
%\begin{equation}\label{eq:cluster_assump}
%    y_j \in B_{\frac{\pi}{2\Omega}}(0),\quad j=1,2.
%\end{equation}

%Moreover, by translation invariance, results established for configurations in $B_{\frac{\pi}{2\Omega}}(0)$ extend directly to configurations in $B_{\frac{\pi}{2\Omega}}(x)$ for any $x\in\mathbb{R}$ after a shift of the coordinate origin.

%\subsubsection{Relative Phase}\label{subsubsec:phase}

To characterize  the influence of the phase, we parameterize the complex amplitudes as
\(
a_j = |a_j|e^{i\theta_j}
\)
with $\theta_j\in(-\pi,\pi]$, $j=1,2$, and define the relative phase and the effective relative phase, respectively, by
\begin{equation}\label{eq:theta_def}
    \theta := \theta_1-\theta_2, \quad \babs{\theta}_{\min} := \min\{\babs \theta,\pi-\babs \theta\},
\end{equation}
which play important roles in our resolution estimation.

% Since the measurements depend on the phases only through the factors $e^{i\theta_j}$, both phases are defined modulo $2\pi$. Hence, the relative phase $\theta$ is also defined modulo $2\pi$, and we may restrict it, without loss of generality, to its principal value in $(-\pi,\pi)$. We further define
% \begin{align}
%     |\theta|_{\min}=\min\{|\theta|,\,\pi-|\theta|\},
% \end{align}

\subsection{Definitions of Computational Resolution Limits}\label{subsec:definitions}

Having specified the measurement model, we now recall the computational resolution limit (CRL) introduced in \cite{liu2024mathematical}, which provides a quantitative characterization of resolution limits. 
%The corresponding CRL definitions for the two tasks considered in this paper are given below.

\subsubsection{CRL for source-number detection}\label{subsubsec:sigma_adm}
The CRL framework for source-number detection is based on the concept of $\sigma$-admissible measure given by: 

\begin{defi}\label{def:sigma_admissible}
Given the measurement $\mathbf{Y}$ generated by $\mu=\sum_{j=1}^n a_j\delta_{y_j}$ in \eqref{equ:modelsetting0}, $\widehat{\mu}=\sum_{j=1}^{d}\widehat{a}_{j}\,\delta_{\widehat{y}_{j}}$ is said to be a \textup{$\sigma$-admissible} discrete measure of\, $\mathbf{Y}$ if
\begin{equation}\label{equ:sigma_adm}
    \bigl|\mathscr{F}[\widehat{\mu}](\omega)-\mathbf{Y}(\omega)\bigr| < \sigma,
    \quad  \forall \,|\omega|\leqs \Omega.
\end{equation}
If further $\widehat{a}_j>0$, $j=1,\cdots,d$, then $\widehat{\mu}$ is said to be a positive \textup{$\sigma$-admissible} discrete measure of\, $\mathbf{Y}$.
\end{defi}

% \begin{remark}[A convenient condition]
% \label{rem:sufficient_2sigma}
% Let $\mu$ be the underlying source measure in \eqref{equ:modelsetting0}. Suppose that a candidate measure $\widehat{\mu}$ satisfies
% \[
% \bigl|\mathscr{F}[\widehat{\mu}](\omega)-\mathscr{F}[\mu](\omega)\bigr| < 2\sigma,
% \quad \forall\, |\omega|\leqs \Omega.
% \]
% Then there exists a measurement $\mathbf{Y}$ consistent with \eqref{equ:modelsetting0} such that $\widehat{\mu}$ is $\sigma$-admissible for $\mathbf{Y}$. 
% \end{remark}

% Notably, the admissible set for the reconstruction process comprises discrete measures whose Fourier data are sufficiently close to $\mathbf{Y}$ in \eqref{equ:modelsetting0}. 
% In general, every such admissible measure is possibly the ground truth without any additional prior information. Consequently, the concept of $\sigma$-admissible discrete measures are introduced.

Note that the set of $\sigma$-admissible measures of $\mathbf{Y}$ characterizes all possible solutions to our super-resolution problem with the given measurement $\mathbf{Y}$. Detecting the source number $n$ is possible only if all of the admissible measures have at least $n$ supports; otherwise, it is impossible to detect the correct source number without additional a priori information. Thus, following definitions similar to those in \cite{liu2021mathematicalhighd,liu2021theorylse,liu2021mathematicaloned, liu2023improved}, we define the computational resolution limit for the source-number detection problem as follows.

\begin{defi}
The computational resolution limit to the source-number detection in the super-resolution problem is defined as the smallest nonnegative number $\mathrm{CRL}_{\mathrm{num}}$ such that for all two-sparse measures $\sum_{j=1}^2 a_j \delta_{y_j}, a_j \in \mathbb{C}, {y}_j \in B_{\frac{ \pi}{2 \Omega}}(0)$ and the associated measurement $\mathbf{Y}$ in \eqref{equ:modelsetting0}, if
$$
\babs{y_1-y_2} \geqs \mathrm{CRL}_{\mathrm{num}},
$$
then there does not exist any $\sigma$-admissible measure of $\mathbf{Y}$ with fewer than two support points.  
\end{defi}

\subsubsection{CRL for location estimation}\label{subsubsec:delta_nbhd}
To formalize stable localization, we use the notion of a $\delta$-neighborhood from \cite{liu2024mathematical}. 

\begin{defi}\label{def:delta_neighborhood}
Let \(\mu = \sum_{j=1}^n a_j \delta_{y_j}\) with \(y_j \in \mathbb{R}\), and let \(\delta>0\) be such that the sets \(B_\delta(y_j)\), \(1\leqs j\leqs n\), are pairwise disjoint. We say that
\(
\widehat{\mu} = \sum_{j=1}^n \widehat{a}_j \delta_{\widehat{y}_j}
\)
lies in the \(\delta\)-neighborhood of \(\mu\) if each \(\widehat{y}_j\) belongs to exactly one of the sets \(B_\delta(y_j)\), \(1\leqs j\leqs n\).
\end{defi}

According to the above definition, a measure in a $\delta$-neighborhood preserves the inner structure of the true set of sources. 
For any stable location estimation algorithm, the output should be a measure in some $\delta$-neighborhood; otherwise, it is impossible to distinguish which is the reconstructed location of some $y_j$'s. 
Therefore, the CRL for stable location estimation is defined as: 
%For ease of exposition, we only consider measures supported in 
%$B_{\frac{\pi}{2\Omega}}(0)$.

\begin{defi}
The computational resolution limit to the stable location estimation in the super-resolution problem is defined as the smallest nonnegative number $\mathrm{CRL}_{\mathrm{supp}}$ such that for all two-sparse measure 
$\sum_{j=1}^{2} a_j\delta_{y_j}$, $a_j \in \mathbb{C}$,
$y_j\in B_{\frac{\pi}{2\Omega}}(0)$ and the associated measurement $\mathbf{Y}$ in \eqref{equ:modelsetting0}, if
\begin{align}
    \babs{y_1-y_2} \geqs \mathrm{CRL}_{\mathrm{supp}},
\end{align}
then there exists $\delta>0$ such that any $\sigma$-admissible measure for $\mathbf{Y}$ with two supports in 
$B_{\frac{\pi}{2\Omega}}(0)$ is within a $\delta$-neighborhood of $\mu$. 
\end{defi}

In the following sections, we estimate the CRLs in different phase regimes. In particular, we derive super-resolution upper bounds (SRUs) that guarantee stable resolvability, together with super-resolution lower bounds (SRLs) below which stable resolution is impossible in general. Therefore, the CRL is bounded by SRU and SRL. Numerical simulations are then presented to validate the theoretical predictions. %Specifically, the SRL characterizes the theoretically optimal resolution achievable by an ideal reconstruction algorithm. 

\section{In-Phase (Positive) Sources $\left( \theta=0 \right)$} \label{section:zero_theta}

We first consider the in-phase case, in which the two source amplitudes are aligned ($\theta=0$) or positive, i.e., $\mu=m \delta_{y_1}+\beta m \delta_{y_2}$, where $m>0$ represents the signal intensity and $\beta\geqs 1$ represents the amplitude ratio.

\subsection{Theoretical Bounds} \label{subsection:boundary_zero_theta}

%This subsection establishes the resolution limits for the in-phase case, in which the two complex amplitudes reduce to positive weight. 
We derive sharp estimates of the SRUs and SRLs for both source-number detection and location estimation. The proofs of the corresponding theorems are deferred to Appendix~\ref{proof:zero_theta}.

\begin{thm}\label{thm:numberresolutionpositive}
Let $\mathbf{Y}$ be generated by the positive measure
$\mu=m \delta_{y_1}+\beta m \delta_{y_2}$ with $y_1, y_2\in B_{\frac{\pi}{2\Omega}}(0)$,
 $\beta\geqs 1$, and $m>0$.
For the case when $\frac{\sigma}{m}<\frac{\beta}{2(\beta+1)}$, if
\begin{align}
\label{equ:num_positive_ub_cond}
\displaystyle
|y_1-y_2|\geqs\frac{2}{\Omega}\arcsin\!\left(\left(\frac{2(\beta+1)\sigma}{\beta m}\right)^{\frac{1}{2}}\right),
\end{align}
then no $\sigma$-admissible measure for $\mathbf{Y}$ can be supported on fewer than two points.
Moreover, if
\begin{align}
\label{equ:num_positive_lb_cond}
\babs{y_1-y_2} < \frac{2}{\Omega}\pare{ \frac{\beta+1}{\beta}\frac{\sigma}{m} }^{\frac{1}{2}},
\end{align}
then there exists a positive $\sigma$-admissible measure $\widehat \mu$ for $\mathbf{Y}$ supported on a single point.
\end{thm}

\begin{remark}
    The condition on noise-to-signal ratio $\frac{\sigma}{m}<\frac{\beta}{2(\beta+1)}$ in Theorem \ref{thm:numberresolutionpositive} and similar conditions in subsequent results are necessary. Otherwise, some source information will be completely buried in the noise. For example, when $\frac{\sigma}{m}\geq \frac{\beta}{2(\beta+1)}$, for any well-separated $y_1, y_2$, the measurements of $m \delta_{y_1} + \beta m \delta_{y_2}$ can be approximated by the single point source $\beta m \delta_{y_2}$ to noise level.  
\end{remark}

Theorem \ref{thm:numberresolutionpositive} presents a sharp estimate on $\mathrm{CRL}_{\mathrm{num}}$ for superresolving two positive sources with imbalanced amplitudes. Next, we introduce the estimate of $\mathrm{CRL}_{\mathrm{supp}}$.

\begin{thm}
\label{thm:supportrecoveryupperboundpositive}
Let $\mathbf{Y}$ be generated by a positive measure $\mu=m \delta_{y_1}+\beta m \delta_{y_2}$, with $y_1, y_2\in B_{\frac{\pi}{2\Omega}}(0)$, $\beta \geqs 1$, and $m>0$. Let $d:=|y_1-y_2|$. Assume that
\begin{align}
\babs{y_1-y_2} \geqs \frac{3}{\Omega}\arcsin\left(2\left( \frac{\sigma}{m } \right)^\frac{1}{3} \right).
\label{equ:separationforlocationrecovery-thm}
\end{align}
If $\widehat{\mu} = \widehat{a}_1 \delta_{\widehat{y}_1} + \widehat{a}_2 \delta_{\widehat{y}_2}$ supported on $B_{\frac{\pi}{2\Omega}}(0)$ is a $\sigma$-admissible measure for $\mathbf{Y}$, then $\widehat{\mu}$ lies within the $\frac{d}{2}$-neighborhood of $\mu$. Moreover, if 
\begin{align}
    \babs{y_1-y_2}<\frac{2.25}{\Omega}\pare{\frac{\sigma}{m}}^{\frac{1}{3}},
    \label{equ:condition_locs_positive}
\end{align}
then there exists a positive $\sigma$-admissible measure $\widehat \mu$ for $\mathbf{Y}$ that does not lie within the $\frac{d}{2}$-neighborhood of $\mu$.
\end{thm}

The preceding results show that, in the in-phase case, the influence of the amplitude ratio $\beta$ on the resolution limit is relatively mild. It does not change the resolution order and only slightly affects the constants in these bounds. We next present numerical experiments to validate these theoretical predictions.

%In particular, $\beta$ affects only the constants in these bounds and does not change the corresponding resolution exponents. 

\begin{remark}
A key difference between the results in this section and those in the next (which addresses small phase differences) is that the positivity of the sources can be used as prior information. However, the results show that this prior information does not enhance resolution.
\end{remark}

\subsection{Numerical Experiments on Source-Number Detection}\label{experiment:zero_theta_num}
\subsubsection{Algorithms Considered}
For source-number detection, we consider two algorithms: the $\ell_0$ method and the singular-value-thresholding (SVT) method. We briefly review these two methods below.

The $\ell_0$ method provides a simple numerical verification procedure for the source-number detection problem. Since validating the SRL for source-number detection amounts to checking whether the measurement $\mathbf{Y}$ admits a one-point $\sigma$-admissible explanation, the key question is whether the measurements generated by two point sources can also be explained by a single-point measure at noise level $\sigma$. For a candidate measure $\widehat{\mu}=\widehat{a}\delta_{\widehat{y}}$ and sampled frequencies $\{\omega_m\}_{m=1}^M\subset[-\Omega,\Omega]$, the feasibility to explain $\mathbf{Y}$ requires
\begin{align*}
\babs{\widehat{a}e^{i\widehat{y}\omega_m}-\mathbf{Y}(\omega_m)}<\sigma,
    \qquad m=1,\ldots,M,
\end{align*}
or equivalently,
\begin{align*}
    \babs{\widehat{a}-\mathbf{Y}(\omega_m)e^{-i\widehat{y}\omega_m}}<\sigma,
    \qquad m=1,\ldots,M.
\end{align*}
For a fixed $\widehat{y}$, define $c_m:=\mathbf{Y}(\omega_m)e^{-i\widehat{y}\omega_m}$, $m=1,\ldots,M$.
Then feasibility is equivalent to the existence of some $\widehat{a}\in\mathbb{C}$ such that $\babs{\widehat{a}-c_m}<\sigma,
    m=1,\ldots,M$. Thus, for each $m$, one obtains a disk $D(c_m,\sigma):=\{z\in\mathbb{C}:|z-c_m|<\sigma\}$. A feasible one-point model exists if and only if the intersection $\cap_{m=1}^M D(c_m,\sigma)$ is nonempty. Since $\mathbb{C}\cong\mathbb{R}^2$ and each disk is convex, Helly's theorem implies that $\bigcap_{m=1}^{M} D(c_m,\sigma)\neq\varnothing$
if and only if
\begin{align*}
    D(c_i,\sigma)\cap D(c_j,\sigma)\cap D(c_k,\sigma)\neq\varnothing,
    \qquad \forall\, 1\le i<j<k\le M.
\end{align*}
This yields a simple numerical certificate for the lower bound: for each candidate $\widehat{y}$, we check whether every triple of disks has a nonempty intersection. Whenever this condition is satisfied, there exists a one-point $\sigma$-admissible measure, and hence the corresponding two-source configuration cannot be declared reliably distinguishable at noise level $\sigma$.

The SVT method, proposed in \cite{liu2024mathematical}, is a simple source-number detection algorithm based on singular values, and its effectiveness for the two-point number detection problem has been established rigorously therein. Specifically, it first constructs a \(2\times 2\) Hankel matrix from the three samples $\mathbf{Y}(-\Omega)$, $\mathbf{Y}(0)$, and $\mathbf{Y}(\Omega)$. It then computes the two singular values of this matrix in decreasing order. The decision is made by thresholding the smaller singular value: if $\widehat{\sigma}_2\geqs 2\sigma$, the algorithm declares that two sources are present; otherwise, it declares a single source.

\subsubsection{Comparison With Theoretical Predictions} \label{experiment:zero_theta}
Based on the two algorithms considered for source-number detection, we conduct numerical experiments to determine their empirical resolution boundaries and compare them with the theoretical SRUs and SRLs derived in Section~\ref{subsection:boundary_zero_theta}. This comparison provides a numerical validation of the proposed bounds.

Throughout the experiments, we set $\Omega=1$ and use ten Fourier measurements, so that $\mathbf{Y}\in\mathbb{C}^{10\times 1}$, which is sufficient for the two-source model considered here. 
For each parameter setting $(d,\sigma)$, with $d$ denoting the minimum separation between the two point sources, we conduct 10,000 Monte Carlo trials and declare the setting successful only when the algorithm succeeds in all trials. 
Successful and unsuccessful settings are marked in green and red, respectively. 
For source-number detection, success means correctly identifying the presence of two distinct sources. 
This all-success criterion provides a stringent assessment of algorithmic stability.

Here, we introduce the super-resolution factor (SRF) as
\[
\mathrm{SRF}: = \frac{d_{\mathrm{RL}}}{d}=\frac{\pi}{d\Omega},
\]
where $d_{\mathrm{RL}}$ is the Rayleigh length in (\ref{equ:rayleighlength1}). For each algorithm, the numerical validation is carried out in two complementary ways. First, we examine the slope of the success--failure boundary in the $\log(\mathrm{SRF})$--$\log(\sigma)$ plane under different values of $\beta$ in order to verify the predicted scaling order of the resolution bounds. Given the sharpness of our estimated SRUs and SRLs, we adopt the lower bound (SRLs) as a reference throughout all numerical discussions in this paper. Second, for a fixed $\beta$, we plot the empirical success--failure phase transition diagram to compare the observed transition boundary against the theoretical lower bound.

\begin{figure}[bp]
\centering
\subfloat[$\ell_0$ \centering]{\includegraphics[height=3.45cm]{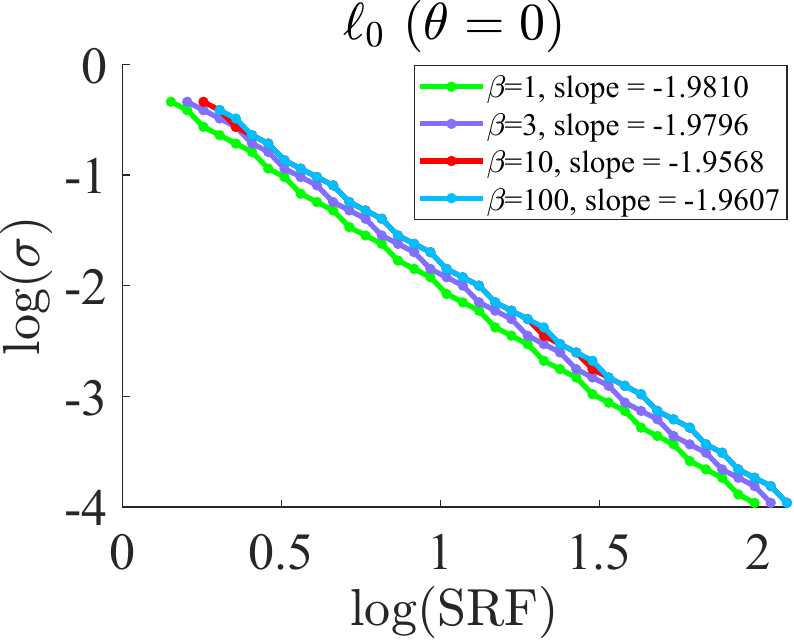}\label{fig:slope_num_LO_theta0}}   \hfil
\subfloat[$\ell_0$ \centering]{\includegraphics[height=3.45cm]{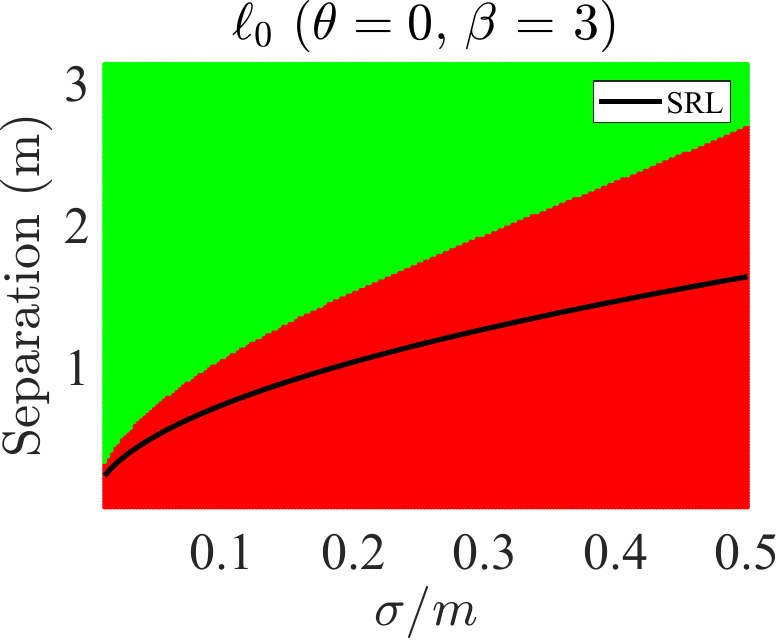}\label{fig:SRL_num_LO_theta0}}  \\

\subfloat[SVT \centering]{\includegraphics[height=3.45cm]{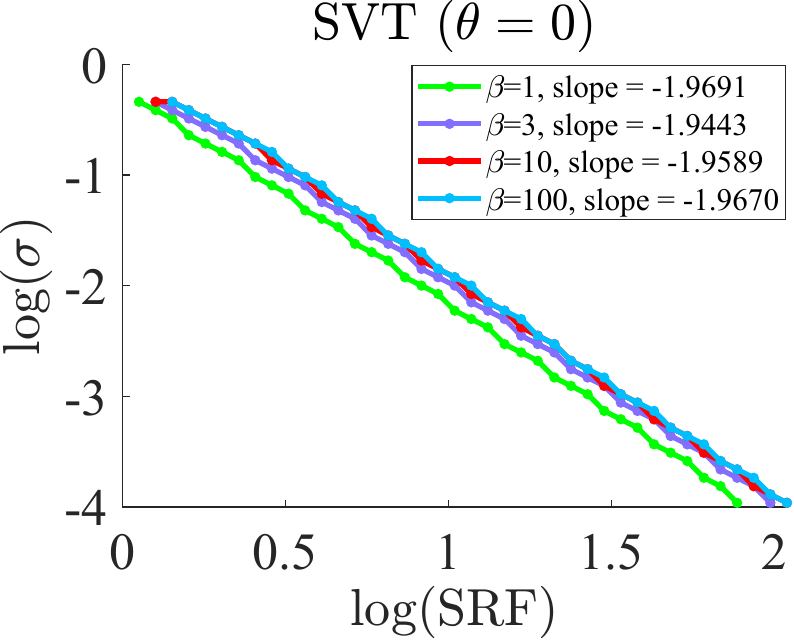}\label{fig:slope_num_MUSIC_theta0}} \hfil 
\subfloat[SVT \centering]{\includegraphics[height=3.45cm]{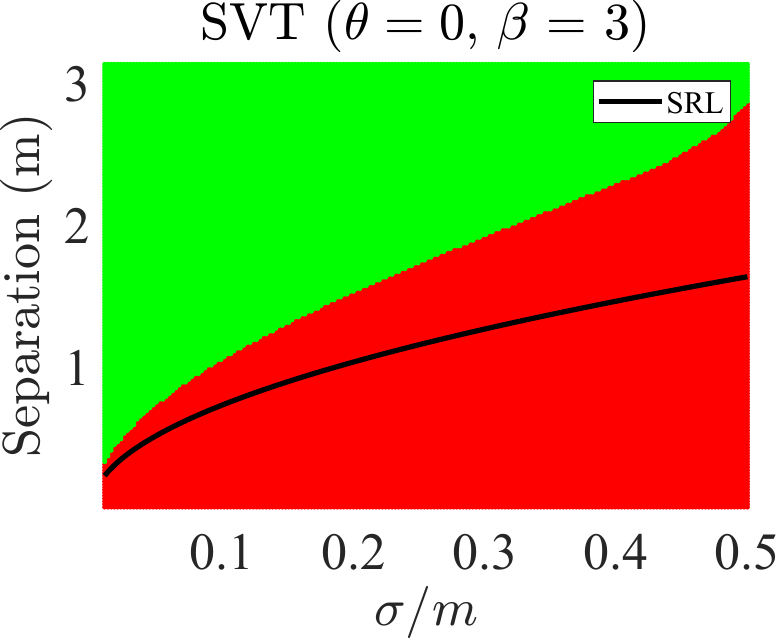}\label{fig:SRL_num_MUSIC_theta0}}  
\caption{ Stable performance of each algorithm for source-number detection under the in-phase regime.  } 
\label{fig:in_phase_num_case}
\end{figure}
For the source-number detection problem, the SRL in \eqref{equ:num_positive_lb_cond} implies
\begin{align*}
\log(\sigma)> -2\,\log(\mathrm{SRF})+\log(m)
-\log\pare{\frac{\beta+1}{\beta}}
+2\,\log\pare{\frac{\pi}{2}},
\end{align*}
which yields a linear lower-bound boundary with slope $-2$ in the 
$\log(\mathrm{SRF})$--$\log(\sigma)$ plane for fixed $m$. 
Thus, an algorithm whose empirical phase-transition boundary has slope close to $-2$ attains the optimal resolution order for source-number detection in this regime. 
As shown in Figs.~\ref{fig:slope_num_LO_theta0} and \ref{fig:slope_num_MUSIC_theta0}, the empirical phase-transition curves of both $\ell_0$ and SVT are nearly linear for $\beta=1,3,10,100$, with fitted slopes close to $-2$. 
This indicates that both algorithms achieve the optimal resolution order predicted by the SRL. 
Moreover, Figs.~\ref{fig:SRL_num_LO_theta0} and \ref{fig:SRL_num_MUSIC_theta0} provide a more refined comparison for $\beta=3$. The empirical resolution boundaries of both algorithms stay above, but remain close to, the theoretical SRL curve in \eqref{equ:num_positive_lb_cond}. 
Therefore, in the in-phase source-number detection setting, both $\ell_0$ and SVT not only exhibit the optimal resolution order, but also approach the best achievable resolution.

% \begin{figure*}[bp]
% \centering
% \subfloat[MUSIC \centering]{\includegraphics[height=3.45cm]{fig/slope_locs_MUSIC_theta0.pdf}\label{fig:slope_locs_MUSIC_theta0}}  \hfil
% \subfloat[ESPRIT \centering]{\includegraphics[height=3.45cm]{fig/slope_locs_ESPRIT_theta0.pdf}\label{fig:slope_locs_ESPRIT_theta0}}  \hfil
% \subfloat[ML \centering]{\includegraphics[height=3.45cm]{fig/slope_locs_ML_theta0.pdf}\label{fig:slope_locs_ML_theta0}}  \hfil
% \subfloat[CVX \centering]{\includegraphics[height=3.45cm]{fig/slope_locs_CVX_theta0.pdf}\label{fig:slope_locs_CVX_theta0}}  \\

% \subfloat[MUSIC \centering]{\includegraphics[height=3.45cm]{fig/locs_SRL_beta3_theta0_MUSIC.pdf}\label{fig:SRL_locs_MUSIC_theta0}}   \hfil
% \subfloat[ESPRIT \centering]{\includegraphics[height=3.45cm]{fig/locs_SRL_beta3_theta0_ESPRIT.pdf}\label{fig:SRL_locs_ESPRIT_theta0}}   \hfil
% \subfloat[ML \centering]{\includegraphics[height=3.45cm]{fig/locs_SRL_beta3_theta0_ML.pdf}\label{fig:SRL_locs_ML_theta0}}\hfil
% \subfloat[CVX \centering]{\includegraphics[height=3.45cm]{fig/locs_beta3_theta0_CVX.pdf}\label{fig:SRL_locs_CVX_theta0}} 

% \caption{ Stable performance of each algorithm for the location estimation problem under the in-phase amplitude case.  } 
% \label{fig:in_phase_locs_case}
% \end{figure*}

\begin{figure*}[bp]
\centering
\subfloat[MUSIC \centering]{\includegraphics[height=3.45cm]{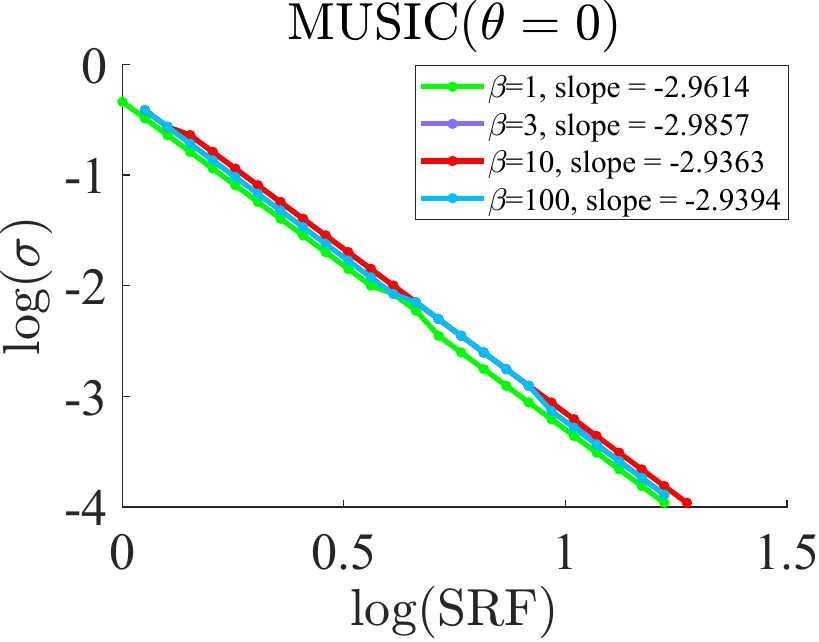}\label{fig:slope_locs_MUSIC_theta0}}  \hfil
\subfloat[ESPRIT \centering]{\includegraphics[height=3.45cm]{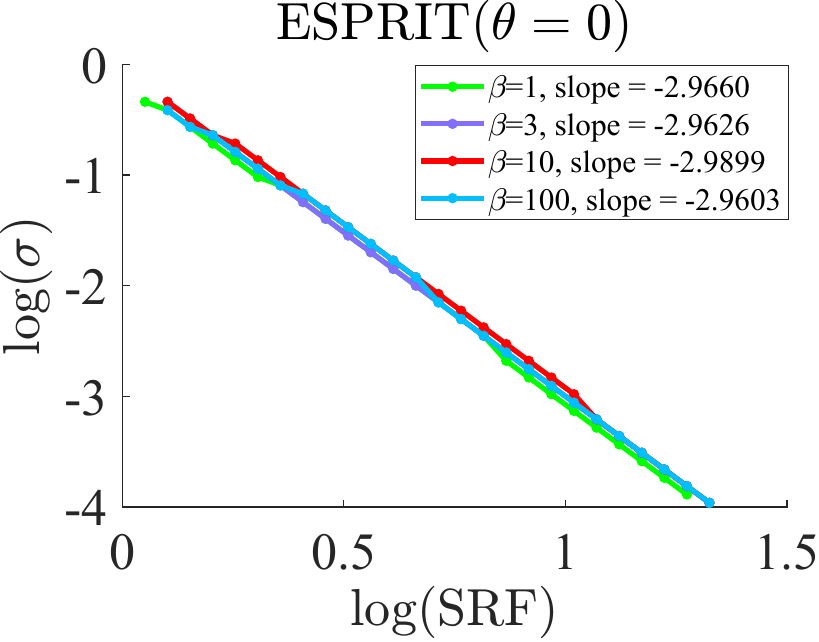}\label{fig:slope_locs_ESPRIT_theta0}}  \hfil
\subfloat[ML \centering]{\includegraphics[height=3.45cm]{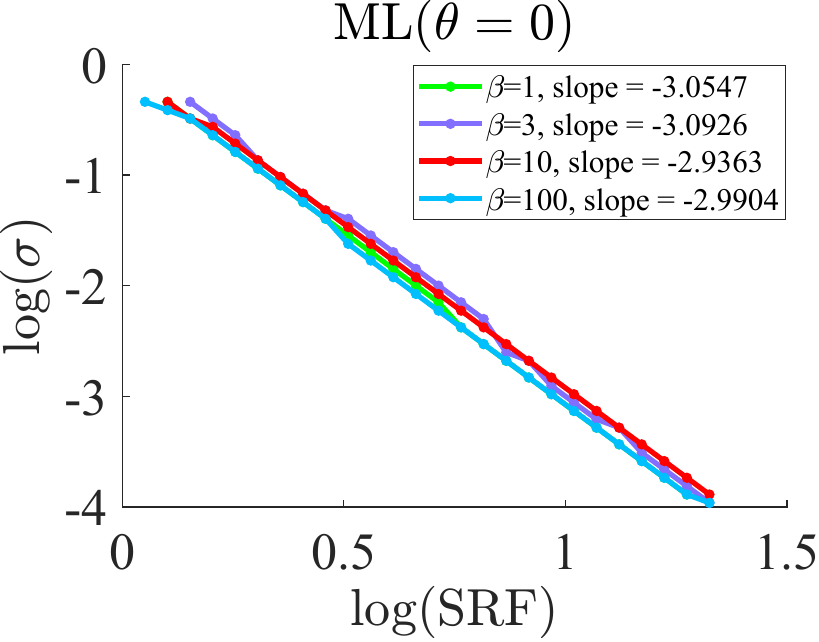}\label{fig:slope_locs_ML_theta0}}  \hfil
\subfloat[CVX \centering]{\includegraphics[height=3.45cm]{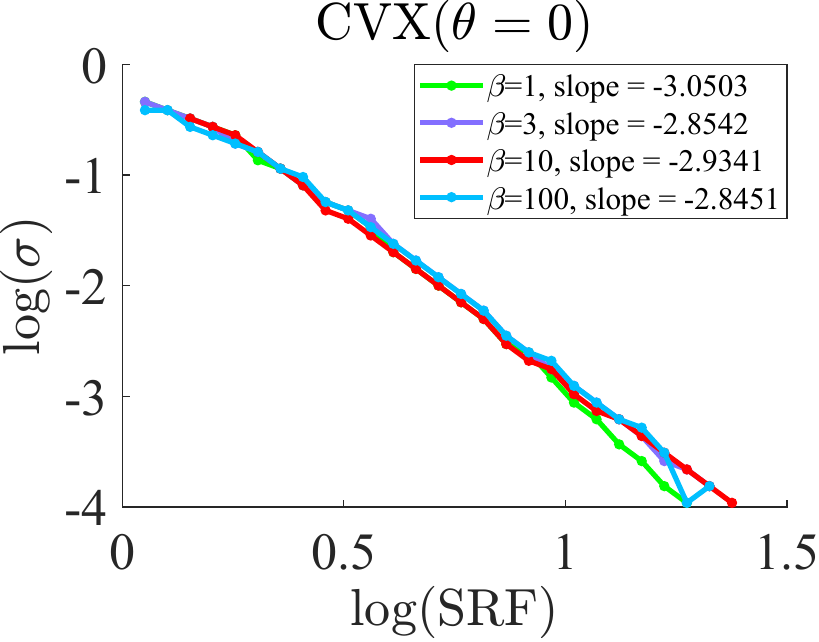}\label{fig:slope_locs_CVX_theta0}} \\

\subfloat[MUSIC \centering]{\includegraphics[height=3.45cm]{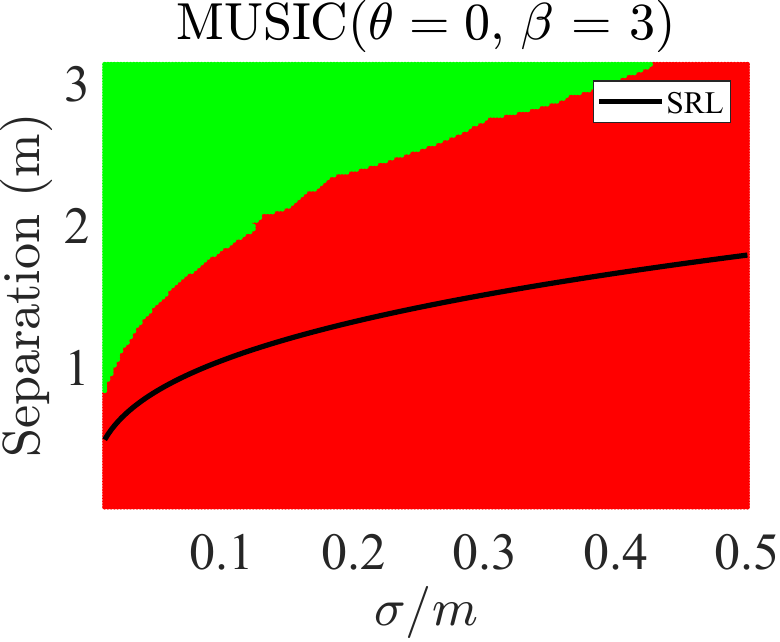}\label{fig:SRL_locs_MUSIC_theta0}}   \hfil
\subfloat[ESPRIT \centering]{\includegraphics[height=3.45cm]{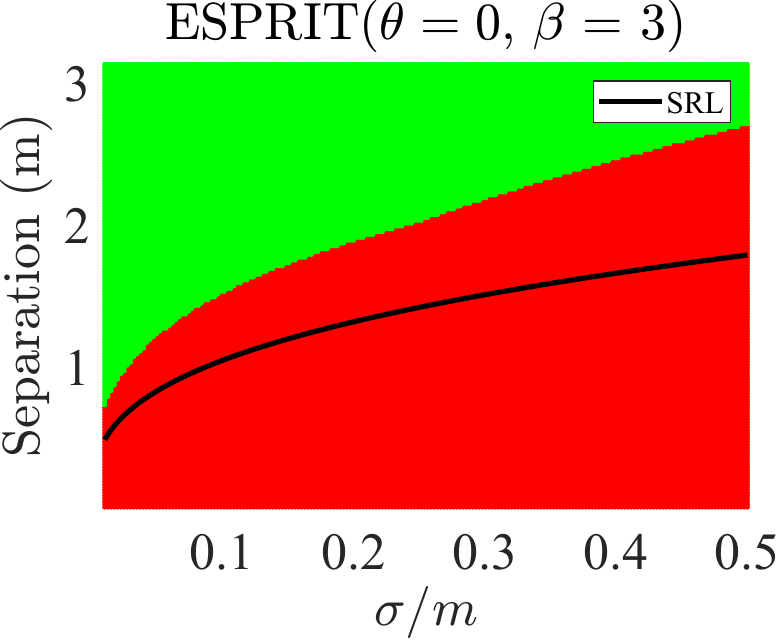}\label{fig:SRL_locs_ESPRIT_theta0}}   \hfil
\subfloat[ML \centering]{\includegraphics[height=3.45cm]{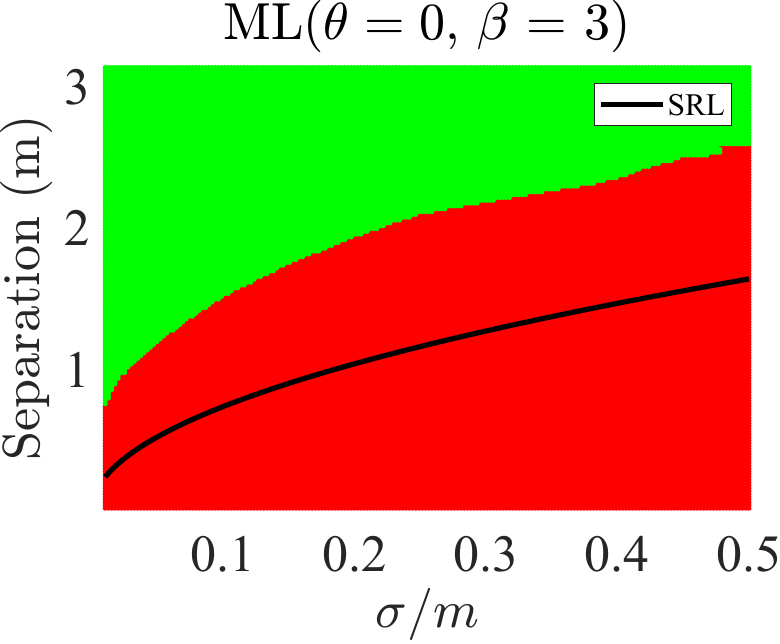}\label{fig:SRL_locs_ML_theta0}}\hfil
\subfloat[CVX \centering]{\includegraphics[height=3.45cm]{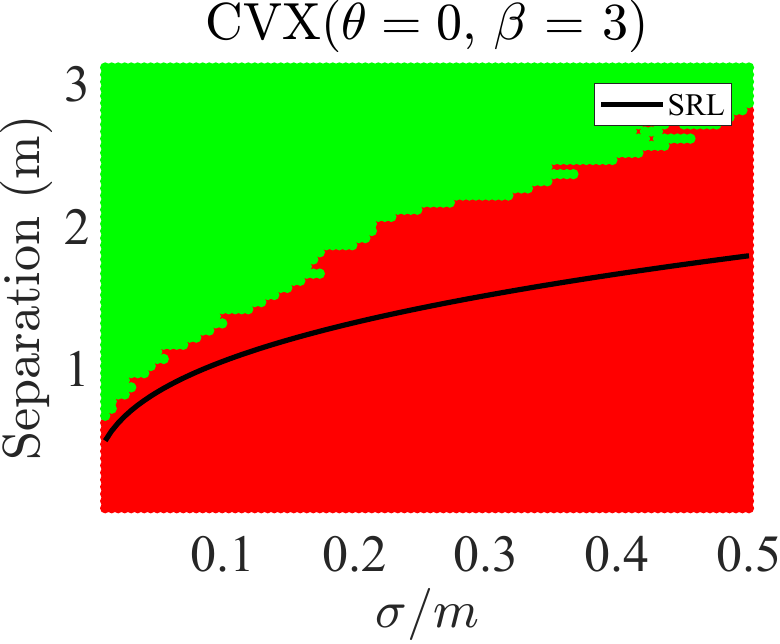}\label{fig:SRL_locs_CVX_theta0}} 

\caption{ Stable performance of each algorithm for location estimation under the in-phase regime.  } 
\label{fig:in_phase_locs_case}
\end{figure*}

\subsection{Numerical Experiments on Location Estimation} \label{experiment:zero_theta_locs}
\subsubsection{Algorithms Considered}
For location estimation, we consider four algorithms: MUSIC, ESPRIT, ML, and CVX. We briefly describe them below.

The MUSIC and ESPRIT methods are representative subspace-based algorithms. In our implementation, when only a single snapshot is available, a Hankel matrix is first constructed from the measurement to form an equivalent multi-snapshot data matrix. MUSIC estimates source locations by exploiting the orthogonality between the signal subspace and the noise subspace \(V\), with pseudospectrum
\[
P_{\mathrm{MUSIC}}(y)
=
\frac{1}{|a(y)^H V V^H a(y)|},
\]
where \(a(y)\) denotes the steering vector at grid point \(y\); the source locations are then obtained from the dominant peaks of this pseudospectrum. ESPRIT, by contrast, estimates source locations through the rotational invariance of two overlapping subarrays. In our implementation, a unitary ESPRIT scheme \cite{haardt2002unitary} is employed. Specifically, forward-backward data augmentation is first applied to construct a centro-Hermitian matrix, which is then transformed into an equivalent real-valued matrix through a unitary transformation. The signal subspace \(U\) is subsequently extracted, and two shifted subarrays \(U_1\) and \(U_2\) are formed so that
\[
U_2 = U_1 \Phi.
\]
The source locations are finally recovered from the eigenvalues of \(\Phi\).

The ML algorithm estimates source locations by searching for the candidate steering subspace that best fits the observation. For each candidate grid point \(y\), one constructs the projection matrix
\[
P_{a(y)}=\frac{a(y)a(y)^H}{a(y)^H a(y)},
\]
where \(a(y)\) is the steering vector at grid point \(y\), and evaluates the projected signal energy through \(\operatorname{tr}\pare{P_{a(y)}R_{\mathrm{sig}}}\), where \(R_{\mathrm{sig}}\) is the sample covariance matrix. In this sense, the ML principle favors the candidate support whose associated steering subspace provides the best fit to the data. In our implementation, this principle is specialized to the two-source setting by exhaustively searching over all pairs \((y_j,y_p)\) of candidate grid points. For each candidate pair, the corresponding complex amplitudes are estimated by least squares, and the pair yielding the smallest residual is selected as the ML estimate.

The CVX algorithm formulates location estimation as a sparse recovery problem over a discrete grid. Specifically, by replacing the nonconvex \(\ell_0\)-minimization with its \(\ell_1\)-relaxation, it solves the constrained convex optimization problem
\[
\min_{\mathbf{x}}\|\mathbf{x}\|_1
\quad
\text{subject to}
\quad
\|\mathbf{Y}-\mathbf{A}\mathbf{x}\|_2 \leq \varepsilon,
\]
where \(\mathbf{A}\) is the sensing dictionary, \(\mathbf{x}\) is the source amplitude vector, and \(\varepsilon\) is a residual tolerance chosen according to the noise level. In our implementation, the resulting convex program is solved using the CVX toolbox in MATLAB.
% , with
% \[
% \varepsilon = c_1\sqrt{M}\sigma,
% \]
% where \(c_1\) is an adaptive coefficient, and \(M\) denotes the number of measurement samples. 
The final CVX-based location estimates are obtained from the dominant peaks of the recovered magnitude spectrum.

\subsubsection{Comparison With Theoretical Predictions} In this subsection, we follow the same experimental procedure as in Section~\ref{experiment:zero_theta_num}. The only additional difference lies in the success criterion for location estimation. A trial is declared successful if the two estimated locations can be matched to the two true source locations and each estimate lies within the \(d/2\)-neighborhood of its corresponding source. Due to the high computational cost of the CVX algorithm, we use \(1000\) Monte Carlo trials and a coarser parameter grid in the resolution--\(\sigma\) plane for CVX experiments.

For the location-estimation problem, the SRL in \eqref{equ:condition_locs_positive} yields
\begin{align*}
    \log(\sigma)>-3\,\log(\mathrm{SRF})+\log(m)
    +3\,\log\pare{\frac{\pi}{2.25}},
\end{align*}
which gives a linear lower-bound boundary with slope $-3$ in the 
$\log(\mathrm{SRF})$--$\log(\sigma)$ plane.  
As shown in Figs.~\ref{fig:slope_locs_MUSIC_theta0}--\ref{fig:slope_locs_CVX_theta0}, the empirical phase-transition curves of MUSIC, ESPRIT, ML, and CVX are approximately linear, and their fitted slopes are close to $-3$ for different values of $\beta$. This indicates that all four algorithms achieve the optimal resolution order predicted by the SRL. The corresponding comparisons in Figs.~\ref{fig:SRL_locs_MUSIC_theta0}--\ref{fig:SRL_locs_CVX_theta0} further show that 
although the empirical resolution boundaries of all algorithms remain above the theoretical SRL, their distances from the SRL differ noticeably. 
Among the tested algorithms, ESPRIT gives the boundary closest to the SRL and hence exhibits the best performance, whereas MUSIC shows the weakest performance in this comparison. 
Hence, beyond verifying the optimal scaling law, these red--green phase-transition diagrams quantify how close each practical algorithm is to the fundamental resolution limit, thereby providing a criterion for selecting the most suitable algorithm. 
%Together with the source-number detection results, the in-phase experiments validate that the optimal resolution orders are $\sigma\propto\mathrm{SRF}^{-2}$ for source-number detection and $\sigma\propto\mathrm{SRF}^{-3}$ for location estimation.

\section{Near-Endpoint Phase Regime $\left(\babs{\theta}_{\min}\asymp\pare{\frac{\sigma}{m}}^\frac12\right)$}
\label{section:small_theta}

We now move from the in-phase case to an out-of-phase regime in which the relative phase is nonzero but remains close to one of the two endpoints, $0$ or $\pi$.

\subsection{Theoretical Bounds}\label{subsection:boundary_small_theta}

In this subsection, we derive the corresponding SRUs and SRLs for source-number detection and location estimation, thereby quantifying the effect of a near-endpoint phase on the resolution limits. 
Proofs of these theorems are provided in Appendix~\ref{proof:small_theta}.

\begin{figure*}[tbp]
\centering
\subfloat[$\ell_0$ \centering]{\includegraphics[height=3.45cm]{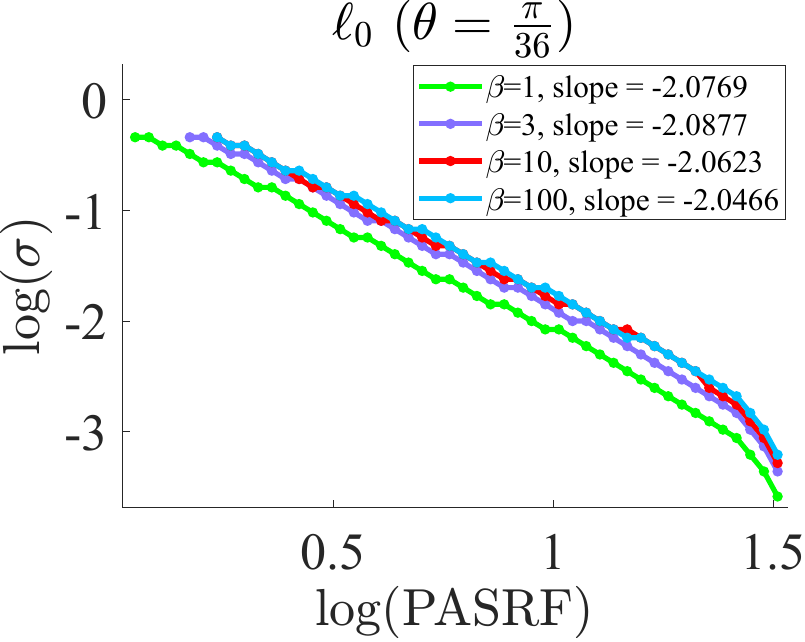}\label{fig:slope_num_L0_theta5}}   \hfil
\subfloat[SVT \centering]{\includegraphics[height=3.45cm]{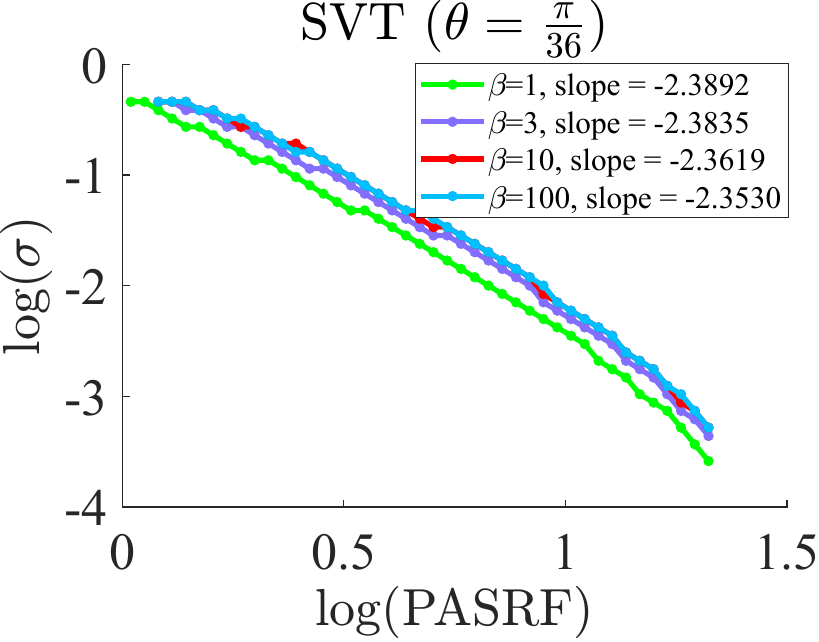}\label{fig:slope_num_Threshold_theta5}} \hfil
\subfloat[$\ell_0$ \centering]{\includegraphics[height=3.45cm]{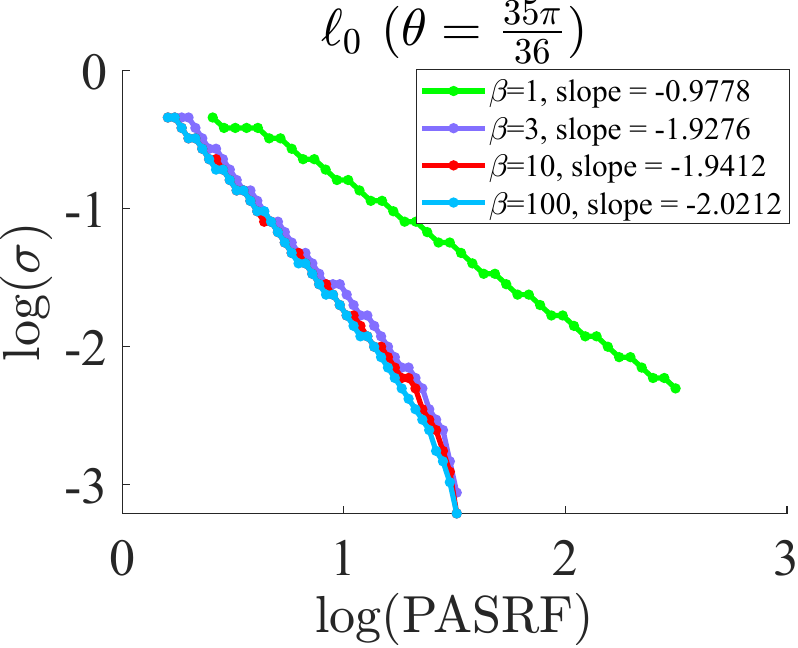}\label{fig:slope_num_L0_theta175}}   \hfil
\subfloat[SVT \centering]{\includegraphics[height=3.45cm]{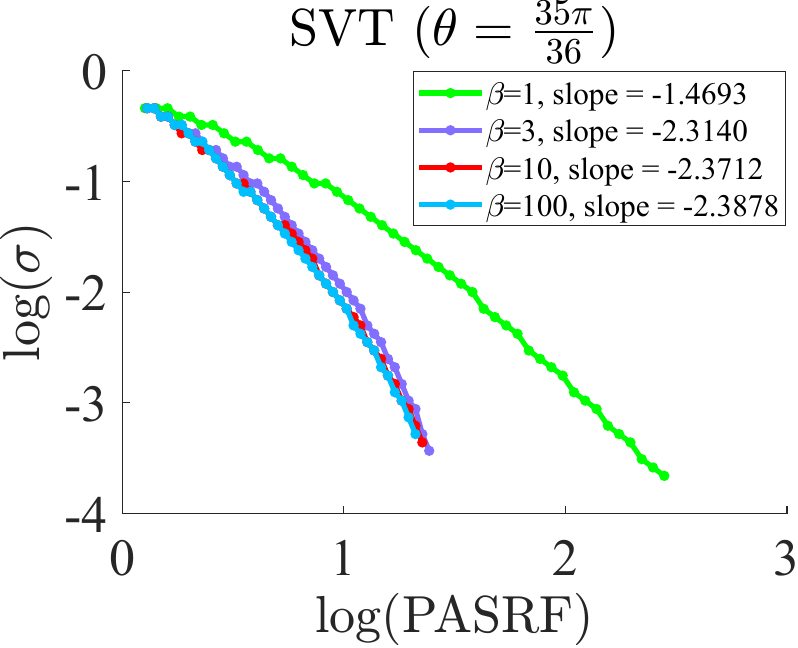}\label{fig:slope_num_Threshold_theta175}} \\

\subfloat[$\ell_0$ \centering]{\includegraphics[height=3.45cm]{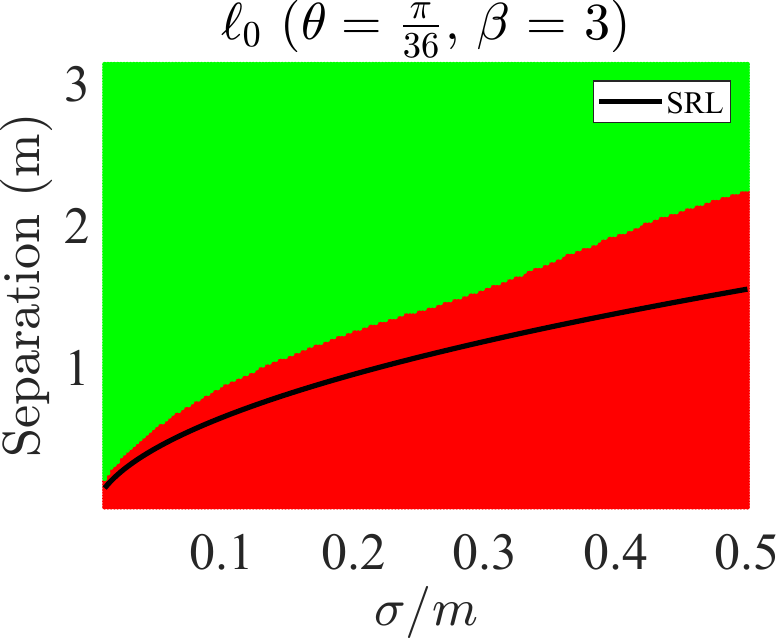}\label{fig:num_SRL_beta3_theta5_L0_fast}}  \hfil
\subfloat[SVT \centering]{\includegraphics[height=3.45cm]{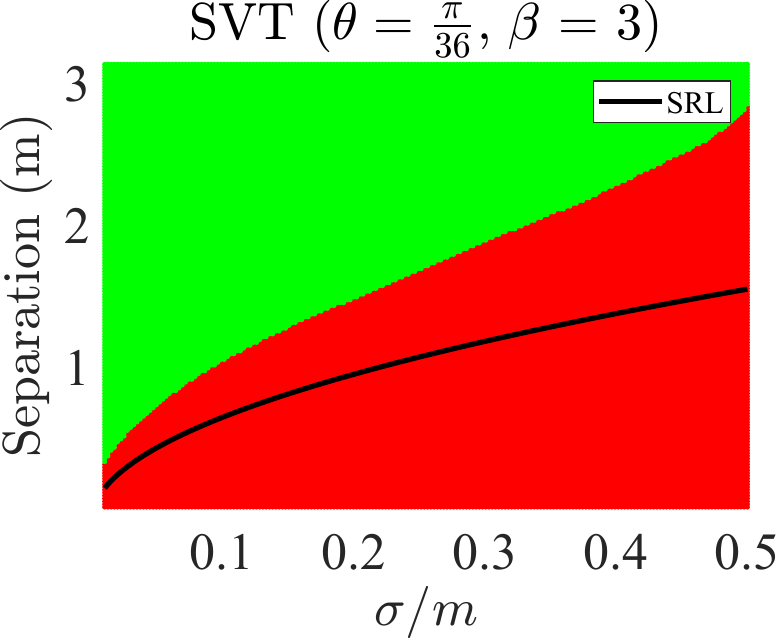}\label{fig:num_SRL_beta3_theta5_MUSIC_fast}}  \hfil
\subfloat[$\ell_0$ \centering]{\includegraphics[height=3.45cm]{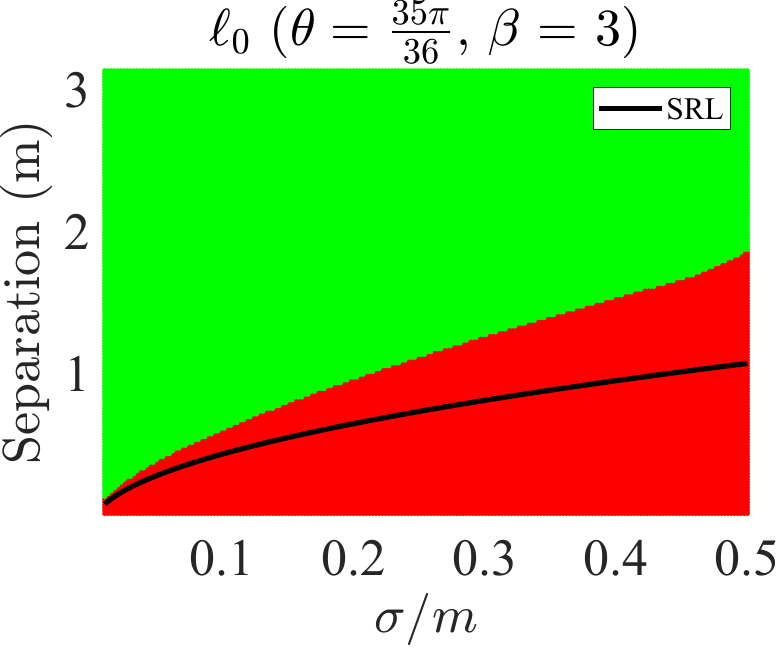}\label{fig:num_SRL_beta3_theta175_L0_only}}  \hfil
\subfloat[SVT \centering]{\includegraphics[height=3.45cm]{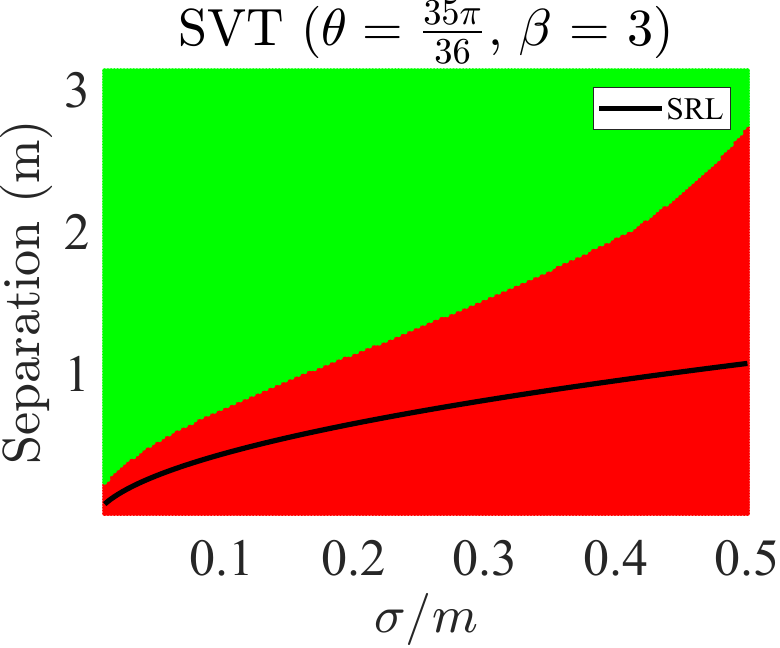}\label{fig:num_SRL_beta3_theta175_MUSIC_only}} 
\caption{ Stable performance of each algorithm for source-number detection under the near-endpoint phase regime.  } 
\label{fig:small_theta_num_case}
\end{figure*}

\begin{thm} \label{thm:twopointresolution_complex_number_UB}
Let $\mathbf{Y}$ be generated by a complex measure
$\mu=me^{i\theta_1} \delta_{y_1}+\beta m e^{i\theta_2} \delta_{y_2}$ with 
$y_1, y_2\in B_{\frac{\pi}{2\Omega}}(0)$, $\theta_1, \theta_2\in \pare{-\pi,\pi}$, $\beta \geqs 1$, and $m>0$.
Define the relative phase $\theta:=\theta_1-\theta_2\in\pare{-\pi,\pi}$ and set $\babs{\theta}_{\min} := \min\{\babs \theta,\pi-\babs \theta\}$.
When $\babs{\theta}_{\min} \leqs\frac{4\pi}{3}\sqrt{1+\frac{1.25}{\beta}}\pare{\frac{\sigma}{m}}^\frac{1}{2}$, if 
\begin{align}
\label{equ:sepacondinumbercomplex_smalltheta}
  \babs{y_1-y_2}\geqs \frac{2\pi\pare{\pare{2+\frac{2.5}{\beta}}\frac{\sigma}{m}}^\frac{1}{2}-  \babs{\theta}_{\min}}{\Omega}, 
\end{align}
then no $\sigma$-admissible measure for $\mathbf{Y}$ can be supported on fewer than two points. Moreover, if
\begin{align}
\label{equ:separation_condition_num_complex_LB_small}
\begin{cases}
     \displaystyle\babs{y_1-y_2}  <  \frac{2\pare{ \frac{\beta+1}{\beta}\frac{\sigma}{m} }^{\frac{1}{2}}-\babs{\theta}_{\min}}{\Omega}, &  \displaystyle\babs{\theta} \asymp \pare{ \frac{\sigma}{m} }^{\frac{1}{2}};\\
     \displaystyle \babs{y_1-y_2}  <  \frac{2\pare{ \frac{\beta-1}{\beta}\frac{\sigma}{m} }^{\frac{1}{2}}-\babs{\theta}_{\min}}{\Omega}, & \displaystyle \pi - \babs{\theta} \asymp \pare{ \frac{\sigma}{m} }^{\frac{1}{2}},
\end{cases}
\end{align}
then there exists a $\sigma$-admissible measure $\widehat \mu$ for $\mathbf{Y}$ supported on a single point.
\end{thm}

\begin{remark}
The lower bound in (\ref
{equ:separation_condition_num_complex_LB_small}) does not cover the case for $\beta =1$ when $\theta \approx \pi$. In particular, we observe a surprising improvement in the resolution order when $\beta=1$ and $\pi-\babs{\theta}\asymp \pare{\frac{\sigma}{m}}^\frac{1}{2}$; see the Theorem~\ref{thm:twopointresolution_complex_number_UB_beta1} below and experiments in Figs.~\ref{fig:slope_num_L0_theta175} and \ref{fig:slope_num_Threshold_theta175} for theoretical and numerical elucidations.
\begin{thm} \label{thm:twopointresolution_complex_number_UB_beta1}
Let $\mathbf{Y}$ be generated by a complex measure
$\mu=me^{i\theta_1} \delta_{y_1}+ m e^{i\theta_2} \delta_{y_2}$ with 
$y_1, y_2\in B_{\frac{\pi}{2\Omega}}(0)$, $\theta_1, \theta_2\in \pare{-\pi,\pi}$, and $m>0$.
Define the relative phase $\theta:=\theta_1-\theta_2\in\pare{-\pi,\pi}$.
When $\pi-\babs{\theta}\asymp \pare{\frac{\sigma}{m}}^\frac{1}{2}$,   if 
\begin{align}
\label{equ:sepacondinumbercomplex_smalltheta_beta1}
  \babs{y_1-y_2}\geqs \frac{2}{\Omega}\arcsin \pare{\frac{1}{\cos\pare{\frac{\pi-|\theta|}{2}}}\frac{\sigma}{m}},
\end{align}
then no $\sigma$-admissible measure for $\mathbf{Y}$ can be supported on fewer than two points. Otherwise, there exists a $\sigma$-admissible measure $\widehat \mu$ for $\mathbf{Y}$ supported on a single point.
\end{thm}
\end{remark}

Having characterized the resolution limits for source-number detection, we next study the resolution limits required for location estimation.

\begin{thm}
\label{thm:supportrecoveryupperboundcomplex_small_theta}
Let $\mathbf{Y}$ be generated by a complex measure
$\mu=me^{i\theta_1} \delta_{y_1}+\beta m e^{i\theta_2} \delta_{y_2}$ with 
$y_1, y_2\in B_{\frac{\pi}{2\Omega}}(0)$, $\theta_1, \theta_2\in \pare{-\pi,\pi}$, $\beta\geqs1$, and $m>0$.
Define $d:= \babs{y_1-y_2}$,  $\theta:=\theta_1-\theta_2\in\pare{-\pi,\pi}$, and $\babs{\theta}_{\min} := \min\{\babs \theta,\pi-\babs \theta\}$.
When $\babs{\theta}_{\min} \leqs1.75\pi \pare{\frac{\sigma}{m}}^\frac{1}{3}$, if
\begin{align}
\label{equ:sepacondilocationcomplex_small_beta1}
    \displaystyle \babs{y_1-y_2}\geqs\frac{6\pi\pare{\frac{\sigma}{m}}^\frac{1}{3}-\babs{\theta}_{\min}}{2\Omega}, 
\end{align}
and $\widehat{\mu} = \widehat{a}_1 \delta_{\widehat{y}_1} + \widehat{a}_2 \delta_{\widehat{y}_2}$ supported on $B_{\frac{\pi}{2\Omega}}(0)$ is a $\sigma$-admissible measure for $\mathbf{Y}$, then $\widehat{\mu}$ lies within the $\frac{d}{2}$-neighborhood of $\mu$. Moreover, if 
\begin{align}
\label{equ:separation_condition_supp_complex_LB_small}
\begin{cases}
     \displaystyle \babs{y_1-y_2} < \frac{2.23\pare{\frac{\sigma}{m}}^\frac{1}{3}-\babs{\theta}_{\min}}{\Omega}, &\displaystyle \babs{\theta}\asymp \pare{\frac{\sigma}{m}}^\frac{1}{3};\\
     \displaystyle \babs{y_1-y_2}<\frac{2.23\pare{\frac{\sigma}{ \beta m}}^\frac{1}{3}-\babs{\theta}_{\min}}{\Omega}, & \displaystyle \pi-\babs\theta \asymp \pare{\frac{\sigma}{m}}^\frac{1}{3},
\end{cases}
\end{align}
then there exists a $\sigma$-admissible measure $\widehat \mu=\widehat{a}_1 \delta_{\widehat{y}_1} + \widehat{a}_2 \delta_{\widehat{y}_2}$ that does not lie within the $\frac{d}{2}$-neighborhood of $\mu$.
\end{thm}

The preceding SRUs and SRLs indicate that, in the near-endpoint regime, the effective relative phase $\babs{\theta}_{\min}$ enhances resolvability while preserving the generic resolution order. Specifically, the phase term appears as a subtraction, and hence decreases the required source separation.  These theoretical predictions are corroborated by the numerical experiments in the next subsection.

% \subsection{Theoretical Boundｓ when $\beta=1$}\label{subsection:boundary_small_theta_beta1}

% \begin{thm}
% \label{thm:supportrecoveryupperboundcomplex_small_theta_beta1}
% Let $\mathbf{Y}$ be generated by a complex measure
% $\mu=me^{i\theta_1} \delta_{y_1}+\beta m e^{i\theta_2} \delta_{y_2}$ with 
% $y_1, y_2\in B_{\frac{\pi}{2\Omega}}(0)$, $\theta_1, \theta_2\in \pare{-\pi,\pi}$, $\beta\geqs1$, and $m>0$.
% Define $d:= \babs{y_1-y_2}$,  $\theta:=\theta_1-\theta_2\in\pare{-\pi,\pi}$, and $\babs{\theta}_{\min} := \min\{\babs \theta,\pi-\babs \theta\}$.
% When $\babs{\theta}_{\min} \leqs1.75\pi \pare{\frac{\sigma}{m}}^\frac{1}{3}$, if
% \begin{align}
% \label{equ:sepacondilocationcomplex_small_beta1}
%     \babs{y_1-y_2}\geqs XXX, 
% \end{align}
% and $\widehat{\mu} = \widehat{a}_1 \delta_{\widehat{y}_1} + \widehat{a}_2 \delta_{\widehat{y}_2}$ supported on $B_{\frac{\pi}{2\Omega}}(0)$ is a $\sigma$-admissible measure for $\mathbf{Y}$, then $\widehat{\mu}$ lies within the $\frac{d}{2}$-neighborhood of $\mu$. Moreover, if 
% \begin{align}
% \label{equ:separation_condition_supp_complex_LB_small_beta1}
%     \babs{y_1-y_2}<\frac{2\arcsin{\pare{\pare{\frac{\sigma}{2m}}^\frac{1}{2}}}-\babs{\theta}_{\min}}{\Omega},\quad\pi-\babs\theta \asymp \pare{\frac{\sigma}{m}}^\frac{1}{3};
% \end{align}
% then there exists a $\sigma$-admissible measure $\widehat \mu=\widehat{a}_1 \delta_{\widehat{y}_1} + \widehat{a}_2 \delta_{\widehat{y}_2}$ that does not lie within the $\frac{d}{2}$-neighborhood of $\mu$.
% \end{thm}

\subsection{Numerical Experiments on Source-Number Detection} 

In this subsection, we choose the $\ell_0$ and SVT algorithms to validate the theoretical SRLs for source-number detection established in Sections~\ref{subsection:boundary_small_theta}. Considering the effect of $\babs{\theta}_{\min}$, we define the phase-adjusted super-resolution factor (PASRF) as
\[
\mathrm{PASRF}=\frac{\pi}{d\Omega+|\theta|_{\min}},
\]
and validate the slopes in the $\log(\mathrm{PASRF})$--$\log(\sigma)$ plane.
\begin{figure*}[bp]
\centering
\subfloat[MUSIC \centering]{\includegraphics[height=3.45cm]{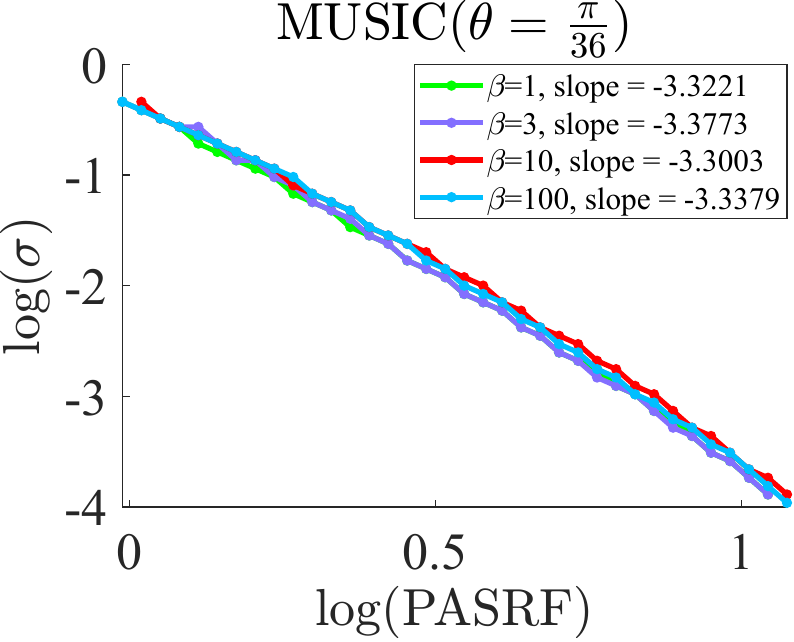}\label{fig:slope_locs_MUSIC_theta5}}  \hfil
\subfloat[ESPRIT \centering]{\includegraphics[height=3.45cm]{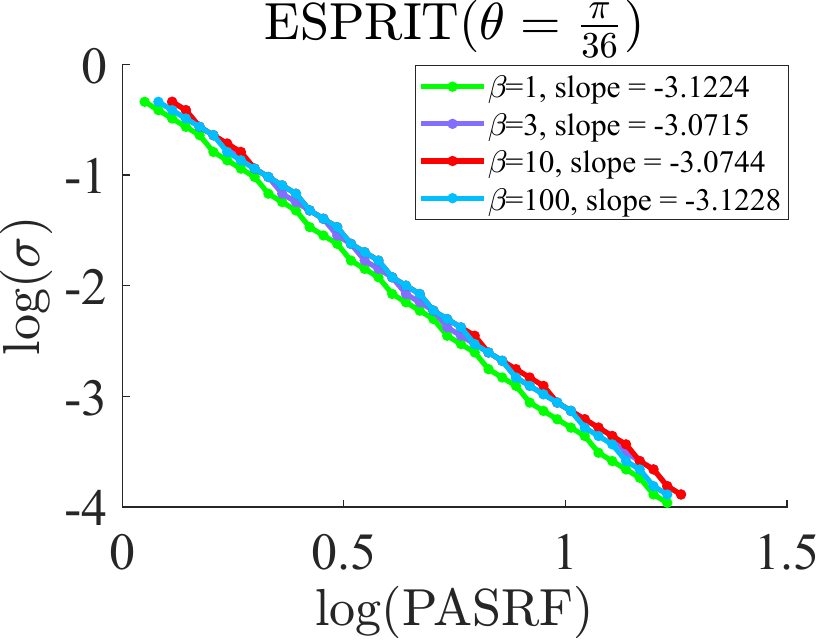}\label{fig:slope_locs_ESPRIT_theta5}}  \hfil
\subfloat[ML \centering]{\includegraphics[height=3.45cm]{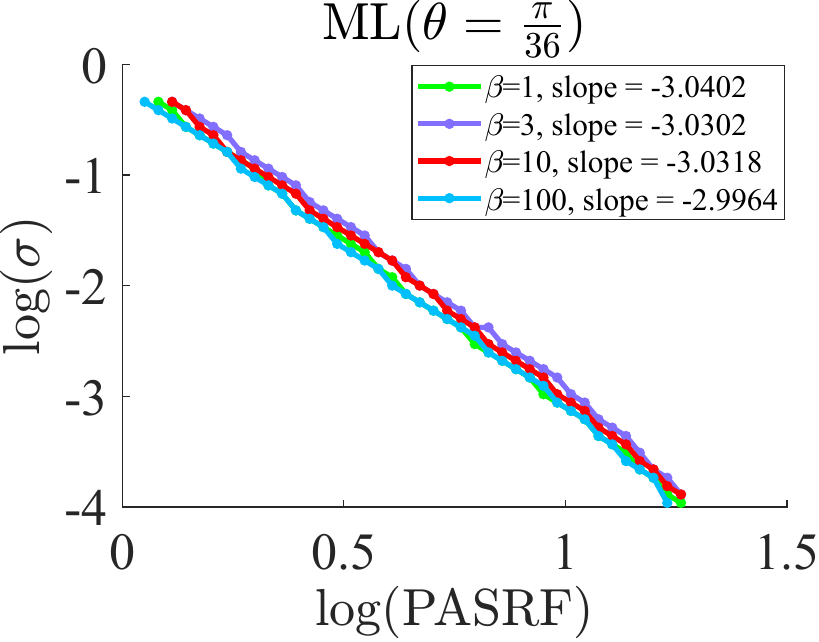}\label{fig:slope_locs_ML_theta5}}  \hfil
\subfloat[CVX \centering]{\includegraphics[height=3.45cm]{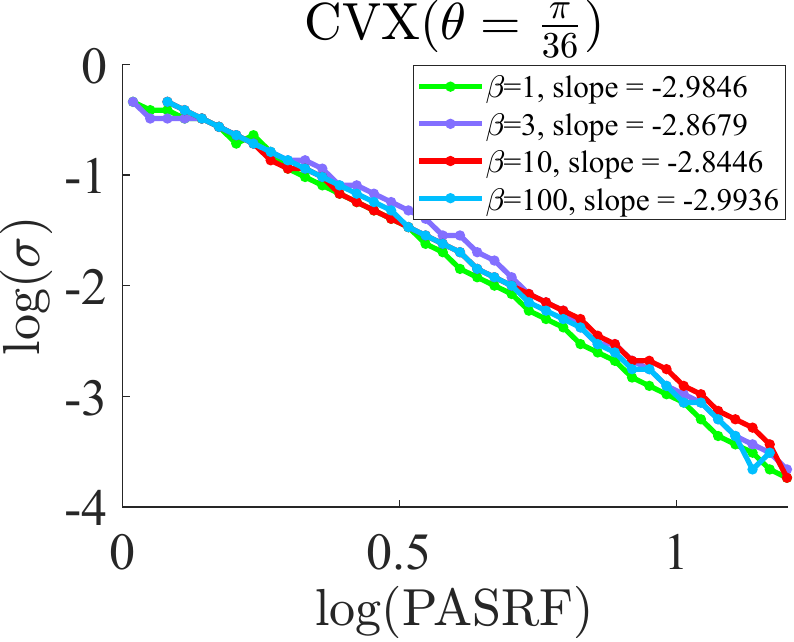}\label{fig:slope_locs_CVX_theta5}}  \\

\subfloat[MUSIC \centering]{\includegraphics[height=3.45cm]{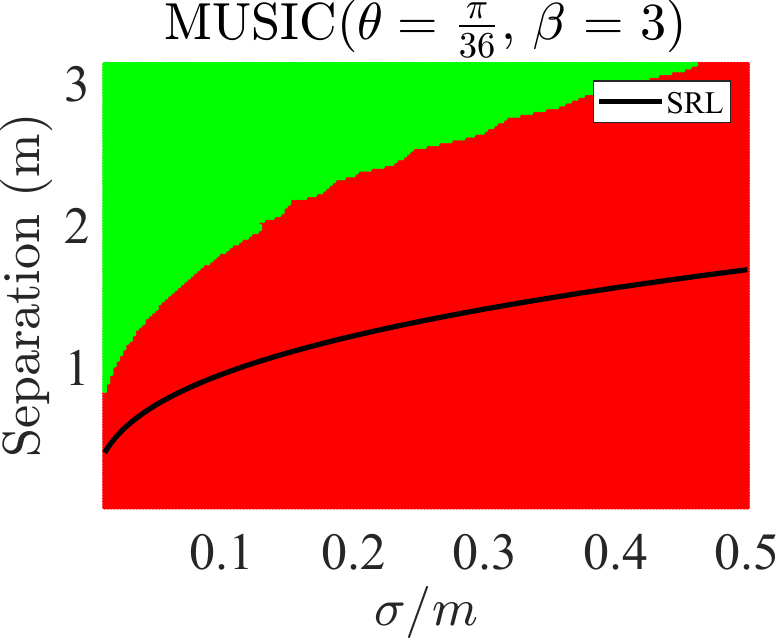}\label{fig:locs_beta3_theta5_MUSIC}}  \hfil
\subfloat[ESPRIT \centering]{\includegraphics[height=3.45cm]{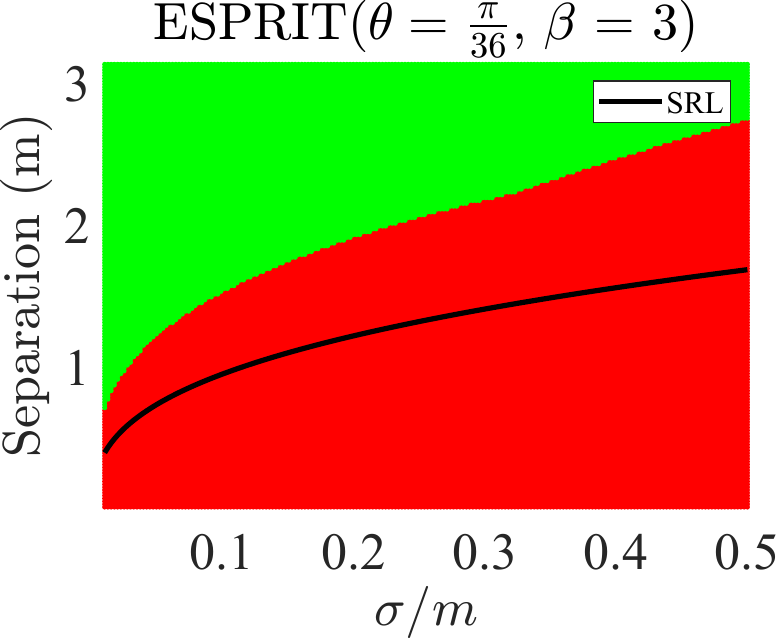}\label{fig:locs_beta3_theta5_ESPRIT}}  \hfil
\subfloat[ML \centering]{\includegraphics[height=3.45cm]{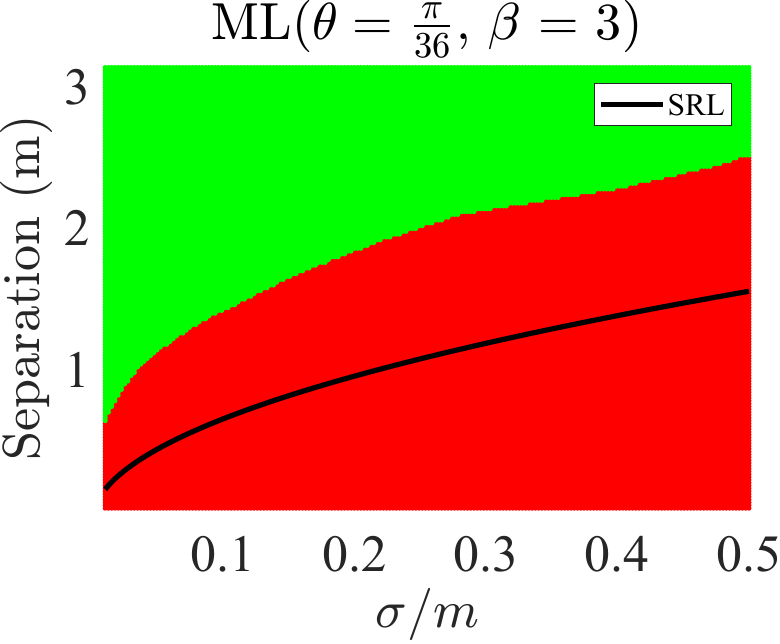}\label{fig:locs_SRL_beta3_theta5_ML}}  \hfil
\subfloat[CVX \centering]{\includegraphics[height=3.45cm]{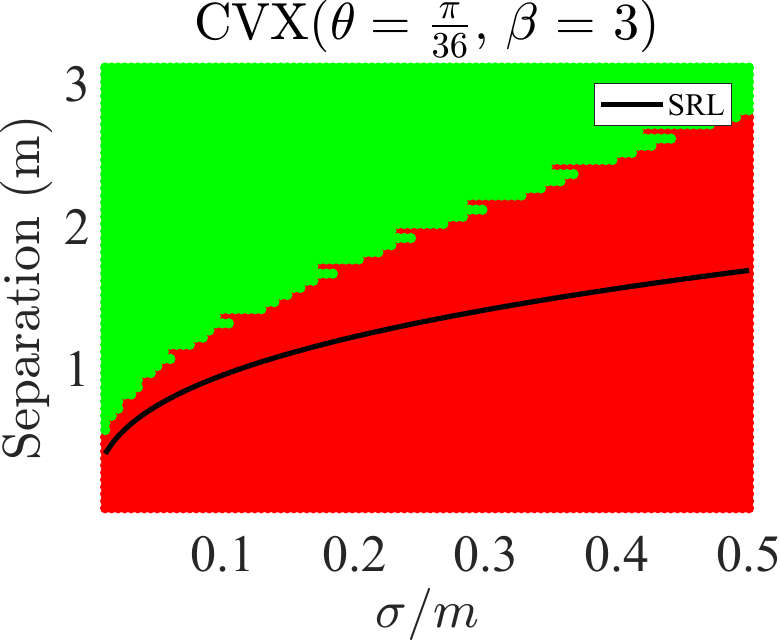}\label{fig:locs_SRL_beta3_theta5_CVX}}\\

\subfloat[MUSIC \centering]{\includegraphics[height=3.45cm]{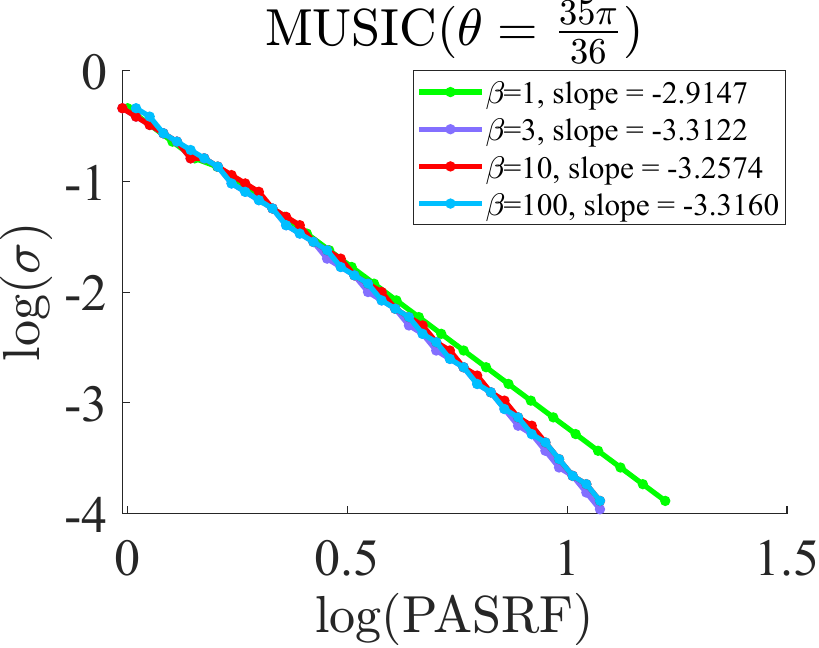}\label{fig:slope_locs_MUSIC_theta175}}  \hfil
\subfloat[ESPRIT \centering]{\includegraphics[height=3.45cm]{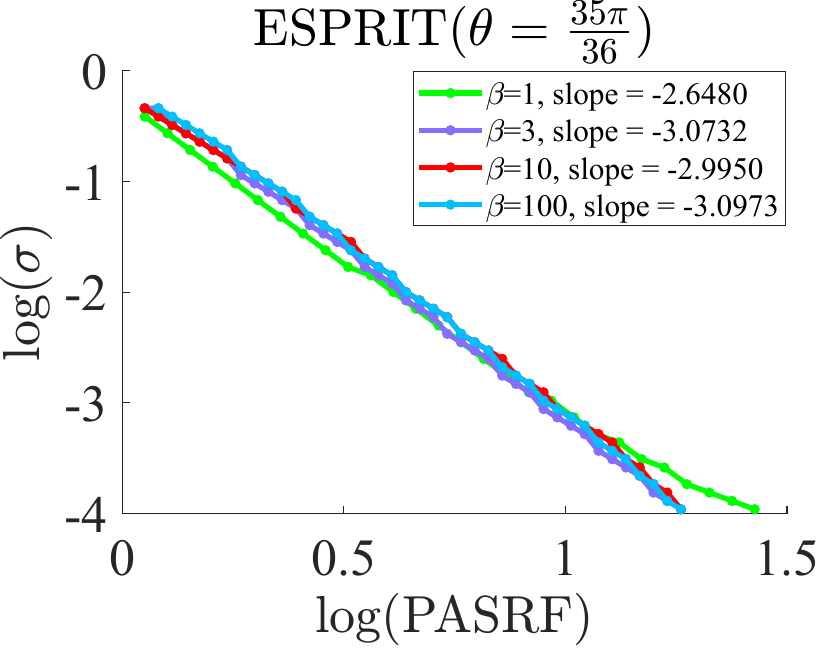}\label{fig:slope_locs_ESPRIT_theta175}}  \hfil
\subfloat[ML \centering]{\includegraphics[height=3.45cm]{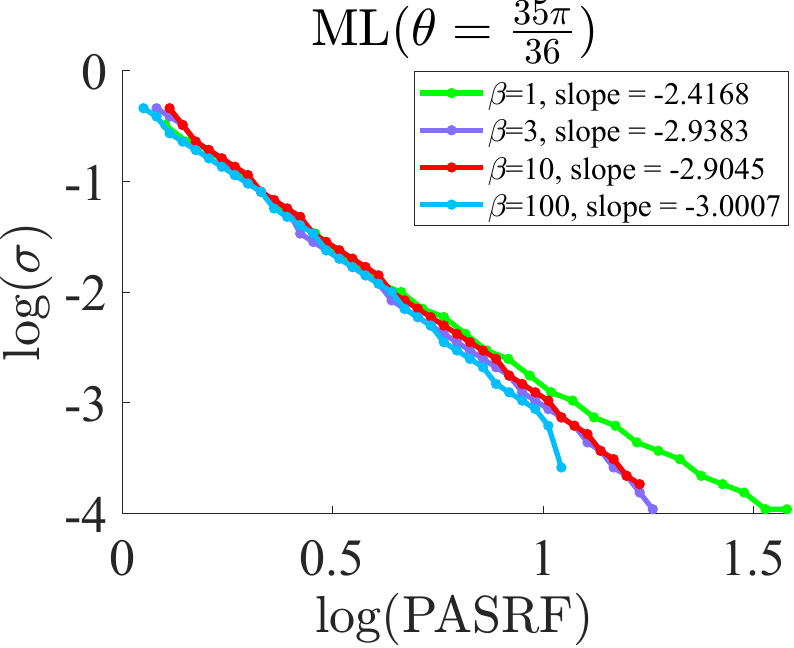}\label{fig:slope_locs_ML_theta175}}  \hfil
\subfloat[CVX \centering]{\includegraphics[height=3.45cm]{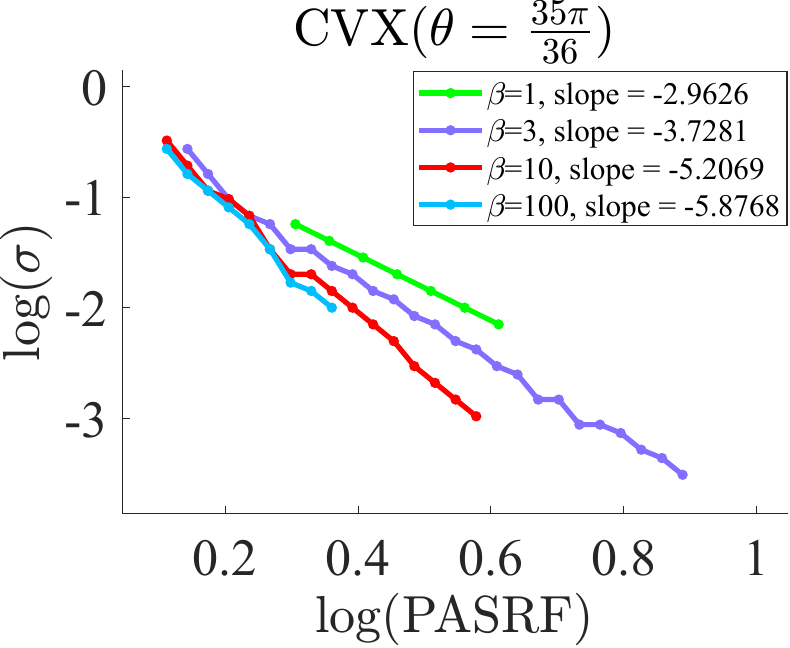}\label{fig:slope_locs_CVX_theta175}}  \\

\subfloat[MUSIC \centering]{\includegraphics[height=3.45cm]{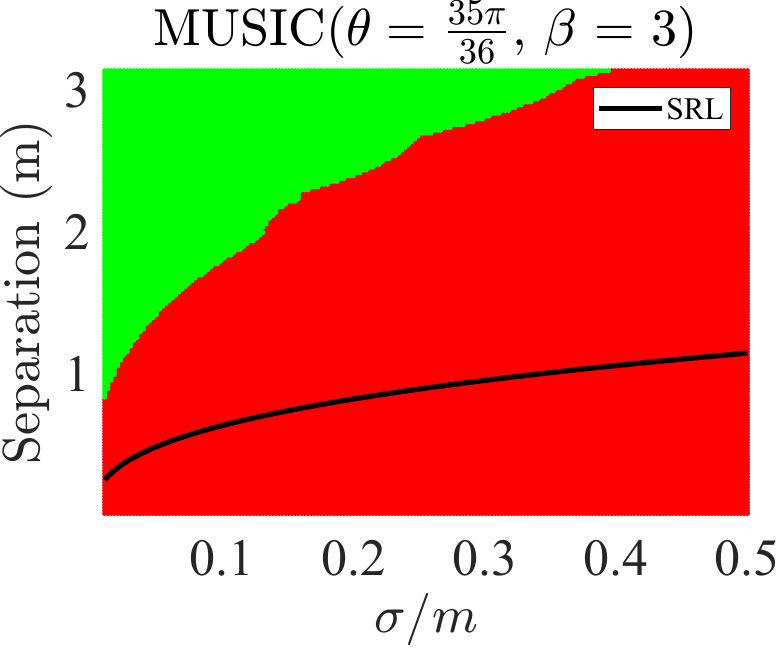}\label{fig:locs_beta3_theta175_MUSIC}}  \hfil
\subfloat[ESPRIT \centering]{\includegraphics[height=3.45cm]{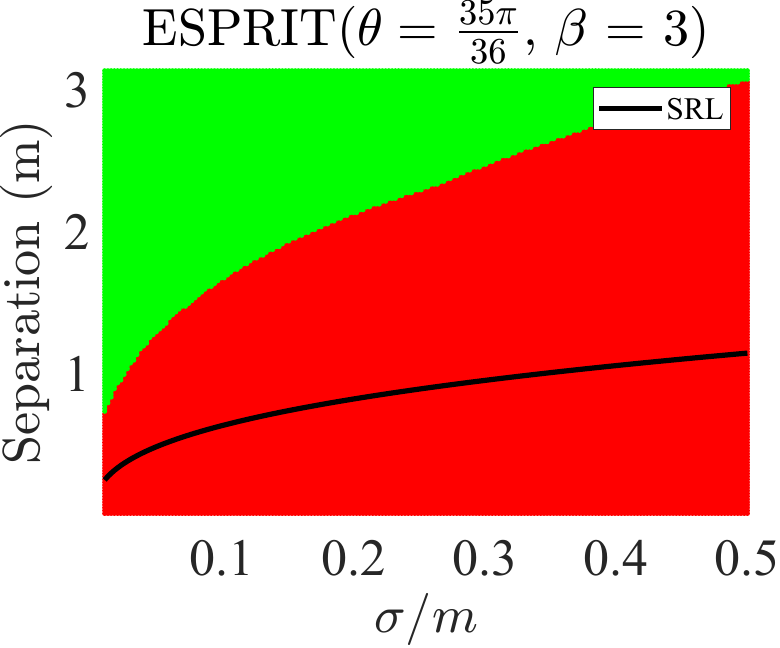}\label{fig:locs_beta3_theta175_ESPRIT}}  \hfil
\subfloat[ML \centering]{\includegraphics[height=3.45cm]{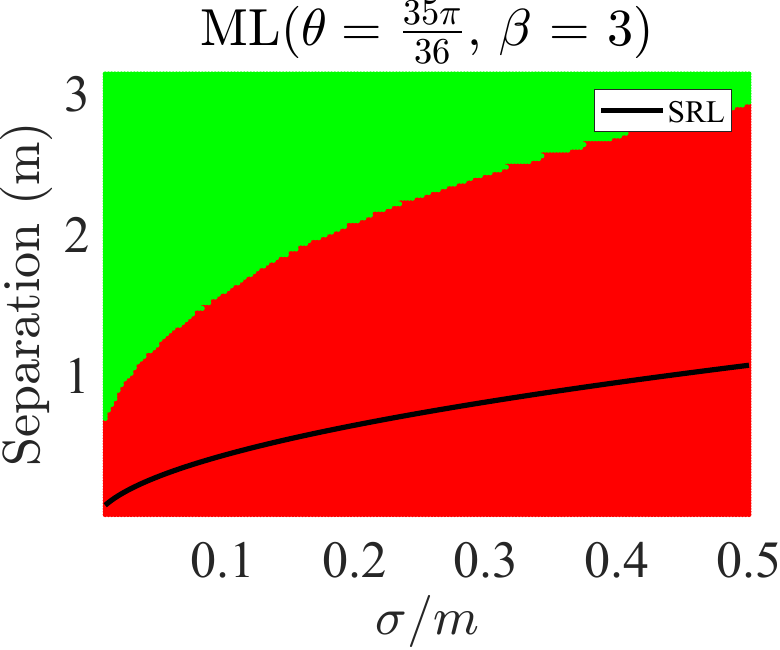}\label{fig:locs_SRL_beta3_theta175_ML}}  \hfil
\subfloat[CVX \centering]{\includegraphics[height=3.45cm]{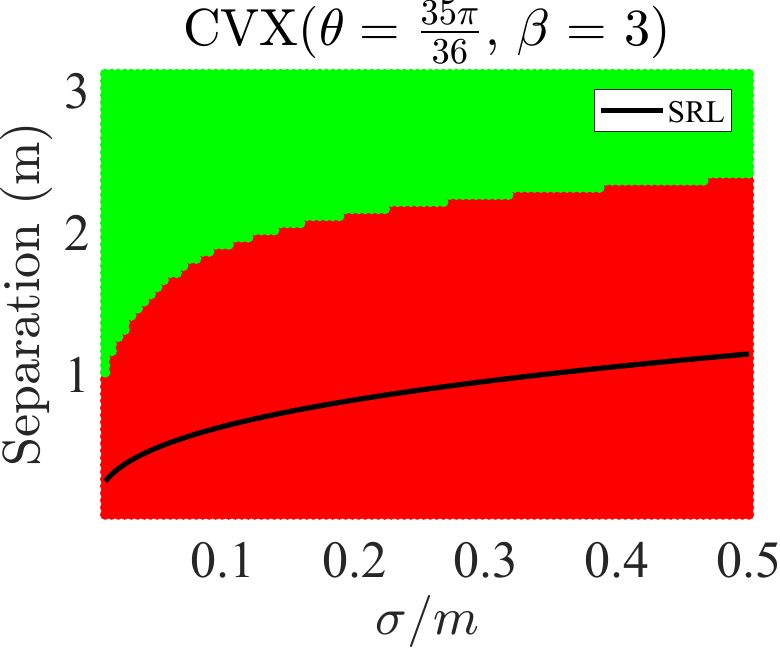}\label{fig:locs_SRL_beta3_theta175_CVX}} 

\caption{ Stable performance of each algorithm for location estimation under the near-endpoint phase regime.  } 
\label{fig:small_theta_locs_case}
\end{figure*}

Fig.~\ref{fig:small_theta_num_case} examines the source-number detection problem in two representative near-endpoint phase regimes, namely $|\theta|=\pi/36=5^\circ$ and $|\theta|=35\pi/36=175^\circ$. 
For the generic case covered by \eqref{equ:separation_condition_num_complex_LB_small}, the SRL implies
\begin{align*}
\begin{cases}
    \log(\sigma)>-2\log(\mathrm{PASRF})-\log\!\left(\frac{\beta+1}{\beta}\right)+C_1, 
    & |\theta|\asymp \left(\frac{\sigma}{m}\right)^{1/2},\\[0.5ex]
    \log(\sigma)>-2\log(\mathrm{PASRF})-\log\!\left(\frac{\beta-1}{\beta}\right)+C_1, 
    & \pi-|\theta|\asymp \left(\frac{\sigma}{m}\right)^{1/2},
\end{cases}
\end{align*}
where $C_1:=\log(m)+2\log(\pi/2)$. 
Thus, the predicted resolution order corresponds to a slope $-2$ in the $\log(\mathrm{PASRF})$--$\log(\sigma)$ plane. By contrast, when $\beta=1$ and $\babs{\theta}$ close to $\pi$, \eqref{equ:sepacondinumbercomplex_smalltheta_beta1} instead gives
\begin{align*}
    \log(\sigma)>
    \log\!\left[\sin\!\left(\frac{\pi}{2\,\mathrm{PASRF}}\right)\right]
    +\log(2m),
\end{align*}
which yields the improved slope $-1$.

The slope plots in Figs.~\ref{fig:slope_num_L0_theta5}--\ref{fig:slope_num_Threshold_theta175} are consistent with these predictions. 
For $|\theta|=\pi/36$, the empirical slopes of the $\ell_0$ method are close to $-2$ for all tested values of $\beta$, showing that it attains the predicted resolution order. 
SVT also yields nearly linear phase-transition boundaries, but its slopes deviate more noticeably from $-2$. 
For $|\theta|=35\pi/36$, the $\ell_0$ method accurately captures the phase-dependent transition: the slope is close to $-1$ in the case $\beta=1$, while it returns to approximately $-2$ for $\beta>1$. 
Although the SVT slopes are less accurate quantitatively, SVT still exhibits a clear change in the case $\beta=1$, with an approximately one-order shift relative to its own generic near-endpoint scaling. 
Thus, both methods reflect the special cancellation phenomenon, while $\ell_0$ matches the predicted scaling more closely.
The red--green phase-transition diagrams in Figs.~\ref{fig:num_SRL_beta3_theta5_L0_fast}--\ref{fig:num_SRL_beta3_theta175_MUSIC_only} give a comparison for $\beta=3$. 
For both near-$0$ and near-$\pi$ phases, the empirical success--failure boundaries remain above the theoretical SRL. 
Moreover, the $\ell_0$ boundaries are closer to the SRL than those of SVT, indicating that $\ell_0$ is empirically closer to the fundamental resolution limit in this setting. 
Therefore, Fig.~\ref{fig:small_theta_num_case} confirms not only the predicted near-endpoint resolution orders, but also the special order improvement.

\subsection{Numerical Experiments on Location Estimation}

In this subsection, similarly as the experimental framework in Section~\ref{experiment:zero_theta_locs}, we choose the MUSIC, ESPRIT, ML, and CVX algorithms to validate the theoretical SRLs for location estimation established in Sections~\ref{subsection:boundary_small_theta}, and validate the slopes in the $\log(\mathrm{PASRF})$--$\log(\sigma)$ plane.

Fig.~\ref{fig:small_theta_locs_case} examines the location-estimation problem in the same two near-endpoint phase regimes, $|\theta|=\pi/36$ and $|\theta|=35\pi/36$. 
From \eqref{equ:separation_condition_supp_complex_LB_small}, the SRL implies
\begin{align*}
\begin{cases}
    \log(\sigma)>-3\log(\mathrm{PASRF}) + C_2, 
    & |\theta|\asymp \left(\frac{\sigma}{m}\right)^{1/3},\\[0.5ex]
    \log(\sigma)>-3\log(\mathrm{PASRF}) + \log(\beta) + C_2, 
    & \pi-|\theta|\asymp \left(\frac{\sigma}{m}\right)^{1/3},
\end{cases}
\end{align*}
where $C_2:=\log(m)+3\log(\pi/2.23)$. Hence, the lower bound predicts the optimal slope of $-3$ without a special order transition for location estimation.
 
For $|\theta|=\pi/36$, Figs.~\ref{fig:slope_locs_MUSIC_theta5}--\ref{fig:slope_locs_CVX_theta5} show that the fitted slopes of ESPRIT, ML, and CVX are close to the predicted resolution scaling $-3$ in the near-$0$ phase regime, whereas MUSIC shows a mild deviation from this optimal scaling.
A similar behavior is observed for MUSIC, ESPRIT, and ML when $|\theta|=35\pi/36$, as shown in Figs.~\ref{fig:slope_locs_MUSIC_theta175}--\ref{fig:slope_locs_ML_theta175}. By contrast, CVX exhibits a substantial degradation in resolution, which is due to the inherent instability in superresolving complex sources (especially when $\theta \approx \pi$) \cite{duval2015exact, tang2015resolution}. In contrast to the source-number detection problem, no anomalous slope transition is observed here as $|\theta|$ approaches $\pi$; rather, the same $-3$ scaling law persists across both representative small $\babs{\theta}_{\min}$ regimes. 
The red--green phase-transition diagrams in Figs.~\ref{fig:locs_beta3_theta5_MUSIC}--\ref{fig:locs_SRL_beta3_theta5_CVX} and Figs.~\ref{fig:locs_beta3_theta175_MUSIC}--\ref{fig:locs_SRL_beta3_theta175_CVX} show that ESPRIT and ML produce empirical boundaries closer to the SRL, whereas MUSIC is more conservative, especially in the near-$\pi$ regime. 
Therefore, Fig.~\ref{fig:small_theta_locs_case}  confirms the correctness of the proposed SRLs for the location estimation problem.

\section{Large Phase Difference $\left(\babs{\theta}_{\min}\gg \pare{\frac{\sigma}{m}}^\frac{1}{2}\right)$}
\label{section:large_theta}

In contrast to the near-endpoint regime, we next consider the case where the relative phase is sufficiently separated from both $0$ and $\pi$, namely, $\babs{\theta}_{\min}\gg \pare{\frac{\sigma}{m}}^\frac{1}{2}$.

\subsection{Theoretical Bounds}
\label{subsection:boundary_large_theta}

In this subsection, we consider the large-phase regime. The following theorems quantify the resulting phase-induced improvement for both tasks, with proofs provided in Appendix~\ref{proof:large_theta}.

\begin{thm} \label{thm:complex_number_UB_large_theta}
Let $\mathbf{Y}$ be generated by a complex measure
$\mu=me^{i\theta_1} \delta_{y_1}+\beta m e^{i\theta_2} \delta_{y_2}$ with 
$y_1, y_2\in B_{\frac{\pi}{2\Omega}}(0)$, $\theta_1, \theta_2\in \pare{-\pi,\pi}$, $\beta\geqs1$, and $m>0$.
Define $d:=\babs{y_1-y_2}$, $\theta:=\theta_1-\theta_2\in\pare{-\pi,\pi}$, and $\babs{\theta}_{\min} := \min\{\babs \theta,\pi-\babs \theta\}$.
When $\babs{\theta}_{\min} >\frac{4\pi}{3}\sqrt{1+\frac{1.25}{\beta}}\pare{\frac{\sigma}{m}}^\frac{1}{2}$, if 
\begin{align}
\label{equ:sepacondinumbercomplex_largetheta}
\begin{cases}
     \babs{y_1-y_2}\geqs\frac{3}{\Omega}\arcsin\pare{\frac{\pare{2+\frac{2.5}{\beta}}\frac{\sigma}{m}}{\sin\babs{\theta}_{\min}}}, & \babs{\theta}_{\min}<\frac{\pi}{4};\\
     \babs{y_1-y_2}\geqs\frac{4}{\Omega}\arcsin\pare{\pare{2\sqrt{2}+\frac{2.5\sqrt{2}}{\beta}}\frac{\sigma}{m}}, &\frac{\pi}{4}\leqs\babs{\theta}_{\min}\leqs\frac{\pi}{2}.
\end{cases}
\end{align}
then no $\sigma$-admissible measure for $\mathbf{Y}$ can be supported on fewer than two points. Moreover, if
\begin{align}
\label{equ:separation_condition_num_complex_LB_large}
\babs{y_1-y_2}  <  \frac{2}{\Omega}\arcsin \pare{\frac{\sigma}{m}},
\end{align}
then there exists a $\sigma$-admissible measure $\widehat \mu$ for $\mathbf{Y}$ supported on a single point.
\end{thm}

\begin{thm}
\label{thm:supportrecoveryupperboundcomplex_large_theta}
Let $\mathbf{Y}$ be generated by a complex measure
$\mu=me^{i\theta_1} \delta_{y_1}+\beta m e^{i\theta_2} \delta_{y_2}$ with 
$y_1, y_2\in B_{\frac{\pi}{2\Omega}}(0)$, $\theta_1, \theta_2\in \pare{-\pi,\pi}$, $\beta\geqs1$, and $m>0$.
Define $d:=\babs{y_1-y_2}$, $\theta:=\theta_1-\theta_2\in\pare{-\pi,\pi}$, and $\babs{\theta}_{\min} := \min\{\babs \theta,\pi-\babs \theta\}$.
When $\babs{\theta}_{\min} >1.75\pi\pare{\frac{\sigma}{m}}^\frac{1}{3}$, if 
\begin{align}
\label{equ:sepacondilocationcomplex_large}
\begin{cases}
         d\geqs\frac{6}{\Omega}\arcsin\pare{\pare{\frac{\frac{2\sigma}{ m}}{\sin\babs{\theta}_{\min}}}^\frac{1}{2}}, & 1.75\pi\pare{\frac{\sigma}{m}}^\frac{1}{3}<\babs{\theta}_{\min}<\frac{\pi}{6} ;\\
         d\geqs\frac{8}{\Omega}\arcsin\pare{\pare{\frac{\frac{2\sigma}{ m}}{\sin\pare{\babs{\theta}_{\min}+\frac{\pi}{3}}}}^\frac{1}{2}}, &  \frac{\pi}{6}\leqs\babs{\theta}_{\min} \leqs\frac{\pi}{2},
    \end{cases}
\end{align}
and $\widehat{\mu} = \widehat{a}_1 \delta_{\widehat{y}_1} + \widehat{a}_2 \delta_{\widehat{y}_2}$ supported on $B_{\frac{\pi}{2\Omega}}(0)$ is a $\sigma$-admissible complex measure for $\mathbf{Y}$, then $\widehat{\mu}$ lies within the $\frac{d}{2}$-neighborhood of $\mu$. Moreover, if 
\begin{align}
\label{equ:separation_condition_locs_complex_LB_large}
\babs{y_1-y_2}  <  \frac{2}{\Omega}\arcsin \pare{\pare{\frac{\sigma}{m}}^\frac{1}{2}},
\end{align}
then there exists a $\sigma$-admissible measure $\widehat \mu=\widehat{a}_1 \delta_{\widehat{y}_1} + \widehat{a}_2 \delta_{\widehat{y}_2}$ that does not lie within the $\frac{d}{2}$-neighborhood of $\mu$.
\end{thm}

The preceding results indicate that the large-phase regime yields an order-level improvement in the resolution bounds: from $\pare{\sigma/m}^{1/2}$ to $\sigma/m$ for source-number detection, and from $\pare{\sigma/m}^{1/3}$ to $\pare{\sigma/m}^{1/2}$ for location estimation. Thus, unlike the near-endpoint regime where the phase term acts mainly as a subtractive correction, a sufficiently large relative phase changes the resolution order itself.

\subsection{Numerical Experiments on Source-Number Detection} \label{experiment:large_theta_num}

In this subsection, following the same experimental framework as in Section~\ref{experiment:zero_theta_num}, we choose the $\ell_0$ and SVT algorithms to validate the theoretical SRLs for source-number detection established in Sections~\ref{subsection:boundary_large_theta}, and validate the slopes in $\log(\mathrm{SRF})$--$\log(\sigma)$ plane.

\begin{figure}[tbp]
\centering
\subfloat[$\ell_0$ \centering]{\includegraphics[height=3.45cm]{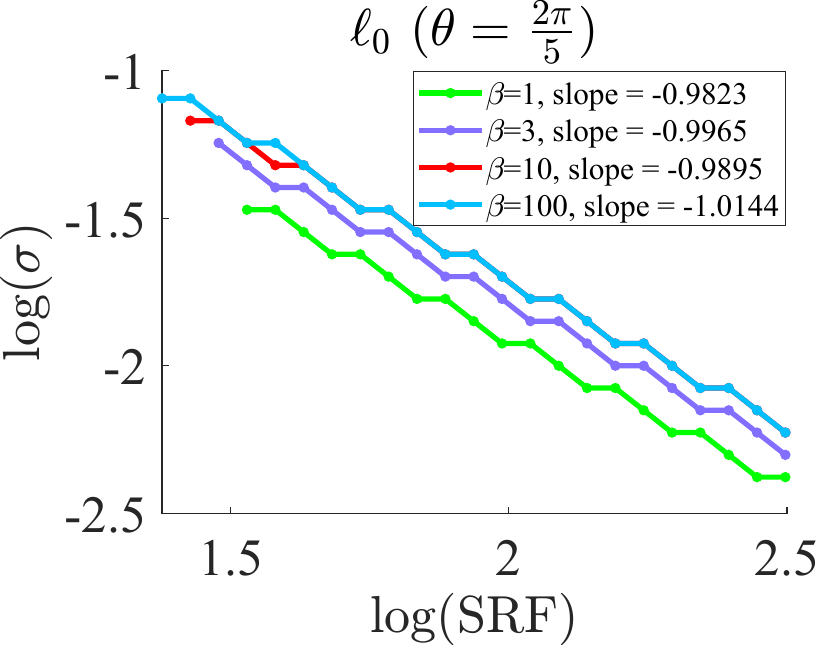}\label{fig:slope_num_L0_theta72}}   \hfil
\subfloat[$\ell_0$ \centering]{\includegraphics[height=3.45cm]{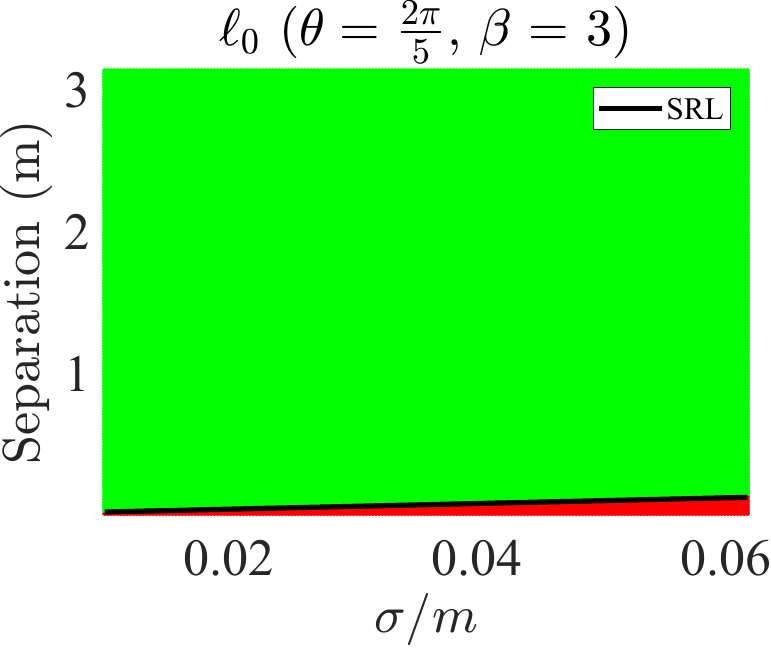}\label{fig:num_SRL_beta3_theta72_L0_only}}  \\

\subfloat[SVT \centering]{\includegraphics[height=3.45cm]{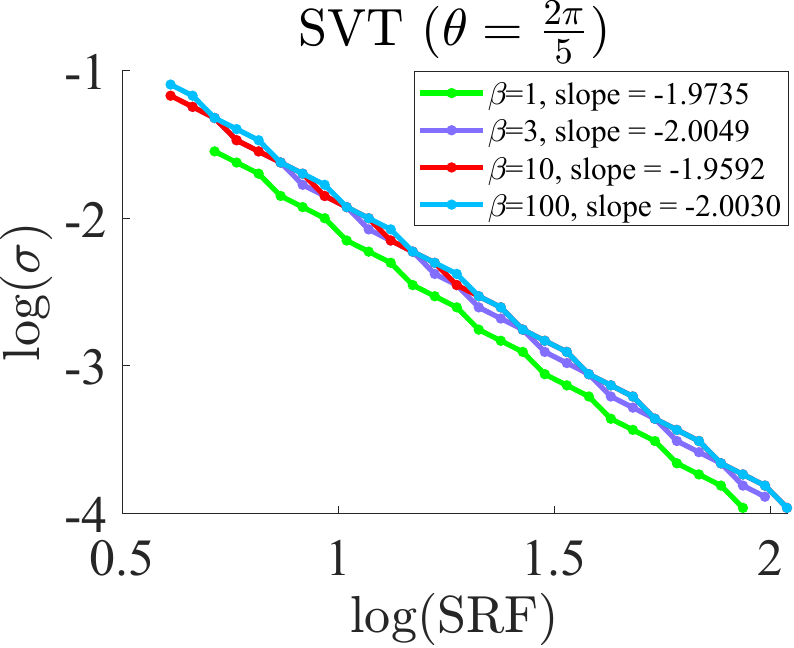}\label{fig:slope_num_Threshold_theta72}} \hfil 
\subfloat[SVT \centering]{\includegraphics[height=3.45cm]{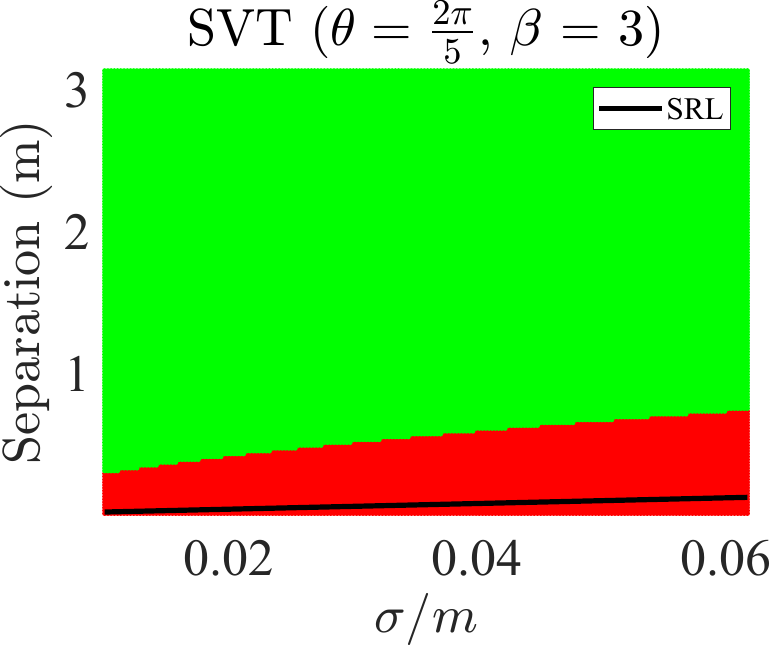}\label{fig:num_SRL_beta3_theta72_MUSIC_only}}  

\caption{ Stable performance of each algorithm for source-number detection under the large-phase regime.  } 
\label{fig:large_theta_num_case}
\end{figure}

Fig.~\ref{fig:large_theta_num_case} examines the source-number detection problem in the large-phase regime. 
From the SRL in \eqref{equ:separation_condition_num_complex_LB_large}, we have
\begin{align*}
    \log(\sigma)>
    \log\!\left[\sin\!\left(\frac{\pi}{2\,\mathrm{SRF}}\right)\right]
    +\log(m),
\end{align*}
which predicts a slope close to $-1$ in the $\log(\mathrm{SRF})$--$\log(\sigma)$ plane for large $\mathrm{SRF}$. 
Here, we choose $\theta=2\pi/5=72^\circ$ as a representative large relative phase. 
The noise levels are restricted to the range satisfying the large-phase condition $\babs{\theta}_{\min} \gg \pare{\frac{\sigma}{m}}^\frac{1}{2}$, and hence the experiment is conducted in a relatively high-SNR regime.

As shown in Fig.~\ref{fig:slope_num_L0_theta72}, the fitted slopes of the $\ell_0$ method are close to $-1$ for all tested values of $\beta$, indicating that it attains the predicted large-phase resolution order. 
The corresponding red--green phase-transition diagram in Fig.~\ref{fig:num_SRL_beta3_theta72_L0_only} further shows that, for $\beta=3$, the empirical boundary remains above but very close to the SRL. 
By contrast, Fig.~\ref{fig:slope_num_Threshold_theta72} shows that the slopes of SVT remain close to $-2$, rather than the predicted value $-1$. 
Hence, SVT does not capture the order-level improvement induced by the large relative phase. 
This is also reflected in Fig.~\ref{fig:num_SRL_beta3_theta72_MUSIC_only}, where the empirical boundary of SVT stays farther away from the SRL. 
Therefore, Fig.~\ref{fig:large_theta_num_case} shows that the large-phase improvement is realized by the $\ell_0$ method in this regime.

% Fig.~\ref{fig:large_theta_num_case} validates the lower bounds in \eqref{equ:separation_condition_num_complex_LB_large} for the source-number detection problem under the large relative phase regimes. In particular, \eqref{equ:separation_condition_num_complex_LB_large} implies
% \begin{align*}
%     \mathrm{log}(\sigma)>\mathrm{log} \bracket{ \sin\pare{ \frac{\pi}{2\,\mathrm{PASRF}}}}+ \mathrm{log}\pare{m},
% \end{align*}
% which predicts a slope close to -1 in the $\mathrm{log}\pare{\mathrm{PASRF}}$-$\mathrm{log}\pare{\sigma}$ plane. We choose $\babs{\theta} = 35\pi/36$ as the representative large relative phase and select the $\sigma$ which are satisfy the condition $\babs{\theta}_{\min} >\frac{4\pi}{3}\sqrt{1+\frac{1.25}{\beta}}\pare{\frac{\sigma}{m}}^\frac{1}{2}$ in Theorem~\ref{thm:complex_number_UB_large_theta}. Fig.~\ref{fig:slope_num_L0_theta72}  shows that the boundary of success-failure transition diagrams for $\ell_0$ method under different $\beta$ are close to -1, consistent with our prediction; Fig.~\ref{fig:num_SRL_beta3_theta72_L0_only} further manifest that empirical resolution do not surpass the CRL lower bound. 
% However, the slopes in Fig.~\ref{fig:slope_num_Threshold_theta72} are close to -2, which proves that in this case, the empirical resolution of the threshold method still has a certain gap compared to the empirical resolution. Fig.~\ref{fig:num_SRL_beta3_theta72_MUSIC_only} further confirms this conclusion, whose empirical resolution is still far away from SRL.

\subsection{Numerical Experiments on Location Estimation}

In this subsection, similarly as the experimental framework in Section~\ref{experiment:zero_theta_locs}, we choose the MUSIC, ESPRIT, ML, and CVX algorithms to validate the theoretical lower bounds for location estimation established in Sections~\ref{subsection:boundary_large_theta},  and validate the slopes in $\log(\mathrm{SRF})$--$\log(\sigma)$ plane.

\begin{figure*}[htbp]
\centering
\subfloat[MUSIC \centering]{\includegraphics[height=3.45cm]{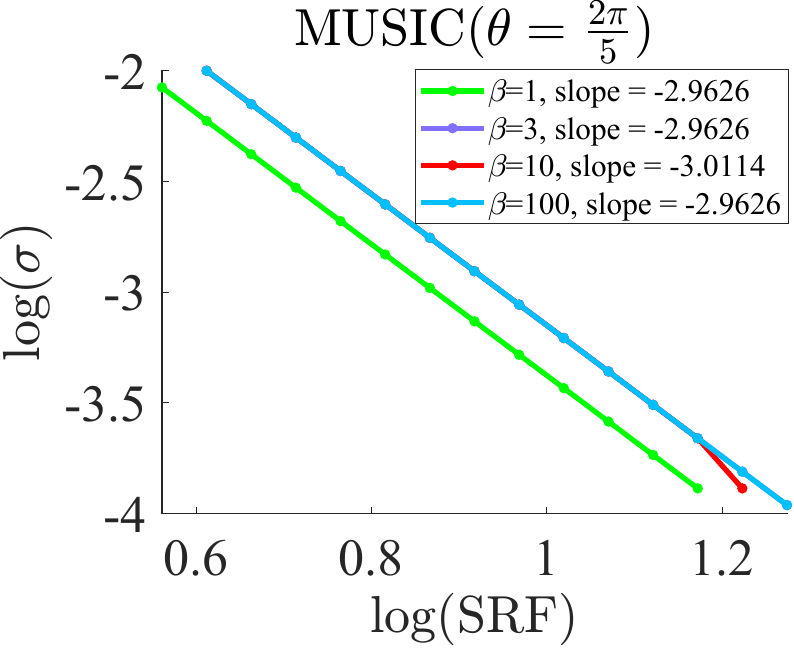}\label{fig:slope_locs_MUSIC_theta72}}  \hfil
\subfloat[ESPRIT \centering]{\includegraphics[height=3.45cm]{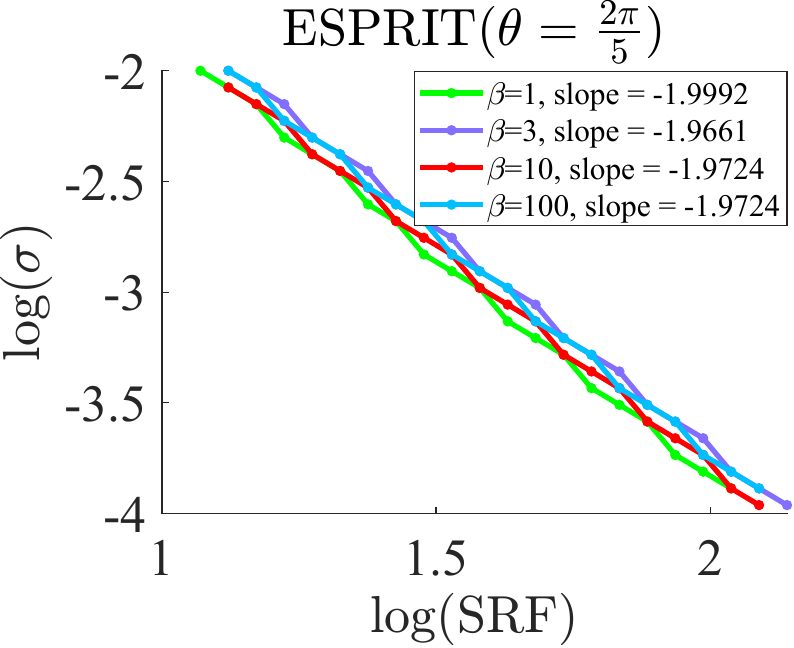}\label{fig:slope_locs_ESPRIT_theta72}}  \hfil
\subfloat[ML \centering]{\includegraphics[height=3.45cm]{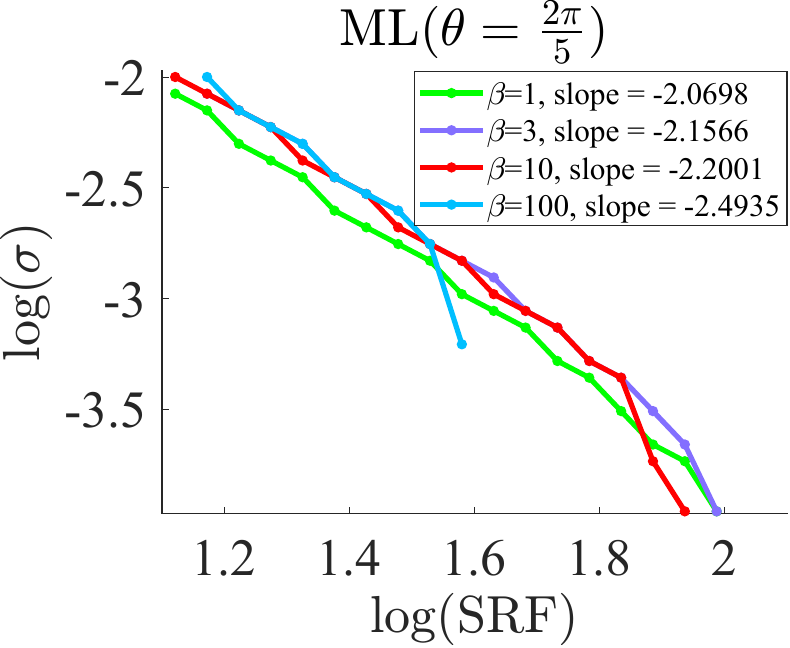}\label{fig:slope_locs_ML_theta72}}  \hfil
\subfloat[CVX \centering]{\includegraphics[height=3.45cm]{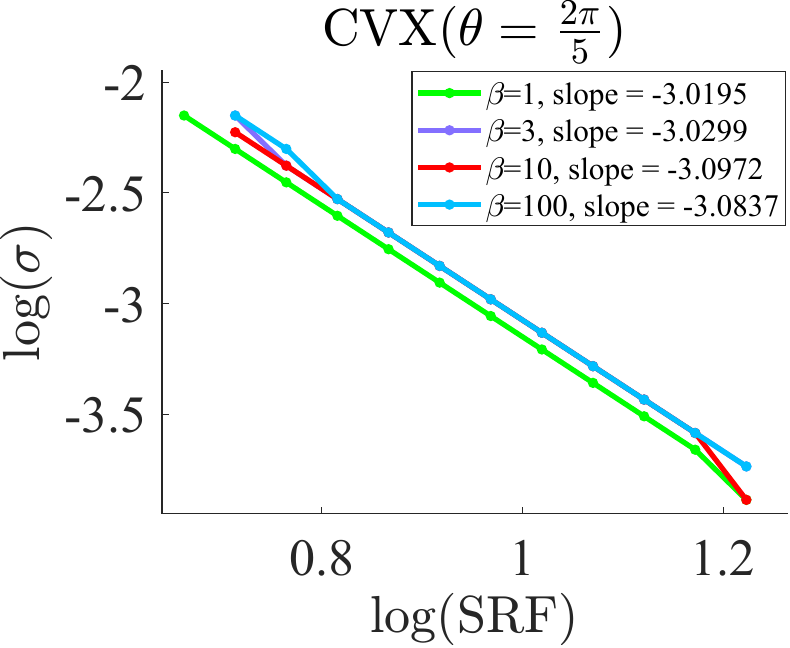}\label{fig:slope_locs_CVX_theta72}}  \\

\subfloat[MUSIC \centering]{\includegraphics[height=3.45cm]{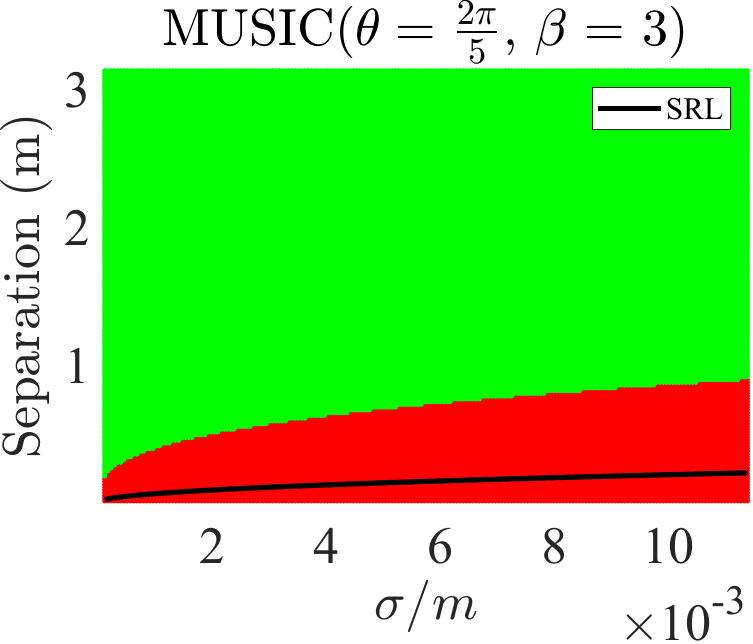}\label{fig:locs_beta3_theta72_MUSIC}}  \hfil
\subfloat[ESPRIT \centering]{\includegraphics[height=3.45cm]{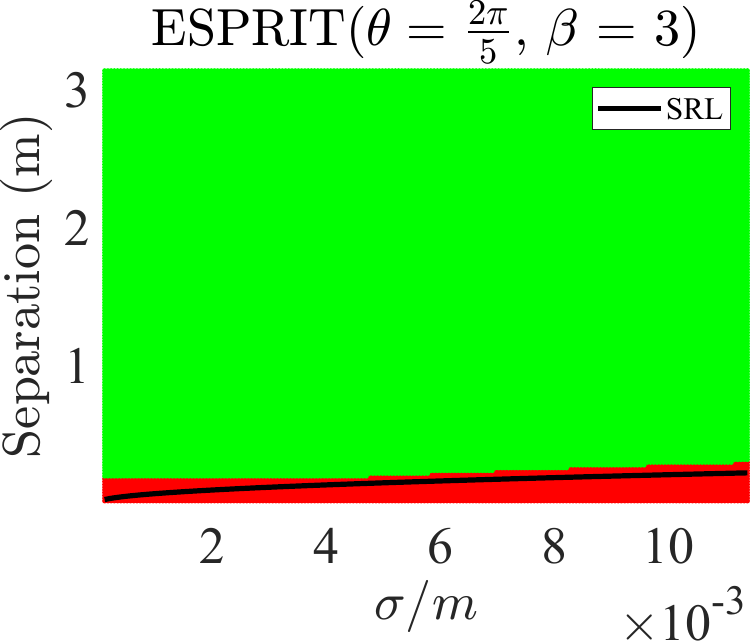}\label{fig:locs_beta3_theta72_ESPRIT}}  \hfil
\subfloat[ML \centering]{\includegraphics[height=3.45cm]{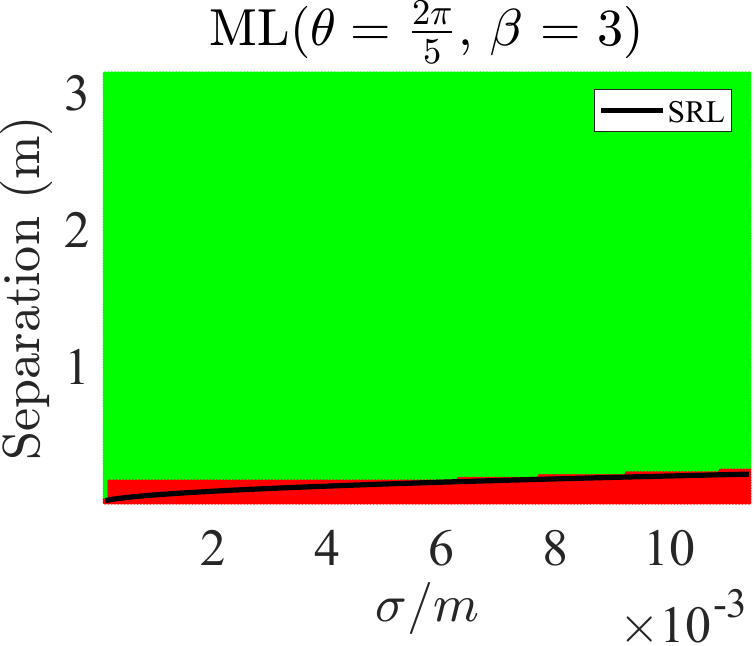}\label{fig:locs_SRL_beta3_theta72_ML}}  \hfil
\subfloat[CVX \centering]{\includegraphics[height=3.45cm]{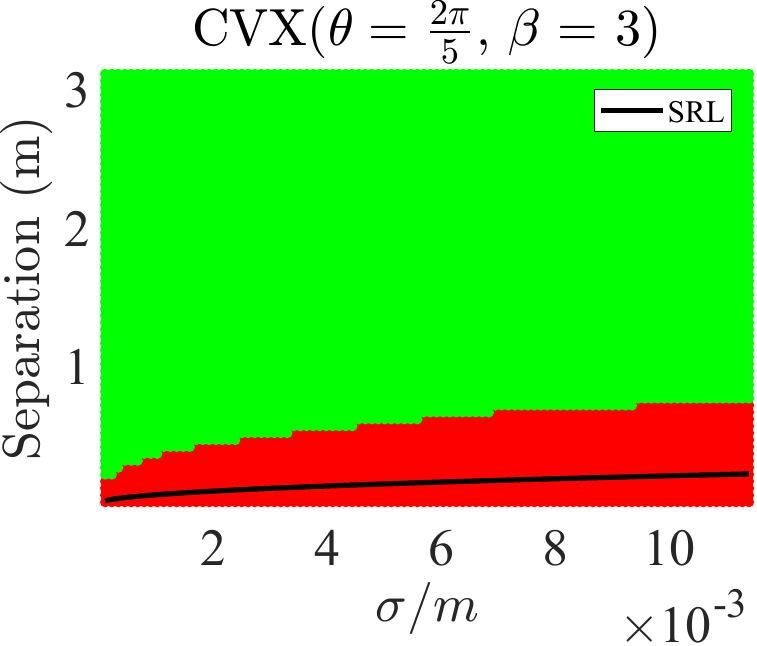}\label{fig:locs_SRL_beta3_theta72_CVX}} 
\caption{ Stable performance of each algorithm for location estimation under the large-phase regime.  } 
\label{fig:large_theta_locs_case}
\end{figure*}

% Fig.~\ref{fig:large_theta_locs_case} validates the lower bounds in \eqref{equ:separation_condition_locs_complex_LB_large} for the location-estimation problem under the large relative phase regimes. In particular, \eqref{equ:separation_condition_locs_complex_LB_large} implies
% \begin{align*}
%     \mathrm{log}(\sigma)>2\,\mathrm{log} \bracket{ \sin\pare{ \frac{\pi}{2\,\mathrm{PASRF}}}}+ \mathrm{log}\pare{m},
% \end{align*}
% which predicts a slope close to -2 in the $\mathrm{log}\pare{\mathrm{PASRF}}$-$\mathrm{log}\pare{\sigma}$ plane. We also choose $\babs{\theta} = 35\pi/36$ as the representative large relative phase and select the $\sigma$ which are satisfy the condition $\babs{\theta}_{\min} >1.75\pi\pare{\frac{\sigma}{m}}^\frac{1}{3}$ in Theorem~\ref{thm:supportrecoveryupperboundcomplex_large_theta}. 
% Fig.~\ref{fig:slope_locs_MUSIC_theta72} shows the slope of MUSIC under different $\beta$ are close to -3, meaning a gap between its resolution and optimal resolution, Fig.\ref{fig:locs_beta3_theta72_MUSIC} further demonstrate this gap.
% Fig.~\ref{fig:slope_locs_ESPRIT_theta72}-\ref{fig:slope_locs_CVX_theta72} shows the slope of ESPRIT, ML and CVX are close to -2, Fig.~\ref{fig:locs_beta3_theta72_ESPRIT}-\ref{fig:locs_SRL_beta3_theta72_CVX} show that their empirical resolutions are still above the CRL lower bound, further confirming the correctness of CRL lower bound itself.

Fig.~\ref{fig:large_theta_locs_case} examines the location estimation problem in the large-phase regime. 
From the SRL in \eqref{equ:separation_condition_locs_complex_LB_large}, we have
\begin{align*}
    \log(\sigma)>
    2\log\!\left[\sin\!\left(\frac{\pi}{2\,\mathrm{SRF}}\right)\right]
    +\log(m),
\end{align*}
which predicts a slope close to $-2$ in the $\log(\mathrm{SRF})$--$\log(\sigma)$ plane for large $\mathrm{SRF}$. 
We also choose $\theta=2\pi/5=72^\circ$ as a representative large relative phase, and restrict the noise levels to the range satisfying the large-phase condition $\babs{\theta}_{\min} \gg \pare{\frac{\sigma}{m}}^\frac{1}{3}$. Similarly, the experiment is conducted in a high-SNR regime.

The slope plots show a clear algorithm-dependent behavior. As shown in Fig.~\ref{fig:slope_locs_MUSIC_theta72}, the fitted slopes of MUSIC remain close to $-3$, indicating that MUSIC also does not capture the order-level improvement predicted in the large-phase regime.  This gap is also reflected in Fig.~\ref{fig:locs_beta3_theta72_MUSIC}, where the empirical boundary stays relatively far above the SRL. 
By contrast, Figs.~\ref{fig:slope_locs_ESPRIT_theta72} and \ref{fig:slope_locs_ML_theta72} show that ESPRIT and ML produce slopes closer to the predicted value $-2$. 
Their red--green phase-transition diagrams in Figs.~\ref{fig:locs_beta3_theta72_ESPRIT} and \ref{fig:locs_SRL_beta3_theta72_ML} also lie closer to the SRL than that of MUSIC. 
Although the measured slopes in Fig.~\ref{fig:slope_locs_CVX_theta72} are roughly close to $-2$, the boundary curves are visibly non-linear and fluctuate across the tested range, suggesting that CVX is less stable for complex-valued sources in this regime.
Fig.~\ref{fig:locs_SRL_beta3_theta72_CVX} further demonstrates this instability with a rough success-failure boundary.

Taken together with the in-phase and near-endpoint experiments, these results indicate that MUSIC does not fully exploit the phase information contained in the complex amplitudes and therefore exhibits relatively conservative behavior. 
By contrast, ESPRIT and ML both make effective use of the complex-valued structure of the measurements, but in different ways. 
ESPRIT exploits the shift-invariance structure, while ML directly fits the full complex-valued measurement model by jointly estimating the locations and complex amplitudes.
Consequently, both methods attain slopes close to the optimal resolution order in the large-phase regime and yield empirical boundaries close to the predicted SRL. 
CVX, in comparison, is less robust for complex-valued sources. 
Hence, among the tested algorithms, ESPRIT and ML provide the most favorable performance in the considered location-estimation experiments.

\section{Conclusion}\label{section:conclusion}

This paper has established a CRL characterization for the two-point super-resolution problem. 
Our analysis shows that, in the complex-amplitude setting, the amplitude ratio and, more importantly, the relative phase can substantially affect resolvability. 
For both source-number detection and location estimation, we derived explicit SRUs and SRLs under different phase regimes, leading to a phase-dependent characterization of stable resolution. 
In the in-phase and near-endpoint phase regimes, the resolution limits obey the same exponent laws as in the classical setting. 
By contrast, in the large-phase regime, the resolution exponents are strictly improved, showing that a sufficiently separated relative phase can enhance the intrinsic stable resolving power of the measurement model. 
In addition, for source-number detection, we identified a critical near-$\pi$ phenomenon in the balanced-amplitude case $\beta=1$, where the scaling law differs from the generic near-endpoint behavior.

The numerical experiments are consistent with the theoretical predictions. 
Across the tested algorithms and parameter regimes, the measured slopes in the $\log(\mathrm{SRF})$/$\log(\mathrm{PASRF})$--$\log(\sigma)$ plane agree with the predicted scaling laws. 
The red--green success--failure diagrams further quantify the distance between each empirical resolution boundary and the corresponding SRL, thereby providing a quantitative benchmark for algorithm selection. 
These results demonstrate that phase information is not merely a nuisance parameter in complex-valued super-resolution, but a key factor that can be exploited to improve stable resolvability.

\section*{Acknowledgments}
The work of Xiaole He and Junglin Wang was supported by the National Natural Science Foundation of China under Grant 62571038 and Grant 62071041. The work of Ping Liu was partially supported by the National Key R\&D Program of China grant number 2024YFA1016000 and the Fundamental Research Funds for the Central Universities grant number 226-2025-00192. 

{\appendices
\section{Preliminaries }

In this section, we collect several preliminary results needed for the subsequent proofs. In particular, we recall the location--amplitude identities and inequalities from \cite{liu2024mathematical}.

\begin{lem}[Location-Amplitude Identities] \label{lem:loc-ampidentity} 
Consider the model
\begin{align*}
\mathscr{F}\bracket{\widehat{\mu}}\pare{\omega}
= \mathscr{F}\bracket{\mu}\pare{\omega} + \mathbf{w}\pare{\omega},
\quad \omega \in \bracket{0,\Omega},
\end{align*}
where $\widehat{\mu}=\sum_{j=1}^{d}\widehat{a}_j\delta_{\widehat{y}_j}$ and
$\mu=\sum_{j=1}^{n}a_j\delta_{y_j}$.
For any fixed $y_t$ and $\widehat{y}_{t'}$, define the set $S_t$ containing all $y_j$'s and
$\widehat{y}_j$'s except $y_t,\widehat{y}_{t'}$ that
\begin{align*}
S_t := \left\{\, y_1,\cdots,y_{t-1},y_{t+1},\cdots,y_n,\,
\widehat{y}_1,\cdots,\widehat{y}_{t'-1},\widehat{y}_{t'+1},\cdots,\widehat{y}_d \,\right\}.
\end{align*}
Let $\#S_t$ be the number of elements in $S_t$ (i.e., $n+d-2$).
Then, for any $0<\omega^*\leqs \frac{\Omega}{\#S_t}$, we have the following relations:
\begin{align}
\widehat{a}_{t'}\prod_{q\in S_t}\pare{e^{i\widehat{y}_{t'}\omega^*}-e^{iq\omega^*}}
- a_t\prod_{q\in S_t}\pare{e^{iy_t\omega^*}-e^{iq\omega^*}}
= \mathbf{w}_1^{\top}\mathbf{v}.
\label{equ:LA_7}
\end{align}
Moreover, for any $0<\omega^*\leqs \frac{\Omega}{\#S_t+1}$, we have
\begin{align}
a_t\prod_{q\in S_t\cup\left\{\widehat{y}_{t'}\right\}}
\pare{e^{iy_t\omega^*}-e^{iq\omega^*}}
=
\pare{e^{i\widehat{y}_{t'}\omega^*}\mathbf{w}_1-\mathbf{w}_2}^{\top}\mathbf{v}.
\label{equ:LA_8}
\end{align}
Here, 
\begin{align*}
&\mathbf{w}_1 = \pare{\mathbf{w}\pare{0},\,\mathbf{w}\pare{\omega^*},\,\cdots,\,\mathbf{w}\pare{\#S_t\omega^*}}^{\top},\\
&\mathbf{w}_2 = (\mathbf{w}\pare{1}, \,\mathbf{w}\pare{\omega^*}, \, \cdots, \,\mathbf{w}\pare{\pare{\#S_t+1}\omega^*})^{\top},
\end{align*}
and the vector $\mathbf{v}$ is given by
\begin{align*}
\mathbf{v} =
\begin{pmatrix}
(-1)^{\#S} \sum_{(q_1,\ldots,q_{\#S}) \in S_{\#S}}
  e^{iq_1\omega^{\ast}}\cdots e^{iq_{\#S}\omega^{\ast}} \\
(-1)^{\#S-1} \sum_{(q_1,\ldots,q_{\#S-1}) \in S_{\#S-1}}
  e^{iq_1\omega^{\ast}}\cdots e^{iq_{\#S-1}\omega^{\ast}} \\
\vdots \\
(-1)\sum_{(q_1)\in S_1} e^{iq_1\omega^{\ast}} \\
1
\end{pmatrix},
\end{align*}
where $S_{t,p} := \left\{ \left\{q_1,\cdots,q_p\right\} \,\middle|\, q_j\in S_t,\ 1\leqs j\leqs p,\  q_{j'}\right.$ and $ q_{j''} $ are different elements in  $S_t$ when   $\left. j'\neq j'' \right\}$,  $p=1,\cdots,\#S_t$.
\end{lem}

\begin{lem}[Location--Amplitude Inequalities]
\label{lem:loc-ampinequality}
Consider the model
\begin{align*}
\mathscr{F}\bracket{\widehat{\mu}}\pare{\omega}
= \mathscr{F}\bracket{\mu}\pare{\omega} + \mathbf w\pare{\omega},
\quad \omega \in \bracket{0,\Omega},
\end{align*}
where $\widehat{\mu}=\sum_{j=1}^{d}\widehat{a}_j\delta_{\widehat{y}_j}$ and
$\mu=\sum_{j=1}^{n}a_j\delta_{y_j}$, and assume that
$\babs{\mathbf w\pare{\omega}}<\sigma,\ \omega\in\bracket{0,\Omega}$.
For any fixed $y_t$ and $\widehat{y}_{t'}$, define the set $S_t$ as
\begin{align*}
S_t := \left\{\, y_1,\cdots,y_{t-1},y_{t+1},\cdots,y_n,\,
\widehat{y}_1,\cdots,\widehat{y}_{t'-1},\widehat{y}_{t'+1},\cdots,\widehat{y}_d \,\right\}.
\end{align*}
Let $\#S_t$ be the number of elements in $S_t$ (i.e., $n+d-2$). Then, for any $0<\omega^*\leqs \frac{\Omega}{\#S_t}$, we have
\begin{align}
\left|
\widehat{a}_{t'}\prod_{q\in S_t}\pare{e^{i\widehat{y}_{t'}\omega^*}-e^{iq\omega^*}}
- a_t\prod_{q\in S_t}\pare{e^{iy_t\omega^*}-e^{iq\omega^*}}
\right| < 2^{\#S_t}\sigma.
\label{equ:cor_10}
\end{align}
Moreover, for any $0<\omega^*\leqs \frac{\Omega}{\#S_t+1}$, we have
\begin{align}
\left|
a_t\prod_{q\in S_t\cup\left\{\widehat{y}_{t'}\right\}}
\pare{e^{iy_t\omega^*}-e^{iq\omega^*}}
\right|
< 2^{\#S_t+1}\sigma.
\label{equ:cor_11}
\end{align}
\end{lem}

% The detailed derivations for Lemma~\ref{lem:loc-ampidentity} and \ref{lem:loc-ampinequality} are shown in \cite{liu2024mathematical}.

\begin{lem}\label{lem:sincosinequality}
For $x\geqs 0$ and $k\in \mathbb N$, it holds that 
\begin{align}
&\sin x \geqs \sum_{j=0}^{2k-1} \frac{(-1)^{j}}{(2j+1)!} x^{2j+1}; \label{equ:sincosineq1}\\
&\sin x \leqs \sum_{j=0}^{2k} \frac{(-1)^{j}}{(2j+1)!} x^{2j+1};\label{equ:sincosineq2}\\
&\cos x \geqs \sum_{j=0}^{2k-1} \frac{(-1)^{j}}{(2j)!} x^{2j};\label{equ:sincosineq3}\\
&\cos x \leqs \sum_{j=0}^{2k} \frac{(-1)^{j}}{(2j)!} x^{2j}. \label{equ:sincosineq4}
\end{align}
\end{lem}
\begin{proof}
First, we have $\sin x \leqs x, \cos x \leqs 1$. Now suppose (\ref{equ:sincosineq2}) holds for $k=s$, then by  
\begin{align*}
    &\int_{0}^{x}\sin t dt = -\cos t|_{0}^x = 1-\cos x \\ \leqs & \int_{0}^{x} \sum_{j=0}^{2s} \frac{(-1)^{j}}{(2j+1)!} t^{2j+1} dt = \sum_{j=0}^{2s} \frac{(-1)^{j}}{(2j+2)!} x^{2j+2},
\end{align*}
which shows that (\ref{equ:sincosineq3}) holds for $k=s+1$. In the same manner, one can prove the four inequalities by mathematical induction. 
\end{proof}

\begin{remark}
    Applying Lemma~\ref{lem:sincosinequality} on $[0,\pi/2]$ yields the following tight bound:
\begin{align}
    x-\frac{x^3}{6}\leqs \sin x \leqs x,\quad
1-\frac{x^2}{2}\leqs \cos x \leqs 1.
\label{equ:scale_sin_cos}
\end{align}
\end{remark}

\section{Proofs of Results in Section \ref{section:zero_theta}}\label{proof:zero_theta}

\subsection{Proof of Theorem~\ref{thm:numberresolutionpositive}}\label{proof:num_theta0_UB}

First, we shall prove \eqref{equ:num_positive_ub_cond}.
Given the measurement $\mathbf{Y}$ generated by $\mu= m \delta_{y_1}+\beta  m \delta_{y_2}$ with $y_1,y_2\in B_{\frac{\pi}{2\Omega}}(0)$, $\beta\geqs 1$, and $m>0$. Suppose that $\widehat{\mu}=\widehat{a}\,\delta_{\widehat{y}}$ is a $\sigma$-admissible measure of\, $\mathbf{Y}$. The Definition~\ref{def:sigma_admissible} and the model \eqref{equ:modelsetting0} imply  that $\mu$ and $\widehat \mu$ satisfy
\begin{align}
\mathscr{F}[\widehat{\mu}](\omega)=\mathscr{F}[\mu](\omega)+\mathbf{W}_1(\omega),
\quad \omega\in[0,2\Omega],
\label{proof-equ:imagemodeling2sigma}
\end{align}
for some $\mathbf{W}_1$
with $\babs{\mathbf{W}_1}<2\sigma$, $\omega\in[0,2\Omega]$. Define $S_1$ and $S_2$ as, respectively,
\begin{align*}
    S_1=\{y_2\},\quad S_2 = \{y_1\}.
\end{align*}
Then $\# S_1 = \# S_2=1$. For any $0<\omega^\ast\leqs \frac{2\Omega}{\# S_t+1}=\Omega$, applying \eqref{equ:cor_11} to \eqref{proof-equ:imagemodeling2sigma} yields
\begin{align}
\babs{\,(e^{i\omega^\ast y_1}-e^{i\omega^\ast y_2})\,(e^{i\omega^\ast y_1}-e^{i\omega^\ast \widehat{y}})}
&<\frac{8\sigma}{m},
\label{equ:ineq_t1}\\
\babs{\,(e^{i\omega^\ast y_2}-e^{i\omega^\ast y_1})\,(e^{i\omega^\ast y_2}-e^{i\omega^\ast \widehat{y}})}
&<\frac{8\sigma}{\beta m}.
\label{equ:ineq_t2}
\end{align}
Define $d:= \babs{y_1-y_2}\in(0,\pi/\Omega)$ and
\begin{align}
\Delta \triangleq \babs{e^{i\omega^\ast y_1}-e^{i\omega^\ast y_2}}=2\sin\!\left(\frac{d \omega^\ast}{2}\right).
\end{align}
Summing \eqref{equ:ineq_t1}--\eqref{equ:ineq_t2} yields
\begin{align}
\Delta\pare{\babs{e^{i\omega^\ast y_1}-e^{i\omega^\ast \widehat{y}}}
+\babs{e^{i\omega^\ast y_2}-e^{i\omega^\ast \widehat{y}}}}
<\frac{8(\beta+1)}{\beta}\frac{\sigma}{m}.
\label{equ:sum_delta}
\end{align}
% Notably, when \eqref{equ:sum_delta} holds, inequalities \eqref{equ:ineq_t1} and \eqref{equ:ineq_t2} can be satisfied simultaneously.
By the triangle inequality,
\begin{align}
\babs{e^{i\omega^\ast y_1}-e^{i\omega^\ast y_2}}
\leqs
\babs{e^{i\omega^\ast y_1}-e^{i\omega^\ast \widehat{y}}}
+
\babs{e^{i\omega^\ast \widehat{y}}-e^{i\omega^\ast y_2}},
\end{align}
the term in parentheses \eqref{equ:sum_delta} is lower bounded by $\Delta$. Therefore,
\begin{align}
\Delta^2 = 4\sin^2\!\left(\frac{d\omega^\ast}{2}\right)
<\frac{8(\beta+1)}{\beta}\frac{\sigma}{m}.
\end{align}
Since $0<\omega^\ast\leqs \Omega$ and  $d\leqs\frac{\pi}{\Omega}$, we obtain the following constraint on the separation  that ensures \eqref{equ:sum_delta} holds
\begin{align}
|y_1-y_2|
<\frac{2}{\Omega}\arcsin\!\left(\left(\frac{2(\beta+1)\sigma}{\beta m}\right)^{\frac{1}{2}}\right)
\label{equ:d_minofnumdec_upperbound}
\end{align}
under the condition $\frac{\sigma}{m}< \frac{\beta}{2(\beta+1)}$.
This contradicts the separation condition in \eqref{equ:num_positive_ub_cond}. Hence no positive
$\sigma$-admissible measure supported on one point can exist whenever \eqref{equ:num_positive_ub_cond} holds.

Next, we shall prove \eqref{equ:num_positive_lb_cond}.
By translation invariance of the measurement model \eqref{proof-equ:imagemodeling2sigma}, we shift the coordinate system so that $y_1=-d/2$ and $y_2=d/2$. 
In this setting, the measurement $\mathbf{Y}$ is generated by
\begin{align*}
    \mu = m\delta_{-\frac{d}{2}}+\beta m\delta_{\frac{d}{2}},\quad d\in\pare{0,\frac{\pi}{\Omega}}, \   \beta\geqs1, \ m>0,
\end{align*}
and we consider the one-support positive measure 
\begin{align*}
    \widehat{\mu}=(\beta+1)m\,\delta_{\frac{\beta-1}{\beta+1}\frac{d}{2}}.
\end{align*}
Therefore, consider the identity
\begin{align*}
e^{-iu}-1=-iu\int_{0}^{1} e^{-iut}\,dt,\quad u\in\mathbb{R},
\end{align*}
and the inequality
\begin{align*}
\babs{e^{iu}-e^{iv}}
=2\babs{\sin\!\pare{\frac{u-v}{2}}}
\leqs |u-v|,
\quad u,v\in\mathbb{R},
\end{align*}
for $|\omega|\leqs \Omega$, we obtain
\begin{align}
& \babs{\mathscr{F}[\widehat{\mu}](\omega)-\mathscr{F}[\mu](\omega)} \notag \\
% =&\babs{(\beta+1)m e^{i\frac{\beta-1}{\beta+1}\frac{d}{2}\omega}-m e^{-i\frac{d}{2}\omega}-\beta m e^{i\frac{d}{2}\omega} }\notag\\
% =&m\babs{e^{-i2\frac{d}{2}\omega}-(\beta+1)e^{-i\frac{d\omega}{\beta+1}}+\beta}\notag\\
=& m\babs{\pare{e^{-id\omega}-1}-(\beta+1)\pare{e^{-i\frac{d\omega}{\beta+1}}-1}}\notag \\
=& m\cdot \babs{d\omega}\int_{0}^{1}
\babs{e^{-id\omega t}-e^{-i\frac{d\omega}{\beta+1}t}}\,dt \notag\\
\leqs & m\cdot \babs{d\omega}\int_{0}^{1}
\babs{d\omega t-\frac{d\omega}{\beta+1}t}\,dt  =\frac{\beta m}{2(\beta+1)}d^2\omega^2 .
\label{equ:delta_form_refined}
\end{align}
Consequently, if
\begin{align*}
\babs{y_1-y_2}  < \frac{2}{\Omega}\pare{\frac{\beta+1}{\beta}\frac{\sigma}{m}}^{\frac{1}{2}},
\end{align*}
then $\babs{\mathscr{F}[\widehat{\mu}](\omega)-\mathscr{F}[\mu](\omega)}< 2\sigma$ for $\omega\in [-\Omega, \Omega]$. Hence, there exists a positive $\sigma$-admissible measure of a certain measurement $\mathbf{Y}$ with only one support whenever \eqref{equ:num_positive_lb_cond} holds.

\subsection{Proof of Theorem~\ref{thm:supportrecoveryupperboundpositive}}\label{proof:locs_theta0_UB}
First, we shall prove \eqref{equ:separationforlocationrecovery-thm}. Given the measurement $\mathbf{Y}$ generated by  $\mu = m\delta_{y_1}+\beta m\delta_{y_2}$ with $y_1,y_2\in B_{\frac{\pi}{2\Omega}}(0)$, $\beta\geqs 1$, and $m>0$.
Suppose that $\widehat{\mu}=\widehat{a}_1\delta_{\widehat{y}_1}+\widehat{a}_2\delta_{\widehat{y}_2}$ is a positive $\sigma\mbox{-}$admissible measure of\, $\mathbf{Y}$.
The Definition~\ref{def:sigma_admissible} and the model \eqref{equ:modelsetting0} imply  that $\mu$ and $\widehat \mu$ satisfy
\begin{align}
\mathscr{F}[\widehat{\mu}](\omega)=\mathscr{F}[\mu](\omega)+\mathbf{W}_1(\omega),
\quad \omega\in[0,2\Omega],
\label{equ:imagemodeling2sigma_locs_theta0}
\end{align}
for some $\mathbf{W}_1$
with $\babs{\mathbf{W}_1\pare{\omega}}<2\sigma$, $\omega\in[0,2\Omega]$. Define $S_1$ as $S_1=\{y_2,\widehat{y}_2\}$. For any $0<\omega^\ast\leqs \frac{2\Omega}{\# S_1+1}=\frac{2\Omega}{3}$, applying \eqref{equ:cor_11} to \eqref{equ:imagemodeling2sigma_locs_theta0} we obtain that
\begin{align}
\prod_{k=1}^{2}\left|\pare{e^{iy_1\omega^\ast}-e^{i\widehat{y}_k\omega^\ast}}
\left(e^{iy_1\omega^\ast}-e^{iy_2\omega^\ast}\right)
\right|
< \frac{2^4\sigma}{m}.
\label{equ:locampineuq1}
\end{align}
Denote the separation by $d:=\babs{y_1-y_2}$. Assume toward a contradiction that $\babs{\widehat{y}_1-y_1}\geqs \frac{d}{2}$. Then also
$\babs{\widehat{y}_2-y_1}\geqs \babs{\widehat{y}_1-y_1}\geqs \frac{d}{2}$ as $\hat y_1$ is assumed to be closer to $y_1$.
Using the inequality 
\begin{align}
    \babs{e^{ix} -e^{iy}}=2\sin\frac{\babs{x-y}}{2}\geqs \frac{2}{\pi}\babs{x-y},
    \label{equ:relationship_exp_sin}
\end{align}
we obtain the lower bound of left hand side (LHS) in (\ref{equ:locampineuq1}): 
\begin{align*}
    &\babs{ \left( e^{i y_{1} \omega^\ast} - e^{i y_{2} \omega^\ast} \right)
    \left( e^{i y_{1} \omega^\ast} - e^{i \widehat{y}_{1} \omega^\ast} \right)
    \left( e^{i y_{1} \omega^\ast} - e^{i \widehat{y}_{2} \omega^\ast} \right)}\\
    =&8 \babs{\sin \frac{\babs{y_1-y_2}\omega^\ast}{2}\sin \frac{\babs{y_1-\widehat y_1}\omega^\ast}{2}\sin \frac{\babs{y_1-\widehat y_2}\omega^\ast}{2}}\\
    \geqs& 8 \babs{\sin \frac{d\omega^\ast}{2}\sin \frac{d\omega^\ast}{4}\sin \frac{d\omega^\ast}{4}}.
\end{align*}
Combining this with  (\ref{equ:locampineuq1}) obtains
\begin{align}
    \sin^2 \left( \frac{d\omega^\ast}{4}\right) \sin\left(\frac{d\omega^\ast}{2}\right) < \frac{2\sigma}{m}.
\label{equ:locampineuq3}
\end{align} 
Since $0<\omega^\ast\leqs\frac{2\Omega}{3}$, applying
\(
\sin \frac{x}{4} = \frac{\sin\frac{x}{2}}{2\cos\frac{x}{4} }\geqs \frac{1}{2}\sin\frac{x}{2}
\) into \eqref{equ:locampineuq3} yields
\begin{align}
    d < \frac{3}{\Omega}\arcsin\left(2\left(\frac{\sigma}{ m } \right)^\frac{1}{3} \right).
    \label{equ:locampineuq5}
\end{align}
Therefore, under the separation condition \eqref{equ:separationforlocationrecovery-thm}, the contradiction shows that $\babs{\widehat{y}_1-y_1}<\frac{d}{2}$.  Define $S_2 = \{y_1,\widehat{y}_1\}$ and applying the same argument at $S_2$ yields
$\babs{\widehat{y}_2-y_2}<\frac{d}{2}$ whenever
\[
d\geqs \frac{3}{\Omega}\arcsin\!\pare{2\pare{\frac{\sigma}{\beta m}}^{\frac{1}{3}}}.
\]
Since $\beta\geqs 1$, the condition \eqref{equ:separationforlocationrecovery-thm} implies the above as well. Hence the condition \eqref{equ:separationforlocationrecovery-thm} ensures that $\babs{\widehat{y}_j-y_j}<\frac{d}{2}$ for $j=1,2$, which proves the $\frac{d}{2}$-neighborhood statement.

Next, we shall prove \eqref{equ:condition_locs_positive}. Note that for the general source locations $y_1,y_2$, after shifting them by $x$, we obtain
\begin{align}
     \mathscr{F}[\hat{\mu}](\omega)e^{i x \omega} = \mathscr{F}[\mu](\omega)e^{i x \omega}+ \mathbf{W}_1(\omega)e^{i x \omega}, \quad \left| \mathbf{W}_1(\omega)e^{i x \omega}\right| < 2\sigma, 
     \label{equ:shifting_invariance}
 \end{align}
with $\omega \in [-\Omega,\Omega]$, we can transform the problem into the case when $y_1=-y_2$. Therefore, we assume $y_1=-d/2$  and $y_2=d/2$. 
In this setting, the measurement $\mathbf{Y}$ is generated by
\begin{align*}
    \mu = m\delta_{-\frac{d}{2}}+\beta m\delta_{\frac{d}{2}},\quad d\in\pare{0,\frac{\pi}{\Omega}},\ \beta\geqs1, \ m>0.
\end{align*}

Define $s:=\sqrt{\beta+1}\geqs\sqrt{2}$, and we consider the two-support positive measure 
\[
\widehat{\mu}
:=\frac{1}{2}m s(s+1)\,\delta_{\frac{d}{2}-\frac{d}{s}}
+\frac{1}{2}m s(s-1)\,\delta_{\frac{d}{2}+\frac{d}{s}} .
\]
A direct computation gives, for $\babs \omega \leqs\Omega$,
\begin{align*}
&\bracket{\mathscr{F}[\widehat\mu](\omega)-\mathscr{F}[\mu](\omega)} \\
= &m\bracket{ e^{-i \phi} -\frac{1}{2}s(s+1)e^{i (\phi - \frac{2\phi}{s})}
+(s^2-1)e^{i \phi} - \frac{1}{2}s(s-1)e^{i (\phi + \frac{2\phi}{s})}} \\
=& 2m \sqrt{\pare{ s^2\sin^2\frac{\phi}{s}-\sin^2 \phi}^2
+\pare{s\sin \frac{\phi}{s} \cos\frac{\phi}{s} - \sin \phi\cos\phi}^2}.
\end{align*}
where $\phi:=\frac{d\omega}{2}$. Since $d\in \pare{0,\frac{\pi}{\Omega}}$ and $|\omega|\leqs\Omega$, we have
\[
0\leqs |\phi|\leqs \frac{d\Omega}{2}\leqs \frac{\pi}{2},
\quad\text{and}\quad 0\leqs \frac{\phi}{s}\leqs \frac{\pi}{2}.
\]
Applying Lemma.\ref{lem:sincosinequality} yields
\begin{align*}
&\pare{ s^2\sin^2\frac{\phi}{s}-\sin^2\phi}^2+\pare{s\sin\frac{\phi}{s}\cos\frac{\phi}{s}-\sin\phi\cos\phi}^2 \\
=&\pare{s\sin\frac{\phi}{s}+\sin\phi}^2\pare{s\sin\frac{\phi}{s}-\sin\phi}^2 \\
&+\pare{s\sin\frac{\phi}{s}\cos\frac{\phi}{s}-\sin\phi\cos\phi}^2\\
\leqs &(2\phi)^2\pare{\frac{\phi^3}{6}}^2 + \pare{\phi-\pare{\phi-\frac{\phi^3}{6}}\pare{1-\frac{\phi^2}{2}}}^2 \\
=& \frac{4}{9}\phi^6+\frac{1}{144}\phi^{10}.
\end{align*}
Consequently, for $|\omega|\leqs\Omega$,
\[
\babs{\mathscr{F}[\widehat\mu](\omega)-\mathscr{F}[\mu](\omega)}
\leqs 2m\pare{\frac{d\Omega}{2}}^3\sqrt{\frac{4}{9}+\frac{1}{144}\pare{\frac{\pi}{2}}^4}.
\]
Thus we obtain $
\babs{\mathcal F[\widehat\mu](\omega)-\mathcal F[\mu](\omega)} <2\sigma$ whenever \eqref{equ:condition_locs_positive} holds, which indicates that $\widehat \mu$ is $\sigma$-admissible of $\mu$ with two supports. Moreover, the support locations of $\widehat{\mu}$ are
\[
\widehat y_1=\frac{d}{2}-\frac{d}{s},\quad \widehat y_2=\frac{d}{2}+\frac{d}{s}.
\]
If $s> 2$, then $\babs{\widehat y_1-y_1}=d-\frac{d}{s}> \frac{d}{2}$ and $\widehat y_2>0$, no relabeling can place a spike within $d/2$ of $y_1$. If $\sqrt2\leqs s\leqs2$, then $|\widehat y_2-y_2|=\frac{d}{s}\geqs\frac{d}{2}$,  no relabeling can place a spike within $d/2$ of $y_2$. Therefore, $\widehat{\mu}$ does not lie in the $\frac{d}{2}$-neighborhood of $\mu$.

\section{Proofs of Results in Section \ref{section:small_theta}} \label{proof:small_theta}

\subsection{Proof of Theorem~\ref{thm:twopointresolution_complex_number_UB}} \label{proof:num_theta_small_UB}
First, we shall prove \eqref{equ:sepacondinumbercomplex_smalltheta}. Given the measurement $\mathbf{Y}$ from the discrete measure  $\mu= m e^{i\theta_1} \delta_{y_1}+\beta m e^{i\theta_2} \delta_{y_2}$ with
$y_1,y_2\in B_{\frac{\pi}{2\Omega}}(0)$, $\theta_1, \theta_2\in \pare{-\pi,\pi}$, $\beta \geqs 1$, and $m>0$. Suppose that $\widehat{\mu}=\widehat{a}\,\delta_{\widehat{y}}$ is a $\sigma$-admissible measure of\, $\mathbf{Y}$. The Definition~\ref{def:sigma_admissible} and model \eqref{equ:modelsetting0} imply that $\mu$ and $\widehat \mu$ satisfy
\begin{align}
\mathscr{F}[\widehat{\mu}](\omega)=\mathscr{F}[\mu](\omega)+\mathbf{W}_1(\omega),
\quad \omega\in[0,2\Omega],
\label{equ:imagemodel2sigma_num_smalltheta}
\end{align}
for some $\mathbf{W}_1$
with $\babs{\mathbf{W}_1}<2\sigma$, $\omega\in[0,2\Omega]$. Define $S_1=\{y_2\}$ and $S_2 = \{y_1\}$, so that $\# S_1 = \# S_2=1$.   For any $0<\omega^\ast\leqs \frac{2\Omega}{\# S_t+1}=\Omega$, applying \eqref{equ:LA_8} into \eqref{equ:imagemodel2sigma_num_smalltheta} at $\omega=0$ and $\omega=\omega^\ast$  yields
\[
a_1e^{iy_1\omega^\ast}+a_2e^{iy_2\omega^\ast}+\mathbf{W}_1(\omega^\ast)
=\pare{a_1+a_2+\mathbf{W}_1(0)}e^{i\widehat y\omega^\ast}.
\]
Taking squared magnitudes  removes the dependence on $\widehat y$, giving:
\begin{align*}
\babs{a_1e^{iy_1\omega^\ast}+a_2e^{iy_2\omega^\ast}+\mathbf{W}_1(\omega^\ast)}^2
=\babs{a_1+a_2+\mathbf{W}_1(0)}^2 .
\end{align*}
% \begin{align*}
% &\babs{a_1e^{iy_1\omega^\ast}+a_2e^{iy_2\omega^\ast}+\mathbf{W}_1(\omega^\ast)}^2\\ 
%     =& \babs{a_1}^2 + \babs{a_2}^2 + \babs{\mathbf{W}_1(\omega^\ast)}^2 + 2\text{Re}\pare{a_1 \overline{a}_2 e^{i\pare{y_1-y_2}\omega^\ast}} \\
%     &+2\text{Re}\pare{a_1 e^{i y_1\omega^\ast}\overline{\mathbf{W}_1(\omega^\ast)}} + 2\text{Re}\pare{a_2 e^{i y_2\omega^\ast}\overline{\mathbf{W}_1(\omega^\ast)}}
% \end{align*}
% \begin{align*}
%     &\babs{a_1+a_2+\mathbf{W}_1(0)}^2\\
%     =& \babs{a_1}^2 + \babs{a_2}^2 + \babs{\mathbf{W}_1(0)}^2 + 2\text{Re}\pare{a_1 \overline{a}_2 } \\
%     &+2\text{Re}\pare{a_1 \overline{\mathbf{W}_1(0)}} + 2\text{Re}\pare{a_2\overline{\mathbf{W}_1(0)}}
% \end{align*}
Expanding both sides and rearranging terms yields
\begin{align*}
   & 2\text{Re}\pare{a_1 \overline{a}_2 e^{i\pare{y_1-y_2}\omega^\ast}}-2\text{Re}\pare{ a_1 \overline{a}_2 } = \babs{\mathbf{W}_1(0)}^2 -\babs{\mathbf{W}_1(\omega^\ast)}^2 \notag \\
   &  + 2\text{Re}\pare{ \overline{\mathbf{W}_1(0)} \pare{a_1+a_2 }} -2\text{Re} \pare{\overline {\mathbf{W}_1(\omega^\ast)}\pare{a_1 e^{i y_1\omega^\ast}+a_2  e^{i y_2\omega^\ast}}},
\end{align*}
where $\overline{x}$ denotes the complex conjugate of $x$. Then the left-hand side (LHS) and right-hand side (RHS) above satisfy
\begin{align*}
    \mathrm{LHS}=&2\beta m^2 \bracket{\cos\pare{\pare{y_1-y_2}\omega^\ast+\pare{\theta_1-\theta_2}}-\cos\pare{\theta_1-\theta_2}} ,\\
    \mathrm{RHS}=& \babs{\mathbf{W}_1(0)}^2 -\babs{\mathbf{W}_1(\omega^\ast)}^2+2m\mathrm{Re}\bracket{\overline{\mathbf{W}_1(0)} \pare{e^{i\theta_1}+\beta e^{i\theta_2}}}\notag\\
&+2m\mathrm{Re}\bracket{\overline{\mathbf{W}_1(\omega^\ast)}\pare{e^{i(\theta_1+y_1\omega^\ast)}+\beta e^{i(\theta_2+y_2\omega^\ast)}}}.
\end{align*}
Let $d :=\babs{y_1-y_2}>0$ and assume $\theta:=\theta_1-\theta_2>0$. By translation invariance \eqref{equ:shifting_invariance}, we assume without loss of generality that 
\begin{align*}
    y_1 = \frac{d}{2},\quad y_2=-\frac{d}{2};\quad \theta_1 = \frac{\theta}{2},\quad \theta_2 = -\frac{\theta}{2}.
\end{align*}
Substituting these assumptions into the LHS yields
\begin{align}
\babs{\mathrm{LHS}} = & 2\beta m^2\babs{\cos(d\omega^\ast+\theta)-\cos\theta} \notag \\
=&4\beta m^2\babs{\sin\!\pare{\frac{d\omega^\ast}{2}}\,
\sin\!\pare{\frac{d\omega^\ast}{2}+\theta}}.
\label{equ:lhs_sin}
\end{align}
Since $\babs{\mathbf{W}_1(\cdot)}<2\sigma$ and $\babs{e^{-i\alpha}+\beta e^{i\alpha}}\leqs 1+\beta$, we have
\begin{align*}
    &\babs{|\mathbf{W}_1(0)|^2-|\mathbf{W}_1(\omega^\ast)|^2}< 4\sigma^2, \\
    &2\babs{\text{Re}\pare{\mathbf{W}_1(\cdot)\,m(e^{-i\alpha}+\beta e^{i\alpha})}}< 4\sigma m(1+\beta).
\end{align*}
Therefore, substituting these assumptions and inequalities into the RHS gives the bound
\begin{equation}
\babs{\mathrm{RHS}} < 4\sigma^2+8(\beta+1)m\sigma .
\label{equ:rhs_bound}
\end{equation}
Combining \eqref{equ:lhs_sin}--\eqref{equ:rhs_bound} and dividing both sides by $4\beta m^2$ yields
\begin{align*}
\babs{\sin\!\pare{\frac{d\omega^\ast}{2}}\,
\sin\!\pare{\frac{d\omega^\ast}{2}+\theta}}< \frac{1}{\beta}\pare{\frac{\sigma}{m}}^{\!2}
+\frac{\beta+1}{\beta}\frac{2\sigma}{m}.
\end{align*}
Under $0<\sigma/m\leqs 1/2$, we further obtain $\frac{1}{\beta}(\sigma/m)^2\leqs \frac{1}{2\beta}(\sigma/m)$, and it follows
\begin{equation}
\babs{\sin\!\pare{\frac{d\omega^\ast}{2}}\,
\sin\!\pare{\frac{d\omega^\ast}{2}+\theta}}
< \pare{2+\frac{2.5}{\beta}}\frac{\sigma}{m} .
\label{equ:basic_ineq_complex_num}
\end{equation}
Set $\omega^*=\Omega$. If $
0<\theta\leqs \frac{4\pi}{3}\sqrt{1+\frac{1.25}{\beta}}\pare{\frac{\sigma}{m}}^\frac{1}{2}$ and $\frac{d\Omega}{2}+\theta\leqs\frac{\pi}{2}$, then
\begin{align}
    \babs{\sin\!\pare{\frac{d\Omega}{2}}\sin\!\pare{\frac{d\Omega}{2}+\theta}}
& \geqs \left(\frac{2}{\pi}\right)^2
\pare{\frac{d\Omega}{2}}\pare{\frac{d\Omega}{2}+\theta}\notag \\
&\geqs \frac{2}{\pi^2}
\pare{\frac{d\Omega}{2}+\frac{\theta}{2}}^{\!2}.
\label{equ:sin_sin_lower}
\end{align}
Therefore, when $
0<\theta\leqs \frac{4\pi}{3}\sqrt{1+\frac{1.25}{\beta}}\pare{\frac{\sigma}{m}}^\frac{1}{2}$, if $\frac{d\Omega}{2}+\theta\leqs \frac{\pi}{2}$ and 
\begin{align}
    d\geqs \ d_\star := \frac{2 \pi\pare{\pare{1+\frac{1.25}{\beta}}\frac{\sigma}{m}}^\frac{1}{2}-\theta}{\Omega},
    \label{equ:d_star_num}
\end{align}
then \eqref{equ:basic_ineq_complex_num} does not hold, thus no $\sigma\mbox{-}$admissible measure of\, $\mathbf{Y}$ can have fewer than two supports. If $d$ increases such that  $\frac{d\Omega}{2}+\theta > \frac{\pi}{2}$ or even $d\Omega>\pi$, we instead choose  $\omega^\ast = \frac{ d_\star \Omega}{d}<\Omega$ so that  $\frac{d\omega^\ast}{2}+\theta\leqs \frac{\pi}{2}$ and $d\omega^\ast\leqs\pi$. 

It remains to show that the preceding argument also covers negative relative phases and the near-$\pi$ phase regime. 
First, suppose that $\theta<0$ and assume $y_1=-\frac d2,\quad y_2=\frac d2$.
Then \eqref{equ:lhs_sin} becomes
\begin{align}
|\mathrm{LHS}|
&=2\beta m^2
\left|
\cos(-d\omega^\ast+\theta)
-\cos(\theta)
\right|  \notag \\
&=4\beta m^2
\left|
\sin\!\left(\frac{d\omega^\ast}{2}\right)
\sin\!\left(\frac{d\omega^\ast}{2}-\theta\right)
\right|.
\label{equ:lhs_sin_negative}
\end{align}
This has exactly the same form as \eqref{equ:lhs_sin}, with $\theta$ replaced by $-|\theta|$. Moreover, the estimate \eqref{equ:rhs_bound} is unchanged. Therefore, the same contradiction argument yields the separation condition \eqref{equ:d_star_num} for $\theta<0$ after replacing $\theta$ by $|\theta|$.
Next, consider the near-$\pi$ regime, namely,
\(
\pi-|\theta|
\leq
\frac{4\pi}{3}
\sqrt{1+\frac{1.25}{\beta}}
\left(\frac{\sigma}{m}\right)^{1/2},
\)
When $\theta>0$, we set $t:=\pi-\theta$ and  assume $y_1=-\frac d2,\quad y_2=\frac d2$.
Then \eqref{equ:lhs_sin} becomes
\begin{align}
|\mathrm{LHS}|
&=2\beta m^2
\left|
\cos(-d\omega^\ast+\pi-t)
-\cos(\pi-t)
\right| \notag \\
&=4\beta m^2
\left|
\sin\!\left(\frac{d\omega^\ast}{2}\right)
\sin\!\left(\frac{d\omega^\ast}{2}+t\right)
\right|.
\label{equ:lhs_sin_near_pi_positive}
\end{align}
When $\theta<0$, we set $t:=\pi+\theta = \pi-|\theta|$ and  assume $y_1=\frac d2,\quad y_2=-\frac d2$. Then \eqref{equ:lhs_sin} becomes
\begin{align}
|\mathrm{LHS}|
&=2\beta m^2
\left|
\cos(d\omega^\ast+t-\pi)
-\cos(t-\pi)
\right| \notag \\
&=4\beta m^2
\left|
\sin\!\left(\frac{d\omega^\ast}{2}\right)
\sin\!\left(\frac{d\omega^\ast}{2}+t\right)
\right|.
\label{equ:lhs_sin_near_pi_negative}
\end{align}
Thus \eqref{equ:lhs_sin_near_pi_positive} and \eqref{equ:lhs_sin_near_pi_negative} have exactly the same form as \eqref{equ:lhs_sin}, with $\theta$ replaced by $t$.
And the estimate \eqref{equ:rhs_bound} still unchanged.  

Combining the negative-phase and near-$\pi$ cases, the condition \eqref{equ:d_star_num} applies verbatim with $\theta$ replaced by $|\theta|_{\min}$.

Next, we shall prove \eqref{equ:separation_condition_num_complex_LB_small}. Let $d :=\babs{y_1-y_2}>0$ and assume $\theta:=\theta_1-\theta_2>0$.
By translation invariance \eqref{equ:shifting_invariance}, we shift the coordinate system so that 
\begin{align*}
    y_1 = \frac{d}{2},\quad y_2=-\frac{d}{2};\quad \theta_1 = \frac{\theta}{2},\quad \theta_2 = -\frac{\theta}{2}.
\end{align*}
In this setting, the measurement $\mathbf{Y}$ is generated by
\begin{align*}
    \mu = m e^{i\frac\theta 2}\delta_{\frac{d}{2}}+\beta m e^{-i\frac\theta 2}\delta_{-\frac{d}{2}},\quad d\in\pare{0,\frac{\pi}{\Omega}},\,\beta\geqs 1,\,m>0.
\end{align*}
When $\theta\asymp\pare{\frac{\sigma}{m}}^\frac{1}{2}$, we consider the discrete measure 
\[
\widehat\mu = (\beta+1)m e^{\pare{i\frac{1-\beta}{\beta+1}\frac{\theta}{2}}} \delta_{\frac{1-\beta}{\beta+1}\frac{d}{2}} .
\]
For $\babs\omega\leqs\Omega$, we have
\begin{align*}
 \babs{\mathcal F[\widehat\mu](\omega)-\mathcal F[\mu](\omega)}
% =&\babs{(\beta+1)me^{ i\frac{1-\beta}{\beta+1}\frac{d\omega+\theta}{2}}-m e^{i\frac{d\omega+\theta}{2}} - \beta m e^{-i\frac{d\omega+\theta}{2}}} \notag \\
=&m\babs{\pare{\beta+1} \pare{e^{i\frac{d\omega+\theta}{\beta+1}}-1}-\pare{e^{i\pare{d\omega+\theta}}-1}} \\
    =& m\babs{d\omega+\theta} \babs{\int_0^1 e^{i\frac{d\omega+\theta}{\beta+1}x}-e^{i\pare{d\omega+\theta}t} dx}\notag \\
    \leqs& m\babs{d\omega+\theta}  \int_0^1 \babs{e^{i\frac{d\omega+\theta}{\beta+1}x}-e^{i\pare{d\omega+\theta}x} }dx\notag \\
    =&2 m\babs{d\omega+\theta}  \int_0^1 \babs{\sin\pare{\frac{\beta \pare{d\omega+\theta}}{2\pare{\beta+1}}x}}dx \notag  \\
    % \leqs & \frac{\beta m}{\beta+1}\pare{d\omega+\theta}^2\int_0^1tdt =
    \leqs &\frac{\beta m}{2\pare{\beta+1}}\pare{d\omega+\theta}^2.
\end{align*}
Consequently, if $\theta\leqs \pare{\frac{\sigma}{m}}^\frac{1}{2} $ and  
\begin{align}
    d <   \frac{2\pare{\frac{\beta+1}{\beta}\frac{\sigma}{m}}^{\frac{1}{2}}-\theta}{\Omega},
    \label{proof:num_LB_case1_smalltheta}
\end{align}
then $\babs{\mathscr{F}[\widehat{\mu}](\omega)-\mathscr{F}[\mu](\omega)}< 2\sigma$, and hence there exists a  $\sigma$-admissible measure of $\mu$ with only one support.

When $\pi - \theta \asymp \pare{\frac{\sigma}{m}}^\frac{1}{2}$, define $t:=\pi-\theta$. Then we rewrite the discrete measure by
\begin{align*}
    \mu = m e^{i\frac\theta 2}\delta_{-\frac{d}{2}}+\beta m e^{-i\frac\theta 2}\delta_{\frac{d}{2}}=im e^{-i\frac{t}{2}}\delta_{-\frac{d}{2}}  -i\beta m e^{i\frac{t}{2}}\delta_{\frac{d}{2}}.
\end{align*}
Assume $\beta>1$ and consider the discrete measure 
\[
\widehat{\mu} = i\pare{1-\beta}me^{i\frac{\beta+1}{\beta-1}\frac{t}{2}}\delta_{\frac{\beta+1}{\beta-1}\frac{d}{2}}.
\]
For $\babs{\omega}\leqs\Omega$, we have
\begin{align*}
     \babs{\mathcal F[\widehat\mu](\omega)-\mathcal F[\mu](\omega)}
% =&\babs{i\pare{1-\beta}me^{i\frac{\beta+1}{\beta-1}\frac{d\omega+t}{2}} - i m e^{-i\frac{d\omega+t}{2}} + i \beta m e^{i\frac{d\omega+t}{2}}} \notag \\
% =m\babs{\pare{1-\beta}e^{i\frac{1}{\beta-1}\pare{d\omega+t}} -  e^{-i\pare{d\omega+t}} +  \beta } \\
=&m\babs{\pare{1-\beta}\pare{e^{i\frac{d\omega+t}{\beta-1}}-1}-\pare{e^{-i\pare{d\omega+t}}-1}} \\
\leqs&m\babs{d\omega+t}\int_0^1\babs{e^{i\frac{d\omega+t}{\beta-1}x}-e^{-i\pare{d\omega+t}x}}dx\\
=&2m\babs{d\omega+t}\int_0^1\babs{\sin\pare{\frac{\beta \pare{d\omega+t}}{2\pare{\beta-1}}x}}dx\\
\leqs &\frac{\beta}{2(\beta-1)}m \pare{d\omega+t}^2
\end{align*}
Consequently, if $\pi-\theta\asymp\pare{\frac{\sigma}{m}}^{\frac{1}{2}}$ and
\begin{align}
    d <   \frac{2\pare{\frac{\beta-1}{\beta}\frac{\sigma}{m}}^{\frac{1}{2}}-\pi+\theta}{\Omega},    
    \label{proof:num_LB_case3_smalltheta}
\end{align}
then $\babs{\mathscr{F}[\widehat{\mu}](\omega)-\mathscr{F}[\mu](\omega)}< 2\sigma$, and hence there exists a  $\sigma$-admissible measure of $\mu$ with only one support.

It remains to consider the case $\theta<0$. When
$\babs{\theta}$ is close to zero, we assume the two locations as $y_1=-\frac d2$ and $y_2=\frac d2$, and take
\(\widehat{\mu} = (\beta+1)m
    e^{i\frac{1-\beta}{\beta+1}\frac{\theta}{2}}
    \delta_{-\frac{1-\beta}{\beta+1}\frac d2}.\)
A direct computation gives, for $|\omega|\leq \Omega$,
\begin{align*}
\babs{\mathcal F[\widehat\mu](\omega)-\mathcal F[\mu](\omega)}
    =&m\babs{\pare{\beta+1} \pare{e^{i\frac{d\omega-\theta}{\beta+1}}-1}-\pare{e^{i\pare{d\omega-\theta}}-1}}.
\end{align*}
Hence the condition in \eqref{proof:num_LB_case1_smalltheta} remains valid with
$\theta$ replaced by $|\theta|$.
Similarly, when $\babs{\theta}$ is near $\pi$ and $\beta>1$, we set $t:=\pi+\theta=\pi-\babs{\theta}$ and assume $y_1=\frac{d}{2}$, $y_2=-\frac{d}{2}$. After taking $\widehat{\mu} =i(\beta-1)m e^{-i\frac{\beta+1}{\beta-1}\frac t2} \delta_{-\frac{\beta+1}{\beta-1}\frac d2}$, a direct computation gives,
\begin{align*}
     \babs{\mathcal F[\widehat\mu](\omega)-\mathcal F[\mu](\omega)}
    =&m\babs{\pare{1-\beta} \pare{e^{i\frac{d\omega+t}{\beta-1}}-1}-\pare{e^{-i\pare{d\omega+t}}-1}}.
\end{align*}
Therefore, the condition in \eqref{proof:num_LB_case3_smalltheta} also applies with
$\pi-\theta$ replaced by $\pi-|\theta|$.

Combining the cases $\theta>0$ and $\theta<0$, the lower-bound construction depends on $|\theta|_{\min}$. Hence the separation condition \eqref{equ:separation_condition_num_complex_LB_small} holds, completing the proof.

\subsection{Proof of Theorem~\ref{thm:twopointresolution_complex_number_UB_beta1}}
Let $\mu=\sum_{j=1}^2 a_j \delta_{y_j}$, $a_j = m e^{i\theta_j}$, and $\hat{\mu}=a\,\delta_{\hat{y}}$. A crucial relation is
\begin{align}
    \mathscr{F}[\hat{\mu}](\omega)  =  \mathscr{F}[\mu](\omega) + \mathbf{W}_1(\omega),  \,\,  |\mathbf{W}_1(\cdot)| < 2\sigma, \,   \omega \in [-\Omega,\Omega].
    \label{equ:basicmodel_num_small_beta1}
\end{align}
Note that if \eqref{equ:basicmodel_num_small_beta1} holds, $\hat{\mu}$ can be a $\sigma$-admissible measure of some $Y$ generated by model \eqref{equ:modelsetting0}. This time, resolving two point sources is impossible. Conversely, if \eqref{equ:basicmodel_num_small_beta1} does not hold, $\hat{\mu}$ cannot be any $\sigma$-admissible measure of some $Y$ generated by $\mu$ as in model \eqref{equ:modelsetting0}. 

By translation invariance \eqref{equ:shifting_invariance}, we can transform the problem into the case when $y_1=-y_2$. First assume $\theta>0$. Since $\pi-\theta \asymp \pare{\frac{\sigma}{m}}^\frac{1}{2}$, let $t=\pi-\theta$. Then we consider that the underlying source is $\mu=me^{i\frac{\theta}{2}}\delta_{y_1}+me^{-i\frac{\theta}{2}}\delta_{y_2}=ime^{-i\frac{t}{2}}\delta_{y_1}-ime^{i\frac{t}{2}}\delta_{y_2}$ with $y_2>0$, $y_1 = -y_2$. The measure $\widehat{\mu}$ is $a\delta_{\widehat{y}}$ with $a$ and $\widehat{y}$ to be determined.
From \eqref{equ:basicmodel_num_small_beta1}, we get
\begin{align*}
    \mathbf{W}_1(\omega)= &a e^{i\widehat{y}\omega} -\pare{ime^{i\pare{-\frac{t}{2}+y_1\omega}}-ime^{i\pare{\frac{t}{2}+y_2\omega}}}\\
    = & a e^{i\widehat{y}\omega} -2m\sin\pare{y_2\omega+\frac{t}{2}}.
\end{align*}
Note that for two non-negative values $x,y$, we have
\begin{align}
     \left|x e^{iq} - y\right|^2 = & \left(x\cos(q)-y\right)^2 + x^2\sin^2(q) \notag \\
     = & x^2 + y^2 - 2xy\cos(q)\geqs (x-y)^2
     \label{equ:absolute_lb}
\end{align} 
and the equality is attained when $q=0$. We only consider the case when
\begin{align}
    y_2\Omega+\frac{t}{2} \leqs \frac{\pi}{2} \quad \mathrm{and} \quad  y_2\Omega+\frac{t}{2} \geqs-\frac{\pi}{2}
    \label{equ:condition_y1Omega}
\end{align}
and we shall see that this coincides with the case in the theorem. By the above condition, we have $\sin\pare{y_2\omega+\frac{t}{2}}\geqs0, \omega\in[-\Omega,\Omega]$. Thus by \eqref{equ:absolute_lb}, for every $\omega$,
\begin{align*}
    \babs{\mathbf{W}_1(\omega)}\geqs \babs{\babs{a}-2m\sin\pare{ y_2\omega+\frac{t}{2}}}
\end{align*}
and the minimum is attained when $\widehat{y}=0$ and $a$ is a positive number. We now try to find the condition on $y_2$ so that there exists $a$ satisfying 
\begin{align*}
    \babs{\babs{a}-2m\sin\pare{ y_1\omega+\frac{t}{2}}}<2\sigma,\quad \omega\in[-\Omega,\Omega].
\end{align*}
This is equivalent to 
\begin{align}
    \max_{\omega,\,\omega'} \babs{2m\sin\pare{\frac{d\omega+t}{2}} - 2m\sin\pare{\frac{d\omega'+t}{2}}} < 4\sigma.
    \label{equ:problem_max}
\end{align}
We denote $d=|y_1-y_2|$ and now the condition \eqref{equ:condition_y1Omega} is
\begin{align}
    \frac{d\Omega+t}{2}\leqs\frac{\pi}{2}\quad \mathrm{and} \quad \frac{-d\Omega+t}{2}\geqs -\frac{\pi}{2}.
    \label{equ:condition2_y1Omega}
\end{align}
Under this condition, \eqref{equ:problem_max} becomes
\begin{align*}
    2m\babs{\sin\pare{\frac{d\Omega+t}{2}}-\sin\pare{\frac{-d\Omega+t}{2}}} = 4m\babs{\sin\pare{\frac{d\Omega}{2}}\cos\pare{\frac{t}{2}}}<4\sigma.
\end{align*}
Since $\sin\pare{\frac{d\Omega}{2}}\geqs0$ and $\cos\pare{\frac{t}{2}}\geqs0$,  we get 
\begin{align}
    d<\frac{2}{\Omega}\arcsin \pare{\frac{1}{\cos\pare{\frac{t}{2}}}\frac{\sigma}{m}}.
    \label{equ:separation_beta1_num_UB}
\end{align}
Now the condition \eqref{equ:condition2_y1Omega} holds when
\begin{align*}
\begin{cases}
    \displaystyle \left(\frac{2}{\Omega}\arcsin\left(\frac{1}{\cos(t/2)}\frac{\sigma}{m}\right)\right)\Omega + t \leqs \pi, \\[2.0ex]
    \displaystyle -\left(\frac{2}{\Omega}\arcsin\left(\frac{1}{\cos(t/2)}\frac{\sigma}{m}\right)\right)\Omega + t \geqs -\pi.
\end{cases}
\end{align*}
These two inequalities are equivalent to
\begin{numcases}{}
    \cos\left(t\right)\geqs \frac{2\sigma}{m}-1, \label{equ:equa_condition1}\\
    \arcsin\left(\frac{1}{\cos(t/2)}\frac{\sigma}{m}\right)\leqs\frac{t+\pi}{2}. \label{equ:equa_condition2}
\end{numcases}
When $t$ is close to zero and \(\sigma/m<1/2\), we have \(\frac{2\sigma}{m}-1<0\), so \eqref{equ:equa_condition1} is automatically satisfied. In addition, since \(\arcsin(x)\le \pi/2\) for all \(x\in[-1,1]\), \eqref{equ:equa_condition2} is also automatically satisfied.
Therefore, if (\ref{equ:separation_beta1_num_UB}) holds, there exists a single point source to be the $\sigma$-admissible measure. Otherwise, no such one-point admissible measure exists. 

When $\theta<0$, denoting $t:=\pi+\theta=\pi-\babs{\theta}$ and replacing $y_j$ with $-y_j$ leaves the \eqref{equ:problem_max} unchanged. Hence, the separation conditions derived in \eqref{equ:separation_beta1_num_UB} applies verbatim after substituting $\theta$ with $\babs{\theta}$. This completes the proof.

\subsection{Proof of Theorem~\ref{thm:supportrecoveryupperboundcomplex_small_theta}}\label{proof:locs_theta_small_UB}
First, we shall prove \eqref{equ:sepacondilocationcomplex_small_beta1}. 
Given the measurement $\mathbf{Y}$ generated by the discrete measure  $\mu= m e^{i\theta_1} \delta_{y_1}+\beta m e^{i\theta_2}\delta_{y_2}$ with
$y_1,y_2\in B_{\frac{\pi}{2\Omega}}(0)$, $\theta_1,\theta_2\in \pare{-\pi,\pi}$, $\beta\geqs1$, and $m>0$. Suppose that $\widehat{\mu}=\widehat{a}_1\delta_{\widehat{y}_1}+\widehat{a}_2\delta_{\widehat{y}_2}$ is a complex $\sigma$-admissible measure of\, $\mathbf{Y}$. The Definition~\ref{def:sigma_admissible} and model \eqref{equ:modelsetting0} imply that $\mu$ and $\widehat \mu$ satisfy
\begin{align}
\mathscr{F}[\widehat{\mu}](\omega)=\mathscr{F}[\mu](\omega)+\mathbf{W}_1(\omega),
\quad \omega\in[0,2\Omega],
\label{equ:imagemodel2sigma_locs_smalltheta}
\end{align}
for some $\mathbf{W}_1$
with $\babs{\mathbf{W}_1\pare{\omega}}<2\sigma$, $\omega\in[0,2\Omega]$. Define $S_1$ and $S_2$ as, respectively,
\begin{align*}
    S_1=\{y_2,\widehat{y}_2\},\quad S_2=\{y_1,\widehat{y}_1\}.
\end{align*}
Then $\#S_1=\#S_2=2$. For any $0<\omega^\ast \leqs\frac{2\Omega}{\#S_t+1}=\frac{2\Omega}{3}$,applying \eqref{equ:LA_8} to \eqref{equ:imagemodel2sigma_locs_smalltheta} gives
\begin{align*}
     &a_1\prod_{k=1}^{2} \pare{e^{iy_1\omega^\ast}-e^{i\hat y_k\omega^\ast}} \pare{e^{iy_1\omega^\ast}-e^{i y_2\omega^\ast}}=\pare{e^{i\hat y_1 \omega^\ast}\mathbf{w}_1-\mathbf{w}_2}^T\mathbf{v}_{S_1},\\
     &a_2 \prod_{k=1}^{2} \pare{e^{iy_2\omega^\ast}-e^{i\hat y_k\omega^\ast}} \pare{e^{iy_2\omega^\ast}-e^{i y_1 \omega^\ast}} =\pare{e^{i\hat y_2 \omega^\ast}\mathbf{w}_1-\mathbf{w}_2}^T\mathbf{v}_{S_2},
\end{align*}
where
\begin{align*}
    \mathbf{v}_{S_1}=&
    \begin{bmatrix}
        e^{i\pare{y_2+\hat y_2}\omega^\ast}, &
        -\pare{e^{i y_2 \omega^\ast}+e^{i\hat y_2\omega^\ast}},&    1
    \end{bmatrix}^\top,\\
    \mathbf{v}_{S_2}= & 
    \begin{bmatrix}
        e^{i\pare{y_1+\hat y_1}\omega^\ast},&
        -\pare{e^{i y_1 \omega^\ast}+e^{i\hat y_1\omega^\ast}},& 1
    \end{bmatrix}^\top,\\
    \mathbf{w}_{1}=&\begin{bmatrix}
        \mathbf{W}_1(0),&
        \mathbf{W}_1(\omega^\ast),&
        \mathbf{W}_1(2\omega^\ast)
    \end{bmatrix}^\top,\\
\mathbf{w}_{2}=&\begin{bmatrix}
        \mathbf{W}_1(\omega^\ast),&
        \mathbf{W}_1(2\omega^\ast),&
        \mathbf{W}_1(3\omega^\ast)
    \end{bmatrix}^\top.
\end{align*}
Define
\begin{align*}
    A=&-e^{i\pare{\hat y_1+\hat y_2} \omega^\ast}\mathbf{W}_1(\omega^\ast)+\bracket{e^{i \hat y_1 \omega^\ast}  +  e^{i\hat y_2 \omega^\ast}}\mathbf{W}_1(2\omega^\ast)-\mathbf{W}_1(3\omega^\ast),\\
    B= &e^{ i\pare{\hat y_1  + \hat y_2}\omega^\ast}\mathbf{W}_1(0)-\pare{e^{i \hat y_1 \omega^\ast} +e^{i\hat y_2\omega^\ast}}\mathbf{W}_1(\omega^\ast)+\mathbf{W}_1(2\omega^\ast).
\end{align*}
A direct calculation shows that
\begin{align*}
     \pare{e^{i\hat y_1 \omega^\ast}\mathbf{w}_1-\mathbf{W}_1}^T\mathbf{v}_{S_1}=A+e^{iy_2\omega^\ast}\cdot B, \\
     \pare{e^{i\hat y_2 \omega^\ast}\mathbf{w}_1-\mathbf{W}_1}^T\mathbf{v}_{S_2}=A+e^{iy_1\omega^\ast}\cdot B.
\end{align*}
Eliminating the common term $A$ yields
\begin{align}
\label{A2}
&a_1\pare{e^{iy_1\omega^\ast}-e^{i\widehat y_1\omega^\ast}}\pare{e^{iy_1\omega^\ast}-e^{i\widehat y_2\omega^\ast}} \notag \\
&+a_2\pare{e^{iy_2\omega^\ast}-e^{i\widehat y_2\omega^\ast}}\pare{e^{iy_2\omega^\ast}-e^{i\widehat y_1\omega^\ast}}
=-B.
\end{align}
Using $\babs{\mathbf{W}_1\pare{\omega^\ast}}<2\sigma$ and the bound $\babs{e^{i\widehat y_1 \omega^\ast}+e^{i\widehat y_2\omega^\ast}}\leqs 2$, we obtain the uniform estimate
\begin{equation}\label{A3}
|B|\leqs 2\sigma+4\sigma+2\sigma=8\sigma.
\end{equation}
Applying the identity
$e^{ip}-e^{iq}=2i\,e^{i(p+q)/2}\sin\!\pare{\pare{p-q}/2}$, the magnitude of the LHS in \eqref{A2} can be written as
\[
|{\rm LHS}|
=4m\sqrt{P_1^2+\beta^2P_2^2+2\beta P_1P_2\cos\!\bracket{\pare{y_1-y_2}\omega^\ast+\theta}},
\]
where
\begin{align*}
    &P_1=\sin\!\pare{\frac{(y_1-\widehat y_1)\omega^\ast}{2}}\sin\!\pare{\frac{(y_1-\widehat y_2)\omega^\ast}{2}},\\
    &P_2=\sin\!\pare{\frac{(y_2-\widehat y_1)\omega^\ast}{2}}\sin\!\pare{\frac{(y_2-\widehat y_2)\omega^\ast}{2}}.
\end{align*}
Denote the separation by $d:=\babs{y_1-y_2}$ and assume $y_1>y_2$. Consequently, 
\begin{align}
    |{\rm LHS}| = &4m\sqrt{P_1^2+\beta^2P_2^2+2\beta P_1P_2\cos\pare{d\omega^\ast+\theta}}\notag  \\
    =&4m\sqrt{(\beta P_2+P_1\cos\pare{d\omega^\ast+\theta})^2+P_1^2\sin^2\pare{d\omega^\ast+\theta}} \notag \\
    \geqs& 4m\babs{P_1}\cdot \babs{\sin(d\omega^\ast+\theta)}
    \label{A4}
\end{align}
Reorder $\widehat{y}_1,\widehat{y}_2$ such that $\babs{\widehat{y}_1-y_1}\leqs \babs{\widehat{y}_2-y_1}$, and suppose, toward a contradiction,  that $\babs{\widehat{y}_1-y_1}\geqs \frac{d}{2}$. Then
$\babs{\widehat{y}_2-y_1}\geqs \frac{d}{2}$ as well. Since $y_1, \widehat{y}_1, \widehat{y}_2\in B_{\pi/(2\Omega)}(0)$ and $\omega^\ast\leqs 2\Omega/3$, it follows that
$\bracket{\frac{(y_1-\widehat y_t)\omega^\ast}{2}}\leqs \frac{\pi}{3}$, and hence
\[
|P_1|\geqs \sin^2\!\pare{\frac{d\omega^\ast}{4}}.
\]
Combining \eqref{A2}--\eqref{A4} yields the key inequality
\begin{equation}\label{equ:basic_inequality_of_complex_location}
\sin^2\!\pare{\frac{d\omega^\ast}{4}}\,\babs{\sin(d\omega^\ast+\theta)}< \frac{2\sigma}{m}.
\end{equation}
Let $\omega^* = \frac{2\Omega}{3}$. When $0<\theta \leqs 1.75\pi\pare{\frac{\sigma}{m}}^\frac{1}{3}$, if $0<\frac{2}{3}d\Omega+\theta\leqs \frac{\pi}{2}$, we have the following lower bound 
\begin{align}
    \sin^2\pare{\frac{2}{3}\frac{d\Omega}{4}}\babs{\sin\pare{\frac{2}{3}d\Omega+\theta}}& \geqs \pare{\frac{2}{\pi}}^3 \pare{\frac{2}{3}\frac{d\Omega}{4}}^2\pare{\frac{2}{3}d\Omega+\theta}\notag \\
    & \geqs \pare{\frac{2}{3}d\Omega+\frac{\theta}{3}}^3.
    \label{equ:sin_sin_sin_lower}
\end{align}
Therefore, \eqref{equ:basic_inequality_of_complex_location} gives
\begin{align}
    d<\frac{6\pi\pare{\frac{\sigma}{m}}^\frac{1}{3}-\theta}{2\Omega}.
    \label{equ:locs_complex_small_for_y1}
\end{align}
Therefore, if (\ref{equ:locs_complex_small_for_y1}) does not hold, then we have $|\hat y_1 - y_1|< \frac{d}{2}$ which gives $|\hat y_1 - y_2|> \frac{d}{2}$. Since the LHS of \eqref{A2} also satisfies the lower bound
\begin{align}
    |{\rm LHS}| = &4m\sqrt{P_1^2+\beta^2P_2^2+2\beta P_1P_2\cos\pare{d\omega^\ast+\theta}}\notag  \\
    =&4m\sqrt{(P_1 + \beta P_2\cos\pare{d\omega^\ast+\theta})^2+\beta^2P_2^2\sin^2\pare{d\omega^\ast+\theta}} \notag \\
    \geqs& 4\beta m\babs{P_2}\cdot \babs{\sin(d\omega^\ast+\theta)},
    \label{A5}
\end{align}
if we assume $|\hat y_2 - y_2| \geqs \frac{d}{2}$, repeating the above argument gives that, 
\begin{equation}\label{equ:basic_inequality_of_complex_location_eq2}
\sin^2\!\pare{\frac{d\omega^\ast}{4}}\,\babs{\sin(d\omega^\ast+\theta)}< \frac{2\sigma}{\beta m}.
\end{equation}
Similarly to the arguments above, this cannot hold when (\ref{equ:locs_complex_small_for_y1}) fails and $\frac{2}{3}d\Omega +\theta \leqs\frac{\pi}{2}$. Therefore, $|\hat y_2 - y_2|<\frac{d}{2}$. In summary, 
when $\frac{2}{3}d\Omega+\theta\leqs\frac{\pi}{2}$ and 
\begin{align}
    d\geqs d_\star = \frac{6\pi\pare{\frac{\sigma}{m}}^\frac{1}{3}-\theta}{2\Omega},
    \label{equ:d_star_locs}
\end{align}
then $\babs{\widehat y_1-y_1}<\frac{d}{2} $ and $\babs{\widehat y_2-y_2}<\frac{d}{2} $. If $d$ increases so that  $\frac{2}{3}d\Omega+\theta > \frac{\pi}{2}$ or even $\frac{2}{3}d\Omega>\pi$, we instead choose  $\omega^\ast = \frac{ \frac{2}{3}d_\star \Omega}{d}<\frac{2}{3}\Omega$ to ensure  $\frac{2}{3}d\omega^\ast+\theta\leqs \frac{\pi}{2}$ and $d\omega^\ast\leqs\pi$ holds.
Moreover, if $\theta<0$, replacing $y_j$ with $-y_j$ also leaves the lower bounds in \eqref{A4} and \eqref{A5} unchanged. Similarly, when $\theta$ is close to $\pi$ or $-\pi$, define $t: = \pi-\babs\theta$. Replacing $(y_j,\theta)$ with $(-y_j,t)$ also leaves the lower bounds in \eqref{A4} and \eqref{A5}  unchanged.  Therefore, the separation condition \eqref{equ:d_star_locs}  applies verbatim when $\theta<0$ and $\babs\theta$ is near $\pi$ upon substituting $\theta$ with $\babs{\theta}_{\min}$. Hence, $\widehat{\mu}$ lies within the $d/2$-neighborhood of $\mu$ whenever \eqref{equ:sepacondilocationcomplex_small_beta1} holds.

Next, we shall prove \eqref{equ:separation_condition_supp_complex_LB_small}. By translation invariance \eqref{equ:shifting_invariance}, we shift the coordinate system so that $y_1=-d/2$ and $y_2=d/2$. 
In this setting, the measurement $\mathbf{Y}$ is generated by
\begin{align*}
    \mu = m e^{i\frac\theta 2}\delta_{\frac{d}{2}}+\beta m e^{-i\frac\theta 2}\delta_{-\frac{d}{2}},\quad d\in\pare{0,\frac{\pi}{\Omega}},\,\beta\geqs 1,\,m>0.
\end{align*}
When $\theta>0$ and $\theta\asymp\pare{\frac{\sigma}{m}}^\frac{1}{3}$, denote $s=\sqrt{\beta+1}\geqs\sqrt2$ and consider the  measure
\begin{equation*}
\widehat \mu
= \frac12 s(s-1)m\,e^{-i\left(\frac{\theta}{2}+\frac{\theta}{s}\right)}
\delta_{-\frac{d}{2}-\frac{d}{s}}
+ \frac12 s(s+1)m\,e^{-i\left(\frac{\theta}{2}-\frac{\theta}{s}\right)}
\delta_{-\frac{d}{2}+\frac{d}{s}} .
\end{equation*}
For $\babs\omega\leqs\Omega$, a direct computation shows that 
\begin{align}
\babs{\mathcal F[\widehat\mu](\omega)-\mathcal F[\mu](\omega)}
= m\sqrt{g(\phi)},
\label{equ:F_diff_gphi_polished_1}
\end{align}
where $\phi:=d\omega+\theta\in\left( 0,\frac{3\pi}{2}\right]$ and
\begin{align*}
g(\phi)
:=\pare{s^2-1+\cos\phi-s^2\cos\frac{\phi}{s}}^2
 +\pare{s\sin\frac{\phi}{s}-\sin\phi}^2 .
\end{align*}
When $\phi\in\left(0,\frac{\pi}{2}\right]$, the condition   $s\geqs\sqrt{2}$ implies that $\frac{\phi}{s}\in\pare{0,\frac{\pi}{2s}}\subset(0,\frac{\pi}{2})$. Therefore, applying \eqref{equ:scale_sin_cos}  obtains
\begin{align}
g(\phi)
&=4\pare{s\sin\frac{\phi}{2s}+\sin\frac{\phi}{2}}^2
   \pare{s\sin\frac{\phi}{2s}-\sin\frac{\phi}{2}}^2
 +\pare{s\sin\frac{\phi}{s}-\sin \phi}^2 \notag \\
&<4\phi^2\cdot\pare{\frac{\phi}{2}-\frac{\phi}{2}+\frac{1}{6}\frac{\phi^3}{2^3}}^2
+\pare{\phi-\phi+\frac{1}{6}\phi^3}^2
= \frac{1}{576}\phi^8+\frac{1}{36}\phi^6 \notag \\
&\leqs \pare{\frac{1}{36}+\frac{1}{576}\cdot\frac{\pi^2}{4}}\phi^6
<0.0321\,\phi^6 .
\label{equ:thm8_UB1_polished}
\end{align}
When $\phi\in\left(\frac{\pi}{2},\pi\right]$.
the condition $s\geqs\sqrt2$ implies that $\frac{\phi}{2s}\in\pare{0,\frac{\pi}{2}}$ and $\phi-\frac{\pi}{2}\in\pare{0,\frac{\pi}{2}}$. Therefore, applying \eqref{equ:scale_sin_cos}  obtains
\begin{align}
g(\phi)=& \bracket{-1-\sin\pare{\phi-\frac{\pi}{2}}+2s^2\sin^2\frac{\phi}{2s}}^2 \notag \\
&+\bracket{2s\sin\frac{\phi}{2s}\cos\frac{\phi}{2s}-\cos\pare{\phi-\frac{\pi}{2}}}^2 \notag \\
<&\bracket{-1-\pare{\phi-\frac{\pi}{2}-\frac16(\phi-\frac{\pi}{2})^3}+\frac{\phi^2}{2}}^2 \notag \\
 &+\bracket{\phi-\pare{1-\frac12(\phi-\frac{\pi}{2})^2}}^2 <\;0.0253\,\phi^6 .
 \label{equ:thm8_UB2_polished}
\end{align}
When $\phi\in\left(\pi,\frac{3\pi}{2}\right]$, the condition $s\geqs\sqrt{2}$ implies that $\frac{\phi}{2s}\in\pare{0,\frac{3\pi}{4\sqrt{2}}}$,  $\frac{\phi-\pi}{2}\in\pare{0,\frac{\pi}{2}}$ and $\phi-\pi\in\pare{0,\frac{\pi}{2}}$. Therefore, applying \eqref{equ:scale_sin_cos}  obtains
\begin{align}
g(\phi)
&=4\pare{s\sin\frac{\phi}{2s}+\cos\frac{\phi-\pi}{2}}^2
   \pare{s\sin\frac{\phi}{2s}-\cos\frac{\phi-\pi}{2}}^2 \notag \\
& +\pare{2s\sin\frac{\phi}{2s}\cos\frac{\phi}{2s}+\sin(\phi-\pi)}^2 \notag \\
&<4\pare{\frac{\phi}{2}+1}^2
   \pare{\frac{\phi}{2}-1+\frac12\pare{\frac{\phi-\pi}{2}}^2}^2
 +\pare{2\phi-\pi}^2  \notag \\
&<0.0149\,\phi^6 .
\label{equ:thm8_UB3_polished}
\end{align}
Combining \eqref{equ:thm8_UB1_polished}-\eqref{equ:thm8_UB3_polished}, we conclude that
\begin{align*}
\babs{\mathcal F[\widehat\mu](\omega)-\mathcal F[\mu](\omega)} <m\sqrt{0.0321} \,\pare{d\omega+\theta}^3 .
\end{align*}
Hence, if $\theta \asymp \pare{\frac{\sigma}{ m}}^\frac{1}{3}$ and  the separation condition 
\begin{align}
    \babs{y_1-y_2} \leqs \frac{2.23\pare{\frac{\sigma}{ m}}^\frac{1}{3}- \theta}{\Omega}
\label{equ:locs_UB1_smalltheta}
\end{align}
holds, then
$\babs{\mathcal F[\widehat\mu](\omega)-\mathcal F[\mu](\omega)}<2\sigma$,
so $\widehat\mu$ is $\sigma$-admissible and has two supports. Moreover, the support locations of $\widehat\mu$ are
\[
\widehat y_1=-\frac{d}{2}-\frac{d}{s},
\quad
\widehat y_2=-\frac{d}{2}+\frac{d}{s},
\quad s=\sqrt{\beta+1}\geqs\sqrt2.
\]
If $s\geqs 2$, then \(\babs{\widehat y_2-y_2} =d\pare{1-\frac{1}{s}}\geqs \frac{d}{2},\), no relabeling can place a spike within $d/2$ of $y_2$.  If $\sqrt2\leqs s<2$, then $\babs{\widehat y_1-y_1} =\frac{d}{s}>\frac{d}{2}$, no relabeling can place a spike within $d/2$ of $y_1$ Therefore, $\widehat\mu$ does not lie in the $\frac{d}{2}$-neighborhood of $\mu$.

When $\pi - \theta \asymp \pare{\frac{\sigma}{m}}^\frac{1}{3}$, define $t:=\pi-\theta$. By translation invariance \eqref{equ:shifting_invariance}, we rewrite the discrete measure by 
\begin{align*}
    \mu = m e^{i\frac\theta 2}\delta_{-\frac{d}{2}}+\beta m e^{-i\frac\theta 2}\delta_{\frac{d}{2}}=im e^{-i\frac{t}{2}}\delta_{-\frac{d}{2}}  -i\beta m e^{i\frac{t}{2}}\delta_{\frac{d}{2}}.
\end{align*}
Define $s' := \sqrt{\frac{\beta-1}{\beta}} \in (0,1)$ and consider the  measure
\begin{align*}
    \widehat{\mu}
    = \frac{s'}{2(1-s')}\, i m\, e^{i(\frac{t}{2}-\frac{t}{s'})}\,
      \delta_{\frac{d}{2}-\frac{d}{s'}}
    - \frac{s'}{2(1+s')}\, i m\, e^{i(\frac{t}{2}+\frac{t}{s'})}\,
      \delta_{\frac{d}{2}+\frac{d}{s'}} .
\end{align*}
A direct computation shows that, for $|\omega|\leqs \Omega$,
\begin{align}
    & \babs{\mathscr{F}[\widehat{\mu}](\omega)-\mathscr{F}[\mu](\omega)}\\
    =& m\beta \sqrt{\pare{s'^2-1+\cos\phi'-s'^2\cos\frac{\phi'}{s'}}^2
    +\pare{s'\sin\frac{\phi'}{s'}-\sin\phi'}^2},
    \label{equ:F_diff_gphi_polished_2}
\end{align}
where $\phi':=d\omega+t$.
Proceeding as in the case $\theta \asymp \pare{\frac{\sigma}{m}}^\frac{1}{3}$, we obtain
\begin{align*}
    \babs{\mathscr{F}[\widehat{\mu}](\omega)-\mathscr{F}[\mu](\omega)} < m\beta\,\sqrt{0.0321}\,(d\omega+t)^3.
\end{align*}
Consequently, if
\begin{align}
    d<\frac{2.23\pare{\frac{\sigma}{ \beta m}}^\frac{1}{3}- \pi +\theta}{\Omega}, \quad\pi-\theta \asymp \pare{\frac{\sigma}{m}}^\frac{1}{3},
    \label{equ:locs_UB3_smalltheta}
\end{align}
then $\babs{\mathscr{F}[\widehat{\mu}](\omega)-\mathscr{F}[\mu](\omega)}< 2\sigma$, and hence $\widehat{\mu}$ is a   $\sigma$-admissible  complex measure of\, $\mathbf{Y}$ but not within the $\frac{d}{2}$-neighborhood of $\mu$.

When $\theta<0$ and $\babs{\theta}$ is close to zero, replacing $\pare{y_j,\, \widehat{y}_j}$ with $\pare{-y_j,\, -\widehat{y}_j}$ leaves the modulus in \eqref{equ:F_diff_gphi_polished_1}  unchanged. When $\theta<0$ and $\babs{\theta}$ is close to $\pi$, define $t:=\pi+\theta=\pi-\babs{\theta}$. Replacing $\pare{y_j,\, \widehat{y}_j,\,\widehat{a}_j}$ with $\pare{-y_j,\, -\widehat{y}_j,\,\overline{\widehat a_j}}$ leaves the modulus in \eqref{equ:F_diff_gphi_polished_2}  unchanged.
Hence, the separation conditions derived in \eqref{equ:locs_UB1_smalltheta} and \eqref{equ:locs_UB3_smalltheta} apply verbatim when $\theta<0$ upon substituting $\theta$ with $|\theta|$.

\section{Proofs of Results in Section \ref{section:large_theta}} \label{proof:large_theta}

\subsection{Proof of Theorem~\ref{thm:complex_number_UB_large_theta}}

First, we shall prove \eqref{equ:sepacondinumbercomplex_largetheta}. Given the measurement $\mathbf{Y}$ from the discrete measure  $\mu=m e^{i\theta_1} \delta_{y_1}+\beta m e^{i\theta_2} \delta_{y_2}$ with
$y_1,y_2\in B_{\frac{\pi}{2\Omega}}(0)$,
$\theta_1,\theta_2\in \pare{-\pi,\pi}$, $\beta\geqs1$, and $m>0$. Suppose that $\widehat{\mu}=\widehat{a}\,\delta_{\widehat{y}}$ is a $\sigma$-admissible measure of\, $\mathbf{Y}$. The Definition~\ref{def:sigma_admissible} and model \eqref{equ:modelsetting0} imply that $\mu$ and $\widehat \mu$ satisfy
\begin{align}
\mathscr{F}[\widehat{\mu}](\omega)=\mathscr{F}[\mu](\omega)+\mathbf{W}_1(\omega),
\quad \omega\in[0,2\Omega],
\label{equ:imagemodel2sigma_num_largetheta}
\end{align}
for some $\mathbf{W}_1$
with $\babs{\mathbf{W}_1\pare{\cdot}}<2\sigma$, $\omega\in[0,2\Omega]$. Considering the argument in the proof of \eqref{equ:sepacondinumbercomplex_smalltheta}, for any $\omega^\ast \in \left(0,\Omega\right]$,  we can obtain the similar relationship as 
\begin{equation}
\babs{\sin\!\pare{\frac{d\omega^\ast}{2}}\,
\sin\!\pare{\frac{d\omega^\ast}{2}+\babs{\theta}_{\min}}}
< \pare{2+\frac{2.5}{\beta}}\frac{\sigma}{m} :=\varepsilon.
\label{equ:basic_ineq_complex_num_large}
\end{equation}
The condition $d<d_\star$, with $d_\star$ defined in \eqref{equ:d_star_num}, is generally not enough to guarantee \eqref{equ:basic_ineq_complex_num_large} when
\(
|\theta|_{\min}>
\frac{4\pi}{3}
\sqrt{1+\frac{1.25}{\beta}}
\left(\frac{\sigma}{m}\right)^{1/2}.
\)
We therefore develop an alternative bound. Since $|\theta|_{\min}$ is symmetric with respect to
$|\theta|=\pi/2$, we focus on the regime
\[
\frac{4\pi}{3}\sqrt{1+\frac{1.25}{\beta}}\pare{\frac{\sigma}{m}}^\frac{1}{2}<\babs{\theta}_{\min}\leqs\frac{\pi}{2},
\]
and distinguish two subregimes.

\medskip
\noindent
\emph{Subregime A:} 
\begin{align*}
    \max\left\{\babs{\theta}_{\min},\frac{d\omega^\ast}{2}\right\}<\frac{d\omega^\ast}{2}+\babs{\theta}_{\min}\leqs\pi-\babs{\theta}_{\min}.
\end{align*}
In this range,
\begin{align*}
    &\babs{\sin\!\pare{\frac{d\omega^\ast}{2}}\,
\sin\!\pare{\frac{d\omega^\ast}{2}+\babs{\theta}_{\min}}}\\
\geqs & \sin\!\pare{\frac{d\omega^\ast}{2}} \cdot \max\left\{ \sin\!\pare{\frac{d\omega^\ast}{2}}, \sin\babs{\theta}_{\min}\right\},
\end{align*}
and hence \eqref{equ:basic_ineq_complex_num_large} implies
\begin{align}
    \sin\!\pare{\frac{d\omega^\ast}{2}}<   \min \pare{\frac{\varepsilon}{\sin\babs{\theta}_{\min}},\,\sqrt{\varepsilon}}
    \label{equ:num_UB_largetheta_eq1}
\end{align}
Moreover, since  $\frac{\pi}{2}\geqs \babs{\theta}_{\min}> \frac{4\pi}{3}\sqrt{1+\frac{1.25}{\beta}}\pare{\frac{\sigma}{m}}^\frac{1}{2} = \frac{2\sqrt{2}\pi}{3}\varepsilon^\frac{1}{2}$,  the bound $\sin t \geqs\frac{2}{\pi}t$ for $\bracket{0,\frac{\pi}{2}}$ yields
\begin{align*}
    \sin\babs{\theta}_{\min} > \sin\pare{\frac{2\sqrt{2}\pi}{3}\varepsilon^\frac{1}{2}} \geqs \frac{2}{\pi}\cdot \frac{2\sqrt{2}\pi}{3}\varepsilon^\frac{1}{2}=\frac{4\sqrt{2}}{3}\varepsilon^\frac{1}{2}>\varepsilon^\frac{1}{2}.
\end{align*}
Consequently, \eqref{equ:num_UB_largetheta_eq1} simplifies to
\begin{align}
    \sin\!\pare{\frac{d\omega^\ast}{2}}< \frac{\varepsilon}{\sin\babs{\theta}_{\min}}=\frac{\pare{2+\frac{2.5}{\beta}}\frac{\sigma}{m}}{\sin \!(\babs{\theta}_{\min})}.
  \label{equ:num_UB_largetheta_eq2}
\end{align}

\medskip
\noindent
\emph{Subregime B:}
\begin{align*}
    \frac{d\omega^\ast}{2}+\babs{\theta}_{\min}\geqs \pi-\babs{\theta}_{\min},
\end{align*}
which implies $\tfrac{d\omega^\ast}{4}+\babs{\theta}_{\min}\geqs \tfrac{\pi}{2}$.
Since $\omega^\ast\in(0,\Omega]$ and $y_1,y_2\in B_{\frac{\pi}{2\Omega}}(0)$, we may select $\omega^\ast=\Omega/2$, so that
\begin{align*}
\frac{\pi}{2}\leqs \babs{\theta}_{\min}+\frac{d\omega^\ast}{2}
=\babs{\theta}_{\min}+\frac{d\Omega}{4}
\leqs \frac{3\pi}{4}.
\end{align*}
Applying \eqref{equ:basic_ineq_complex_num_large} with $\omega^\ast=\Omega/2$ then gives
\begin{align*}
\babs{\sin\!\pare{\frac{d\Omega}{4}}\sin\frac{3\pi}{4}}<\babs{\sin\!\pare{\frac{d\Omega}{4}}\,
\sin\!\pare{\frac{d\Omega}{4}+\babs{\theta}_{\min}}}
< \varepsilon,
\end{align*}
and therefore
\begin{align}
    \sin\pare{\frac{d\Omega}{4}}<\sqrt{2}\varepsilon=\sqrt{2}\pare{2+\frac{2.5}{\beta}}\frac{\sigma}{m}.
\label{equ:num_UB_largetheta_eq3}
\end{align}
Moreover, since $d\omega^\ast\leqs\pi$, the condition $\frac{d\omega^\ast}{2}+\babs{\theta}_{\min}<\pi-\babs{\theta}_{\min}$ is ensured whenever $\babs{\theta}_{\min}<\frac{\pi}{4}$. Hence, denote $C_1\pare{\frac{\sigma}{m}} =  \frac{4\pi}{3}\sqrt{1+\frac{1.25}{\beta}}\pare{\frac{\sigma}{m}}^\frac{1}{2}$ and combining \eqref{equ:num_UB_largetheta_eq2} and \eqref{equ:num_UB_largetheta_eq3} yields the following sufficient conditions 
\begin{align}
\begin{cases}
     d<\frac{3}{\Omega}\arcsin\pare{\frac{\pare{2+\frac{2.5}{\beta}}\frac{\sigma}{m}}{\sin\babs{\theta}_{\min}}}, & C_1\pare{\frac{\sigma}{m}}<\babs{\theta}_{\min}<\frac{\pi}{4};\\
     d<\frac{4}{\Omega}\arcsin\pare{\pare{2\sqrt{2}+\frac{2.5\sqrt{2}}{\beta}}\frac{\sigma}{m}}, &\frac{\pi}{4}\leqs\babs{\theta}_{\min}\leqs\frac{\pi}{2}.
\end{cases}
\label{equ:num_complex_largetheta}
\end{align}

Next, we shall prove \eqref{equ:separation_condition_num_complex_LB_large}. By translation invariance \eqref{equ:shifting_invariance} of the measurement model \eqref{equ:imagemodel2sigma_num_largetheta}, we shift the coordinate system so that $y_1=-d/2$ and $y_2=d/2$. 
In this setting, the measurement $\mathbf{Y}$ is generated by
\begin{align*}
    \mu = m e^{i\frac\theta 2}\delta_{\frac{d}{2}}+\beta m e^{-i\frac\theta 2}\delta_{-\frac{d}{2}},\quad d\in\pare{0,\frac{\pi}{\Omega}},\,\beta\geqs 1,\,m>0.
\end{align*}
Consider the one-support complex measure
\[
\widehat \mu = m\pare{\beta  e^{-i\frac{\theta}{2}}+ e ^{i\frac{\theta}{2}}}\delta_{-\frac{d}{2}}.
\]
For all $\babs \omega\leqs \Omega$, we have
\begin{align*}
    & \babs{\mathcal F[\widehat\mu](\omega)-\mathcal F[\mu](\omega)} \\
    =& m\babs{\pare{\beta  e^{-i\frac{\theta}{2}}+ e ^{i\frac{\theta}{2}}}e^{-i\frac{d\omega}{2}}  -   e ^{i\frac{ \theta}{2}} e^{i\frac{d\omega}{2}}-\beta  e ^{-i\frac{ \theta}{2}}  e^{-i\frac{d\omega}{2}}}\\
    =&m\babs{e^{i \frac{\theta-d\omega}{2}}-e^{i \frac{\theta+d\omega}{2}}}=2m\sin\pare{\frac{d\omega}{2}}.
\end{align*}
Therefore, if 
\begin{align*}
    d< \frac{2}{\Omega}\mathrm{arcsin}\pare{\frac{\sigma}{m}},
\end{align*}
then $\babs{\mathscr{F}[\widehat{\mu}](\omega)-\mathscr{F}[\mu](\omega)}< 2\sigma$. Hence, there exists a  $\sigma$-admissible measure of $\mu$ with only one support whenever \eqref{equ:separation_condition_num_complex_LB_large} holds.

\subsection{Proof of Theorem~\ref{thm:supportrecoveryupperboundcomplex_large_theta}}

First, we shall prove \eqref{equ:sepacondilocationcomplex_large}. 
Given the measurement $\mathbf{Y}$ from the discrete measure $\mu=m e^{i\theta_1} \delta_{y_1}+\beta m e^{i\theta_2} \delta_{y_2}$ with
$y_1,y_2\in B_{\frac{\pi}{2\Omega}}(0)$,
$\theta_1\,\theta_2\in \pare{-\pi,\pi}$, $\beta\geqs1$, and $m>0$. Suppose that  $\widehat{\mu}=\widehat{a}_1\,\delta_{\widehat{y}_1}+\widehat{a}_2\,\delta_{\widehat{y}_2}$ is a $\sigma$-admissible measure of\, $\mathbf{Y}$. By Definition~\ref{def:sigma_admissible}, we can write the following mismatch model
\begin{align}
\mathscr{F}[\widehat\mu](\omega)=\mathscr{F}[\mu](\omega)+\mathbf{W}_1(\omega),\;\omega\in [0,2\Omega].
\label{equ:model0_num_large_theta}
\end{align}
for some $\mathbf{W}_1$
with $\babs{\mathbf{W}_1\pare{\cdot}}<2\sigma$, $\omega\in[0,2\Omega]$. Considering the argument in the proof of \eqref{equ:sepacondilocationcomplex_small_beta1},  after reordering $\babs{\widehat{y}_1-y_1}\leqs \babs{\widehat{y}_2-y_1}$ and supposing $\babs{\widehat{y}_1-y_1}\geqs \frac{d}{2}$, we can obtain the similar relationship as (\ref{equ:basic_inequality_of_complex_location}), that for any $\omega^\ast \in \left(0,\frac{2\Omega}{3}\right]$,
\begin{equation}\label{equ:basic_ineq_complex_locs_large}
\sin^2\!\pare{\frac{d\omega^\ast}{4}}\,\babs{\sin\pare{d\omega^\ast+\babs{\theta}_{\min}}} < \frac{2\sigma}{m}.
\end{equation}
The condition $d<d_\star$, with $d_\star$ defined in \eqref{equ:d_star_locs}, is generally not enough to guarantee    \eqref{equ:basic_ineq_complex_locs_large} when $\babs{\theta}_{\min}>1.75\pi\pare{\frac{\sigma}{m}}^\frac{1}{3}$.  the lower generally not sufficiently small to be below $\frac{2\sigma}{m}$. We therefore develop an alternative bound. Since $\babs{\theta}_{\min}$ is symmetric with respect to
$\babs{\theta}=\pi/2$, we focus on the regime
\begin{align*}
    1.75\pi\pare{\frac{\sigma}{m}}^\frac{1}{3}<\babs{\theta}_{\min}\leqs\frac{\pi}{2},
\end{align*}
and distinguish two subregimes.

\medskip
\noindent
\emph{Subregime A:} 
\begin{align*}
    \max{\{ d\omega^\ast, \babs{\theta}_{\min} \} } \leqs d\omega^\ast + \babs{\theta}_{\min} \leqs \pi-\babs{\theta}_{\min}.
\end{align*}
In this range, 
\begin{align*}
    &\sin^2\!\pare{\frac{d\omega^\ast}{4}}\,\babs{\sin(d\omega^\ast+\babs{\theta}_{\min})}\\
    \geqs&\sin^2\!\pare{\frac{d\omega^\ast}{4}}\cdot \max\left\{ \sin\!\pare{ d\omega^\ast}, \sin\babs{\theta}_{\min}\right\},
\end{align*}
and hence \eqref{equ:basic_ineq_complex_locs_large} implies
\begin{align*}
    \sin\!\pare{\frac{d\omega^\ast}{4}} <\min\left\{ \pare{\frac{\frac{2\sigma}{m}}{\sin \pare{\babs{\theta}_{\min}}}}^\frac{1}{2},\pare{\frac{2\sigma}{m}}^\frac{1}{3}  \right\}.
\end{align*}
Moreover, since $1.75\pi\pare{\frac{\sigma}{m}}^\frac{1}{3}<\babs{\theta}_{\min}<\frac{\pi}{2}$, we have
\begin{align*}
    \sin \pare{\babs{\theta}_{\min}}> \frac{2}{\pi}\cdot \babs{\theta}_{\min} \geqs \frac{2}{\pi} \cdot \frac{2\pare{3\sqrt{2}-3}\pi}{\sqrt{2}}\pare{\frac{\sigma}{m}}^\frac{1}{3} > \pare{\frac{2\sigma}{m}}^\frac{1}{3}.
\end{align*}
Consequently,  \eqref{equ:basic_ineq_complex_locs_large} reduces to
\begin{align}
    \sin\!\pare{\frac{d\omega^\ast}{4}}< \pare{\frac{\frac{2\sigma}{m}}{\sin \pare{\babs{\theta}_{\min}}}}^\frac{1}{2}.
    \label{equ:UB_eq1_locs_largetheta}
\end{align}

\medskip
\noindent
\emph{Subregime B:}
\begin{align*}
    \pi-\babs{\theta}_{\min} \leqs d\omega^\ast + \babs{\theta}_{\min} \leqs \pi,
\end{align*}
which implies $ \babs{\theta}_{\min} + \frac{d\omega^\ast}{2} \geqs \frac{\pi}{2}$. Since  $\omega^\ast \in \left(0,\frac{2\Omega}{3}\right]$ and $y_1,y_2\in B_{\frac{\pi}{2\Omega}}$, we may choose $\omega^\ast=\frac{\Omega}{2}$, so that
\begin{align*}
    \frac{\pi}{2}\leqs\babs{\theta}_{\min} + d\omega^\ast=\babs{\theta}_{\min} + \frac{d\Omega}{2}\leqs \frac{\pi}{3}+\babs{\theta}_{\min}.
\end{align*}
Applying  \eqref{equ:basic_ineq_complex_locs_large} with then $\omega^\ast=\frac{\Omega}{2}$ gives
\begin{align*}
    \babs{\sin^2\pare{\frac{d\Omega}{8}}\sin\pare{\frac{\pi}{3}+\babs{\theta}_{\min}}}<\babs{\sin^2\pare{\frac{d\Omega}{8}}\,
\sin\!\pare{\frac{d\Omega}{2}+\babs{\theta}_{\min}}}
< \frac{2\sigma}{m},
\end{align*}
and thus,
\begin{align}
    \sin\pare{\frac{d\Omega}{8}}< \pare{\frac{\frac{2\sigma}{m}}{\sin\pare{\frac{\pi}{3}+\babs{\theta}_{\min}}}}^\frac{1}{2}.
    \label{equ:UB_eq2_locs_largetheta}
\end{align}
Furthermore, since $d\omega^\ast\leqs\frac{2\pi}{3}$, the condition $d\omega^\ast+\babs{\theta}_{\min}<\pi-\babs{\theta}_{\min}$ is ensured whenever $\babs{\theta}_{\min}<\frac{\pi}{6}$. Hence, combining \eqref{equ:UB_eq1_locs_largetheta} and \eqref{equ:UB_eq2_locs_largetheta} gives the following conditions
\begin{align}
    \begin{cases}
         d<\frac{6}{\Omega}\arcsin\pare{\pare{\frac{\frac{2\sigma}{m}}{\sin \pare{\babs{\theta}_{\min}}}}^\frac{1}{2}}, & 1.75\pi\pare{\frac{\sigma}{m}}^\frac{1}{3}<\babs{\theta}_{\min}<\frac{\pi}{6} ;\\
         d<\frac{8}{\Omega}\arcsin\pare{\pare{\frac{\frac{2\sigma}{m}}{\sin\pare{\babs{\theta}_{\min}+\frac{\pi}{3}}}}^\frac{1}{2}}, &  \frac{\pi}{6}\leqs\babs{\theta}_{\min} \leqs\frac{\pi}{2}.
    \end{cases}
    \label{equ:condi_locs_eq1}
\end{align}
Therefore, if (\ref{equ:condi_locs_eq1}) does not hold, then we have $|\hat y_1 - y_1|< \frac{d}{2}$ which gives $|\hat y_1 - y_2|> \frac{d}{2}$. Assuming now $|\hat y_2 - y_2| \geqs \frac{d}{2}$, we will have the same relation as (\ref{equ:basic_inequality_of_complex_location_eq2}), that for any $\omega^\ast \in \left(0,\frac{2\Omega}{3}\right]$,
\begin{equation}\label{equ:basic_ineq_complex_locs_large_eq2}
\sin^2\!\pare{\frac{d\omega^\ast}{4}}\,\babs{\sin(d\omega^\ast+\babs{\theta}_{\min})} < \frac{2\sigma}{\beta m}.
\end{equation}
By the same arguments as above, this cannot hold when (\ref{equ:condi_locs_eq1}) does not hold. Therefore, $|\hat y_2 - y_2|<\frac{d}{2}$. In summary, under separation condition (\ref{equ:sepacondilocationcomplex_large}) in the theorem, we have $\babs{\widehat y_1-y_1}<\frac{d}{2} $ and $\babs{\widehat y_2-y_2}<\frac{d}{2} $.

Next, we shall prove \eqref{equ:separation_condition_locs_complex_LB_large}. By translation invariance of the measurement model \eqref{equ:model0_num_large_theta}, we shift the coordinate system so that $y_1=-d/2$ and $y_2=d/2$. 
In this setting, the measurement $\mathbf{Y}$ is generated by
\begin{align*}
    \mu =m e^{i\frac{ \theta}{2}}\delta_{\frac{d}{2}} + \beta m e^{-i\frac{ \theta}{2}}\delta_{-\frac{d}{2}},  \quad\beta\geqs1,\quad m>0.
\end{align*}
Consider the two-support complex measure
\begin{align*}
    \widehat{\mu}= \pare{\frac{1}{2}e^{i\frac{\theta}{2}}+\beta e^{-i\frac{\theta}{2}}}m \delta_{-\frac{d}{2}} + \frac{1}{2} m e^{i\frac{ \theta}{2}}\delta_{\frac{3d}{2}}.
\end{align*}
For $\babs \omega\leqs\Omega$, we have
\begin{align*}
    & \babs{\mathcal F[\widehat\mu](\omega)-\mathcal F[\mu](\omega)} \\
    =&m\babs{ \pare{\frac{1}{2}e^{i\frac{\theta}{2}}+\beta e^{-i\frac{\theta}{2}}}e^{-i\frac{d\omega}{2}}+\frac{1}{2} e^{i\pare{\frac{\theta}{2}+\frac{3d\omega}{2}}} - e^{i\pare{\frac{\theta}{2} +  \frac{d\omega}{2}}}-\beta  e^{-i \pare{\frac{\theta}{2}+\frac{d\omega}{2}}}} \\
    =&m\babs{\frac{1}{2}e^{-i\frac{d\omega}{2}}+\frac{1}{2}e^{i\frac{3d\omega}{2}}-e^{i\frac{d\omega}{2}}} = m\babs{\frac{1}{2}\pare{e^{-i d\omega}+e^{i d\omega}} -1}\\
    =&2m\sin^2 \pare{\frac{d\omega}{2}}.
\end{align*}
Therefore, if 
\begin{align*}
    d<\frac{2}{\Omega}\arcsin \pare{\pare{\frac{\sigma}{m}}^\frac{1}{2}},
\end{align*}
then $\babs{\mathcal F[\widehat\mu](\omega)-\mathcal F[\mu](\omega)}<2\sigma$, hence there exists a $\sigma$-admissible measure of\, $\mathbf{Y}$ with two supports that does not lie in the $\frac{d}{2}$-neighborhood of $\mu$.

}   %appendix的大括号

\bibliographystyle{IEEEtran}
\bibliography{reference_final}  %参考文献都放在这个包里

@article{STED,
title={Video-rate far-field optical nanoscopy dissects synaptic vesicle movementVideo-rate far-field optical nanoscopy dissects synaptic vesicle movement},
author={Westphalsilvio, V. and Rizzolimarcel, O. and Lauterbachdirk, A. and Kaminereinhard, J. and Hell, S.W.},
journal={Science},
	volume={320},
pages={246--249},
year={2008},
}

@article{PALM,
title={Imaging intracellular fluorescent proteins at nanometer resolution},
author={Betzig, E. and Patterson, G.H. and Sougrat, R. and  Lindwasser, O.W. and Olenych, S. and Bonifacino, J.S. and Davidson, M.W. and
Lippincott-Schwartz, J. and Hess, H.F.},
journal={Science},
	volume={313},
pages={1642--1645},
year={2006},
}

@article{STORM,
title={Sub-diffraction-limit imaging by stochastic optical reconstruction
microscopy (STORM)},
author={Rust, M.J. and Bates, M. and Zhuang, X.},
journal={ Nat. Methods},
	volume={3},
number={10},
pages={793--796},
year={2006},
}

@article{hell1994breaking,
	title={Breaking the diffraction resolution limit by stimulated emission: stimulated-emission-depletion fluorescence microscopy},
	author={Hell, Stefan W and Wichmann, Jan},
	journal={Optics letters},
	volume={19},
	number={11},
	pages={780--782},
	year={1994},
	publisher={Optical Society of America}
}

@article{hess2006ultra,
	title={Ultra-high resolution imaging by fluorescence photoactivation localization microscopy},
	author={Hess, Samuel T and Girirajan, Thanu PK and Mason, Michael D},
	journal={Biophysical journal},
	volume={91},
	number={11},
	pages={4258--4272},
	year={2006},
	publisher={Elsevier}
}

@article{batenkov2020conditioning,
	title={Conditioning of partial nonuniform Fourier matrices with clustered nodes},
	author={Batenkov, Dmitry and Demanet, Laurent and Goldman, Gil and Yomdin, Yosef},
	journal={SIAM Journal on Matrix Analysis and Applications},
	volume={41},
	number={1},
	pages={199--220},
	year={2020},
	publisher={SIAM}
}

@article{denoyelle2017support,
	title={Support recovery for sparse super-resolution of positive measures},
	author={Denoyelle, Quentin and Duval, Vincent and Peyr{\'e}, Gabriel},
	journal={Journal of Fourier Analysis and Applications},
	volume={23},
	number={5},
	pages={1153--1194},
	year={2017},
	publisher={Springer}
}

@article{duval2015exact,
	title={Exact support recovery for sparse spikes deconvolution},
	author={Duval, Vincent and Peyr{\'e}, Gabriel},
	journal={Foundations of Computational Mathematics},
	volume={15},
	number={5},
	pages={1315--1355},
	year={2015},
	publisher={Springer}
}

@article{liu2021mathematicalhighd,
  title={A mathematical theory of computational resolution limit in multi-dimensional spaces},
  author={Liu, Ping and Zhang, Hai},
  journal={Inverse Problems},
  volume={37},
  number={10},
  pages={104001},
  year={2021},
  publisher={IOP Publishing}
}

@article{liu2021mathematicaloned,
	title={A mathematical theory of computational resolution limit in one dimension},
	author={Liu, Ping and Zhang, Hai},
	journal={Applied and Computational Harmonic Analysis},
	volume={56},
	pages={402--446},
	year={2022},
	publisher={Elsevier}
}

@article{liu2021theorylse,
	title={A Theory of Computational Resolution Limit for Line Spectral Estimation},
	author={Liu, Ping and Zhang, Hai},
	journal={IEEE Transactions on Information Theory},
	volume={67},
	number={7},
	pages={4812--4827},
	year={2021},
	publisher={IEEE}
}

@article{morgenshtern2020super,
	title={Super-Resolution of Positive Sources on an Arbitrarily Fine Grid},
	author={Morgenshtern, Veniamin I},
	journal={ J. Fourier Anal. Appl.},
	volume={ 28},
	number={1},
	pages={Paper No. 4.},
	year={2021},
}

@article{poon2019,
	author = {Poon, Clarice. and Peyr{\'e}, Gabriel.},
	title = {MultiDimensional Sparse Super-Resolution},
	journal = {SIAM Journal on Mathematical Analysis},
	volume = {51},
	number = {1},
	pages = {1-44},
	year = {2019},	
}

@article{tang2014near,
	title={Near minimax line spectral estimation},
	author={Tang, Gongguo and Bhaskar, Badri Narayan and Recht, Benjamin},
	journal={IEEE Transactions on Information Theory},
	volume={61},
	number={1},
	pages={499--512},
	year={2014},
	publisher={IEEE}
}

@article{liu2023improved,
  title={Improved resolution estimate for the two-dimensional super-resolution and a new algorithm for direction of arrival estimation with uniform rectangular array},
  author={Liu, Ping and Ammari, Habib},
  journal={Foundations of Computational Mathematics},
  pages={1--50},
  year={2023},
  publisher={Springer}
}

@article{liu2022rslpositive,
	title={A mathematical theory of resolution limits for super-resolution of positive sources},
	author={Liu, Ping and He, Yanchen and Ammari, Habib},
	journal={arXiv preprint arXiv:2211.13541},
	year={2022}
}

@inproceedings{liu2024mathematical,
  title={A mathematical theory of super-resolution and two-point resolution},
  author={Liu, Ping and Ammari, Habib},
  booktitle={Forum of Mathematics, Sigma},
  volume={12},
  pages={e83},
  year={2024},
  organization={Cambridge University Press}
}

@article{liu2021theory,
  title={A theory of computational resolution limit for line spectral estimation},
  author={Liu, Ping and Zhang, Hai},
  journal={IEEE Transactions on Information Theory},
  volume={67},
  number={7},
  pages={4812--4827},
  year={2021},
  publisher={IEEE}
}

@article{candes2014towards,
  title={Towards a mathematical theory of super-resolution},
  author={Cand{\`e}s, Emmanuel J and Fernandez-Granda, Carlos},
  journal={Communications on pure and applied Mathematics},
  volume={67},
  number={6},
  pages={906--956},
  year={2014},
  publisher={Wiley Online Library}
}

@article{donoho1992superresolution,
  title={Superresolution via sparsity constraints},
  author={Donoho, David L},
  journal={SIAM journal on mathematical analysis},
  volume={23},
  number={5},
  pages={1309--1331},
  year={1992},
  publisher={SIAM}
}

@article{demanet2015recoverability,
  title={The recoverability limit for superresolution via sparsity},
  author={Demanet, Laurent and Nguyen, Nam},
  journal={arXiv preprint arXiv:1502.01385},
  year={2015}
}

@article{li2021stable,
  title={Stable super-resolution limit and smallest singular value of restricted Fourier matrices},
  author={Li, Weilin and Liao, Wenjing},
  journal={Applied and Computational Harmonic Analysis},
  volume={51},
  pages={118--156},
  year={2021},
  publisher={Elsevier}
}

@inproceedings{tang2015resolution,
  title={Resolution limits for atomic decompositions via Markov-Bernstein type inequalities},
  author={Tang, Gongguo},
  booktitle={2015 International Conference on Sampling Theory and Applications (SampTA)},
  pages={548--552},
  year={2015},
  organization={IEEE}
}

@article{li2022stability,
  title={Stability and super-resolution of MUSIC and ESPRIT for multi-snapshot spectral estimation},
  author={Li, Weilin and Zhu, Zengying and Gao, Weiguo and Liao, Wenjing},
  journal={IEEE Transactions on Signal Processing},
  volume={70},
  pages={4555--4570},
  year={2022},
  publisher={IEEE}
}

@article{fernandez2016super,
  title={Super-resolution of point sources via convex programming},
  author={Fernandez-Granda, Carlos},
  journal={Information and Inference: A Journal of the IMA},
  volume={5},
  number={3},
  pages={251--303},
  year={2016},
  publisher={Oxford University Press}
}

@inproceedings{ding2024esprit,
  title={The ESPRIT algorithm under high noise: Optimal error scaling and noisy super-resolution},
  author={Ding, Zhiyan and Epperly, Ethan N and Lin, Lin and Zhang, Ruizhe},
  booktitle={2024 IEEE 65th Annual Symposium on Foundations of Computer Science (FOCS)},
  pages={2344--2366},
  year={2024},
  organization={IEEE}
}

@article{batenkov2021super,
  title={Super-resolution of near-colliding point sources},
  author={Batenkov, Dmitry and Goldman, Gil and Yomdin, Yosef},
  journal={Information and Inference: A Journal of the IMA},
  volume={10},
  number={2},
  pages={515--572},
  year={2021},
  publisher={Oxford University Press}
}

@article{duval2020characterization,
  title={A characterization of the non-degenerate source condition in super-resolution},
  author={Duval, Vincent},
  journal={Information and Inference: A Journal of the IMA},
  volume={9},
  number={1},
  pages={235--269},
  year={2020},
  publisher={Oxford University Press}
}

@article{hockmann2024analysis,
  title={Analysis of the sparse super resolution limit using the Cram{\'e}r-Rao lower bound},
  author={Hockmann, Mathias},
  journal={IEEE Transactions on Information Theory},
  volume={71},
  number={1},
  pages={390--395},
  year={2024},
  publisher={IEEE}
}

@article{da2020stable,
  title={On the stable resolution limit of total variation regularization for spike deconvolution},
  author={Da Costa, Maxime Ferreira and Chi, Yuejie},
  journal={IEEE Transactions on Information Theory},
  volume={66},
  number={11},
  pages={7237--7252},
  year={2020},
  publisher={IEEE}
}

@book{bertsekas2009convex,
  title={Convex optimization theory},
  author={Bertsekas, Dimitri},
  volume={1},
  year={2009},
  publisher={Athena Scientific}
}

@article{bubeck2015convex,
  title={Convex optimization: Algorithms and complexity},
  author={Bubeck, S{\'e}bastien},
  journal={Foundations and trends in Machine Learning},
  volume={8},
  number={3-4},
  pages={231--357},
  year={2015},
  publisher={Emerald Publishing limited}
}

@article{li2023convex,
  title={Convex and non-convex optimization under generalized smoothness},
  author={Li, Haochuan and Qian, Jian and Tian, Yi and Rakhlin, Alexander and Jadbabaie, Ali},
  journal={Advances in Neural Information Processing Systems},
  volume={36},
  pages={40238--40271},
  year={2023}
}

@book{katayama2005subspace,
  title={Subspace methods for system identification},
  author={Katayama, Tohru},
  year={2005},
  publisher={Springer}
}

@article{lee2012subspace,
  title={Subspace methods for joint sparse recovery},
  author={Lee, Kiryung and Bresler, Yoram and Junge, Marius},
  journal={IEEE Transactions on Information Theory},
  volume={58},
  number={6},
  pages={3613--3641},
  year={2012},
  publisher={IEEE}
}

@article{liao2016iterative,
  title={Iterative methods for subspace and DOA estimation in nonuniform noise},
  author={Liao, Bin and Chan, Shing-Chow and Huang, Lei and Guo, Chongtao},
  journal={IEEE Transactions on Signal Processing},
  volume={64},
  number={12},
  pages={3008--3020},
  year={2016},
  publisher={IEEE}
}

@article{agarwal2016multiple,
  title={Multiple signal classification algorithm for super-resolution fluorescence microscopy},
  author={Agarwal, Krishna and Mach{\'a}{\v{n}}, Radek},
  journal={Nature communications},
  volume={7},
  number={1},
  pages={13752},
  year={2016},
  publisher={Nature Publishing Group UK London}
}

@article{roy2002esprit,
  title={ESPRIT-estimation of signal parameters via rotational invariance techniques},
  author={Roy, Richard and Kailath, Thomas},
  journal={IEEE Transactions on acoustics, speech, and signal processing},
  volume={37},
  number={7},
  pages={984--995},
  year={2002},
  publisher={IEEE}
}

@article{zhao2018super,
  title={Super-resolution of 3-D GPR signals to estimate thin asphalt overlay thickness using the XCMP method},
  author={Zhao, Shan and Al-Qadi, Imad L},
  journal={IEEE Transactions on Geoscience and Remote sensing},
  volume={57},
  number={2},
  pages={893--901},
  year={2018},
  publisher={IEEE}
}

@article{mangilli2022new,
  title={A new approach for the estimation of lake ice thickness from conventional radar altimetry},
  author={Mangilli, Anna and Thibaut, Pierre and Duguay, Claude R and Murfitt, Justin},
  journal={IEEE Transactions on Geoscience and Remote Sensing},
  volume={60},
  pages={1--15},
  year={2022},
  publisher={IEEE}
}

@article{yakovlev2015non,
  title={Non-destructive evaluation of polymer composite materials at the manufacturing stage using terahertz pulsed spectroscopy},
  author={Yakovlev, Egor V and Zaytsev, Kirill I and Dolganova, Irina N and Yurchenko, Stanislav O},
  journal={IEEE Transactions on Terahertz science and Technology},
  volume={5},
  number={5},
  pages={810--816},
  year={2015},
  publisher={IEEE}
}

@article{liu2017height,
  title={Height measurement of low-angle target using MIMO radar under multipath interference},
  author={Liu, Yuan and Jiu, Bo and Xia, Xiang-Gen and Liu, Hongwei and Zhang, Lei},
  journal={IEEE Transactions on Aerospace and Electronic Systems},
  volume={54},
  number={2},
  pages={808--818},
  year={2017},
  publisher={IEEE}
}

@article{liu2025signal,
  title={Signal-level fusion-based height estimation of low-elevation target for distributed radar},
  author={Liu, Qi and Guo, Rui and Wang, Jiajia and Xu, Shiyou and Chen, Zengping},
  journal={IEEE Transactions on Instrumentation and Measurement},
  year={2025},
  publisher={IEEE}
}

@article{wang2020efficient,
  title={Efficient and unambiguous two-target resolution via subarray-based four-channel monopulse},
  author={Wang, Shengbin Luo and Xu, Zhen-Hai and Yang, Xiao and Li, Zhongren and Wang, Guoyu},
  journal={IEEE Transactions on Signal Processing},
  volume={68},
  pages={885--900},
  year={2020},
  publisher={IEEE}
}

@article{gini2003layover,
  title={Layover solution in multibaseline SAR interferometry},
  author={Gini, Fulvio and Lombardini, Fabrizio and Montanari, Monica},
  journal={IEEE Transactions on Aerospace and Electronic Systems},
  volume={38},
  number={4},
  pages={1344--1356},
  year={2003},
  publisher={IEEE}
}

@article{martorella20143d,
  title={3D interferometric ISAR imaging of noncooperative targets},
  author={Martorella, Marco and Stagliano, Daniele and Salvetti, Federica and Battisti, Nicola},
  journal={IEEE Transactions on Aerospace and Electronic Systems},
  volume={50},
  number={4},
  pages={3102--3114},
  year={2014},
  publisher={IEEE}
}

@article{lombardini2003reflectivity,
  title={Reflectivity estimation for multibaseline interferometric radar imaging of layover extended sources},
  author={Lombardini, Fabrizio and Montanari, Monica and Gini, Fulvio},
  journal={IEEE Transactions on Signal Processing},
  volume={51},
  number={6},
  pages={1508--1519},
  year={2003},
  publisher={IEEE}
}

@article{liu2024multi,
  title={Multi-View Data-Based Layover Information Compensation Method for SAR Image Mosaic},
  author={Liu, Rui and Wang, Feng and Jiao, Niangang and You, Hongjian and Hu, Yuxin and Zhou, Guangyao and Chen, Yao},
  journal={Remote Sensing},
  volume={16},
  number={3},
  pages={564},
  year={2024},
  publisher={MDPI}
}

@article{huang2023joint,
  title={Joint estimation of unresolved leader--follower in the presence of dense false signals using monopulse radar},
  author={Huang, Qianlan and Fan, Hongqi and Cai, Fei and Xiao, Huaitie},
  journal={IEEE Transactions on Aerospace and Electronic Systems},
  volume={59},
  number={6},
  pages={9635--9649},
  year={2023},
  publisher={IEEE}
}

@article{lee2015unambiguous,
  title={Unambiguous angle estimation of unresolved targets in monopulse radar},
  author={Lee, Seung-phil and Cho, Byung-Lae and Lee, Sang-min and Kim, Ji-eun and Kim, Young-soo},
  journal={IEEE Transactions on Aerospace and Electronic Systems},
  volume={51},
  number={2},
  pages={1170--1177},
  year={2015},
  publisher={IEEE}
}

@article{blair2002unresolved,
  title={Unresolved Rayleigh target detection using monopulse measurements},
  author={Blair, William D and Brandt-Pearce, Maite},
  journal={IEEE Transactions on Aerospace and Electronic Systems},
  volume={34},
  number={2},
  pages={543--552},
  year={2002},
  publisher={IEEE}
}

@article{angle2021multiple,
  title={Multiple target tracking with unresolved measurements},
  author={Angle, R Blair and Streit, Roy L and Efe, Murat},
  journal={IEEE Signal Processing Letters},
  volume={28},
  pages={319--323},
  year={2021},
  publisher={IEEE}
}

@article{nandakumaran2008joint,
  title={Joint detection and tracking of unresolved targets with monopulse radar},
  author={Nandakumaran, N and Sinha, A and Kirubarajan, T},
  journal={IEEE Transactions on Aerospace and Electronic Systems},
  volume={44},
  number={4},
  pages={1326--1341},
  year={2008},
  publisher={IEEE}
}

@article{he2017high,
  title={High-resolution imaging and 3-D reconstruction of precession targets by exploiting sparse apertures},
  author={He, Xingyu and Tong, Ningning and Hu, Xiaowei},
  journal={IEEE Transactions on Aerospace and Electronic Systems},
  volume={53},
  number={3},
  pages={1212--1220},
  year={2017},
  publisher={IEEE}
}

@article{lan2017distributed,
  title={Distributed ECM algorithm for OTHR multipath target tracking with unknown ionospheric heights},
  author={Lan, Hua and Liang, Yan and Wang, Zengfu and Yang, Feng and Pan, Quan},
  journal={IEEE Journal of Selected Topics in Signal Processing},
  volume={12},
  number={1},
  pages={61--75},
  year={2017},
  publisher={IEEE}
}

@article{wang2021layover,
  title={Layover compensation method for regional spaceborne SAR imagery without GCPs},
  author={Wang, Huabin and Cheng, Qian and Wang, Taoyang and Zhang, Guo and Wang, Yunming and Li, Xin and Jiang, Boyang},
  journal={IEEE Transactions on Geoscience and Remote Sensing},
  volume={59},
  number={10},
  pages={8367--8381},
  year={2021},
  publisher={IEEE}
}

@article{rogel2021time,
  title={Time of arrival and angle of arrival estimation algorithm in dense multipath},
  author={Rogel, Nuriel and Raphaeli, Dan and Bialer, Oded},
  journal={IEEE Transactions on Signal Processing},
  volume={69},
  pages={5907--5919},
  year={2021},
  publisher={IEEE}
}

@inproceedings{kase2018doa,
  title={DOA estimation of two targets with deep learning},
  author={Kase, Yuya and Nishimura, Toshihiko and Ohgane, Takeo and Ogawa, Yasutaka and Kitayama, Daisuke and Kishiyama, Yoshihisa},
  booktitle={2018 15th Workshop on Positioning, Navigation and Communications (WPNC)},
  pages={1--5},
  year={2018},
  organization={IEEE}
}

@article{haardt2002unitary,
  title={Unitary ESPRIT: How to obtain increased estimation accuracy with a reduced computational burden},
  author={Haardt, Martin and Nossek, Josef A},
  journal={IEEE transactions on signal processing},
  volume={43},
  number={5},
  pages={1232--1242},
  year={2002},
  publisher={IEEE}
}

@incollection{pan2002maximum,
  title={Maximum likelihood estimation},
  author={Pan, Jian-Xin and Fang, Kai-Tai},
  booktitle={Growth curve models and statistical diagnostics},
  pages={77--158},
  year={2002},
  publisher={Springer}
}

@article{mestre2020resolution,
  title={On the resolution probability of conditional and unconditional maximum likelihood DoA estimation},
  author={Mestre, Xavier and Vallet, Pascal},
  journal={IEEE Transactions on Signal Processing},
  volume={68},
  pages={4656--4671},
  year={2020},
  publisher={IEEE}
}

@article{schenck2020probability,
  title={Probability of resolution of partially relaxed deterministic maximum likelihood: An asymptotic approach},
  author={Schenck, David and Mestre, Xavier and Pesavento, Marius},
  journal={IEEE Transactions on Signal Processing},
  volume={69},
  pages={852--866},
  year={2020},
  publisher={IEEE}
}

\begin{IEEEbiography}[{\includegraphics[width=1in,height=1.25in,clip,keepaspectratio]{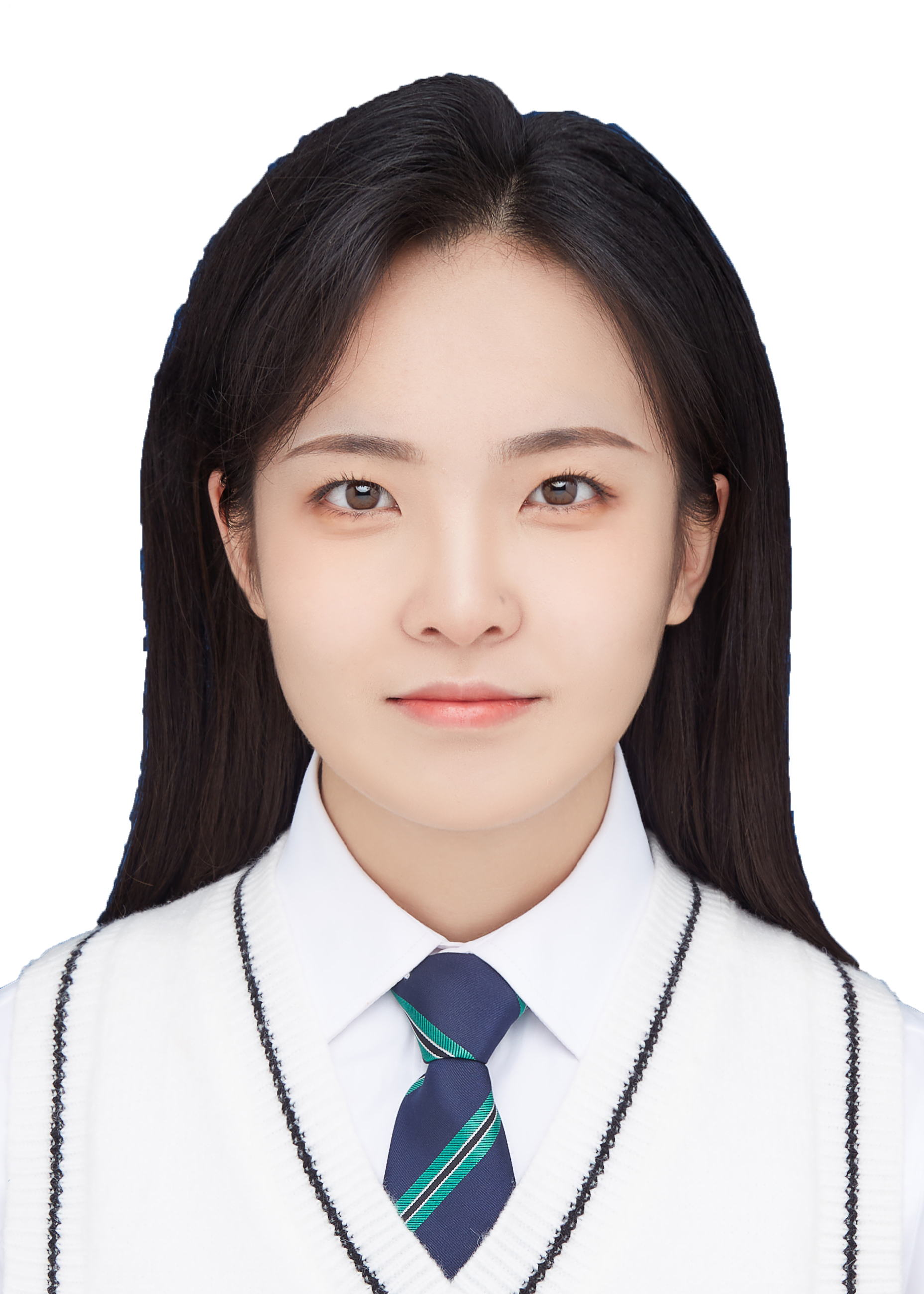}}]{Xiaole He}
was born in Shuozhou, Shanxi, China, in 1999. She received the B.S. degree from  Xidian University, China, in 2021. 

She is currently pursuing the Ph.D. degree at Beijing Institute of Technology, Beijing, China. Her research interests include ISAR imaging, sparse imaging, and super-resolution imaging.
\end{IEEEbiography}

\begin{IEEEbiography}[{\includegraphics[width=1in,height=1.25in,clip,keepaspectratio]{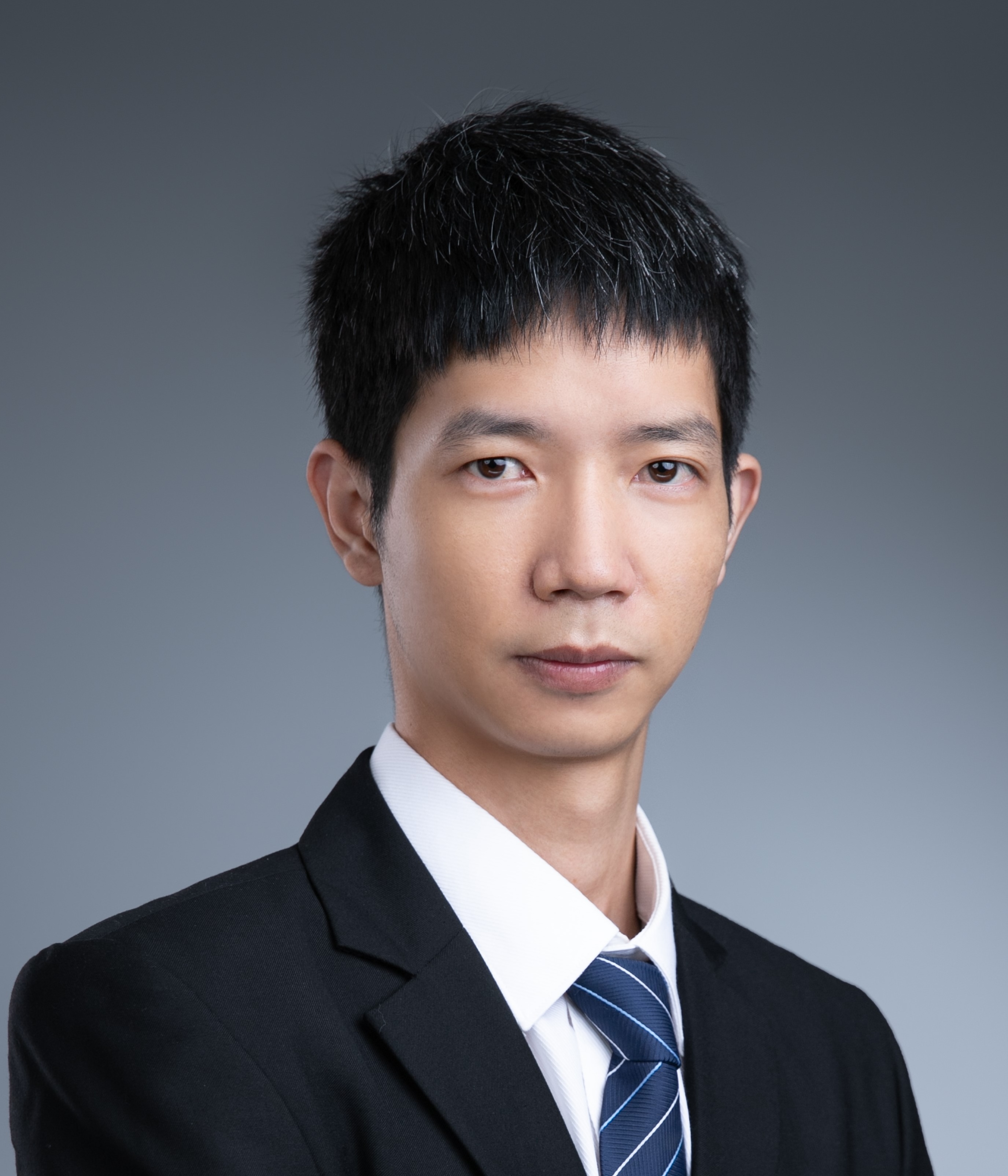}}]{Ping Liu}
received the B.S. degree in mathematics from Wuhan University in 2016, and the Ph.D. degree in mathematics from the Hong Kong University of Science and Technology in 2021.

He was a postdoc with the Department of Mathematics, ETH Z{\"u}rich, from 2021 to 2024. Since 2024, he has been a researcher jointly with the School of Mathematical Sciences and the Institute of Fundamental and Transdisciplinary Research, Zhejiang University, Hangzhou, China.
His current research interests include super-resolution imaging, array signal processing, inverse problems, and topological phononics.

\end{IEEEbiography}

\begin{IEEEbiography}[{\includegraphics[width=1in,height=1.25in,clip,keepaspectratio]{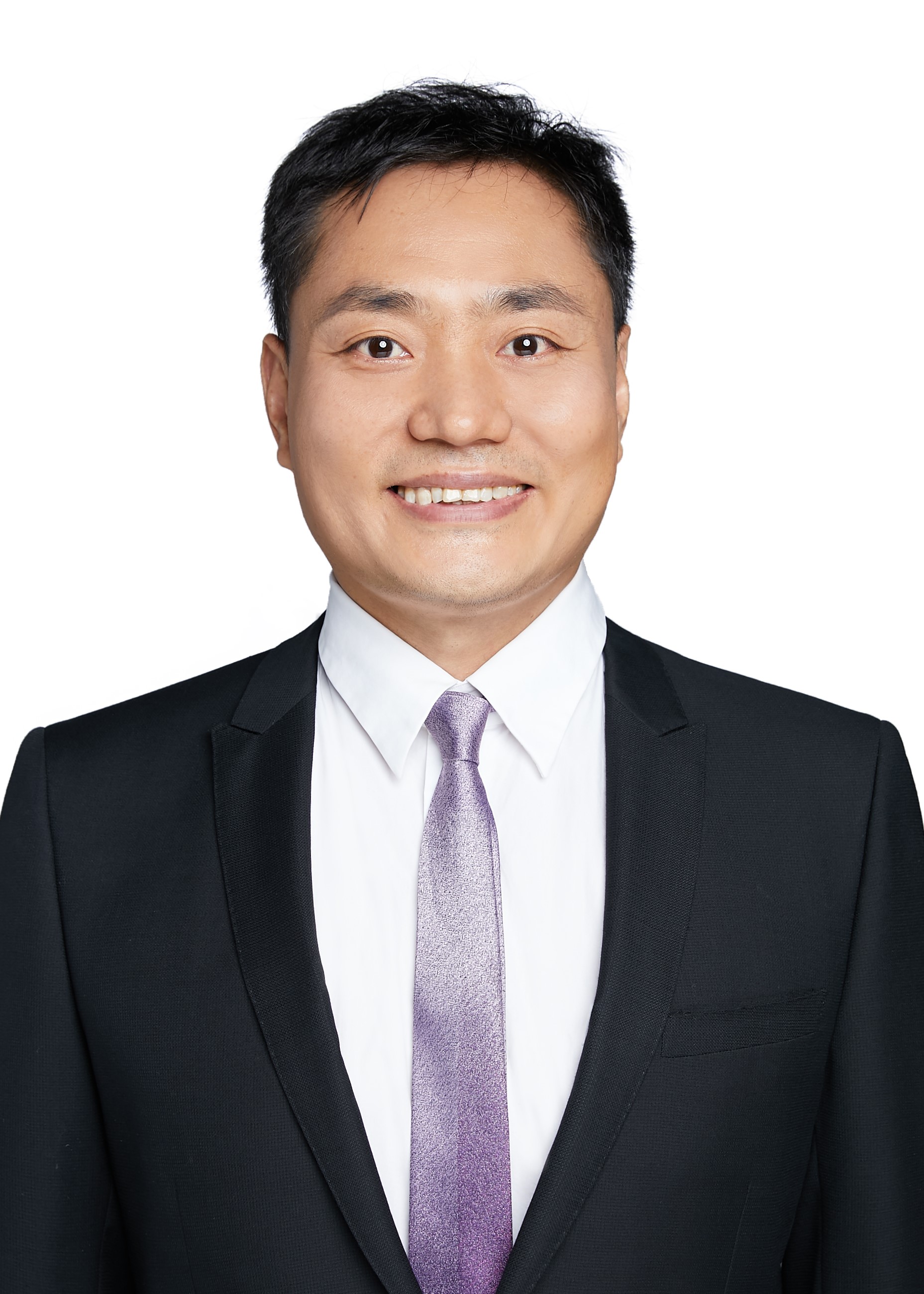}}]{Junling Wang}
(Member, IEEE) received the B.S. and M.E. degrees from China University of Petroleum, Qingdao, China, in 2005 and 2008, respectively, and the Ph.D. degree from Beijing Institute of Technology (BIT), Beijing, China, in 2013.

He was an exchange student with the Department of Signal Theory and Communications, Universitat Polit{\`e}cnica de Catalunya, Barcelona, Spain, in 2010. Since 2013, he has been with the School of Information and Electronics, BIT, Beijing, China, where he is currently an Associate Professor. His current research interests include satellite detection and imaging, and radar signal processing.
\end{IEEEbiography}

\vfill

\end{document}